\pdfoutput=1 

\documentclass[options]{asl}

\usepackage{ellipsis} 
\usepackage[abbreviations,british]{foreign} 

\usepackage{xcolor}
\usepackage{hyperref}

\definecolor{darkmidnightblue}{rgb}{0.0, 0.2, 0.4}
\definecolor{persianplum}{rgb}{0.44, 0.11, 0.11}

\hypersetup{
	colorlinks  = true, 
	citecolor   = persianplum, 
	urlcolor    = persianplum, 
	linkcolor   = persianplum
}

\usepackage{tikz}
\usepackage{todonotes}
\usepackage{xspace}

\usepackage{caption}
\usepackage{float}
\usepackage{subcaption}

\usepackage[h]{esvect}

\usepackage{graphicx}

\makeatletter
\def\desclabel#1#2{\begingroup
	\def\@currentlabel{#1}%
	#1\label{#2}\endgroup
}
\makeatother

\usetikzlibrary{shadows}
\tikzset{
	diagonal fill/.style 2 args={fill=#2, path picture={
			\fill[#1, sharp corners] (path picture bounding box.south west) -|
			(path picture bounding box.north east) -- cycle;}},
	reversed diagonal fill/.style 2 args={fill=#2, path picture={
			\fill[#1, sharp corners] (path picture bounding box.north west) |- 
			(path picture bounding box.south east) -- cycle;}}
}

\usetikzlibrary{hobby,backgrounds,calc,trees}

\usepackage{tikz-cd}
\usetikzlibrary{matrix}

\usepackage{stackengine}
\usepackage{xspace}
\usepackage{xcolor}
\usepackage{multirow}
\usepackage{longtable}
\usepackage[linesnumbered,ruled,vlined]{algorithm2e}

\SetCommentSty{mycommfont}
\SetKwInput{KwData}{Input}

\newtheorem{theorem}{Theorem}
\newtheorem{lemma}{Lemma}

\usetikzlibrary{intersections}
\usetikzlibrary{decorations}
\usetikzlibrary{snakes}
\usepackage{enumitem}
\newtheorem*{theorem*}{Theorem}
\newtheorem*{lemma*}{Lemma}
\newtheorem*{proposition*}{Proposition}
\usepackage{tikz}
\usetikzlibrary{calc,arrows,decorations.pathmorphing,decorations.markings,decorations.pathreplacing,snakes}
\tikzstyle{element}=[circle,draw=blue,fill=blue,inner sep=0pt,minimum size=3]

\tikzset{->-/.style={decoration={
			markings,
			mark=at position #1 with {\arrow{>}}},postaction={decorate}}}
\tikzset{-<-/.style={decoration={
			markings,
			mark=at position #1 with {\arrow{<}}},postaction={decorate}}}

\tikzset{-|>-/.style={decoration={
			markings,
			mark=at position #1 with {\arrow{open triangle 60}}},postaction={decorate}}}

\makeatletter
\newcommand*{\coloneqq}{\mathrel{\rlap{%
											\raisebox{0.3ex}{$\m@th\cdot$}}%
											\raisebox{-0.3ex}{$\m@th\cdot$}}%
											=}
\makeatother

\title[The Adjacent Fragment]
{The Adjacent Fragment and Quine's Limits of Decision}
\author{Bartosz Bednarczyk}
\revauthor{Bednarczyk, Bartosz}
\address{Computational Logic Group, Technische Universität Dresden, Germany\\
Institute of Computer Science, University of Wrocław, Poland}
\email{bartosz.bednarczyk@cs.uni.wroc.pl}
\thanks{Bartosz Bednarczyk was supported by the ERC Consolidator Grant No. 771779 (DeciGUT).}
\author{Daumantas Kojelis}
\revauthor{Kojelis, Daumantas}
\address{Department of Computer Science, University of Manchester, UK}
\email{daumantas.kojelis@manchester.ac.uk}
\author{Ian Pratt-Hartmann}
\revauthor{Pratt-Hartmann, Ian}
\address{Department of Computer Science, University of Manchester, UK\\
Institute of Computer Science, University of Opole, Poland}
\email{ian.pratt@manchester.ac.uk}
\thanks{Ian Pratt-Hartmann was 
supported by the NCN grant 2018/31/B/ST6/03662.}

\theoremstyle{definition}

\renewcommand{\phi}{\varphi}

\newcommand{\logic}[1]{\mathsf{#1}}

\newcommand{\FO}{\ensuremath{\logic{FO}}}
\newcommand{\FOt}{\FO^{2}}
\newcommand{\FOts}{\FO^{2*}}

\newcommand{\FL}{\ensuremath{\mathcal{FL}}}          
\newcommand{\FLv}[1]{\ensuremath{\mathcal{FL}^{#1}}} 

\newcommand{\AF}{\ensuremath{\mathcal{AF}}}          
\newcommand{\AFv}[1]{\ensuremath{\mathcal{AF}^{#1}}} 

\newcommand{\GA}{\ensuremath{\mathcal{GA}}} 
\newcommand{\GF}{\ensuremath{\mathcal{GF}}} 

\newcommand{\complexityclass}[1]{\textsc{#1}} 
\newcommand{\NP}{\complexityclass{NPTime}}

\newcommand{\PSpace}{\complexityclass{PSpace}}

\newcommand{\AExpSpace}{\complexityclass{AExpSpace}}
\newcommand{\ExpTime}{\complexityclass{ExpTime}} 
\newcommand{\NExpTime}{\complexityclass{NExpTime}} 
\newcommand{\TwoExpTime}{\complexityclass{2ExpTime}} 
\newcommand{\Tower}{\complexityclass{Tower}} %

\newcommand{\set}[1]{\ensuremath{\{ #1 \}}} 
\newcommand{\sizeof}[1]{|\!|#1|\!|}         
\newcommand{\fv}{\ensuremath{\mathrm{fv}}} 
\newcommand{\acl}{\ensuremath{\mathrm{acl}}} 

\newcommand{\bA}{\ensuremath{\mathbf{A}}} 
\newcommand{\bD}{\ensuremath{\mathbf{D}}} 
\newcommand{\bx}{\ensuremath{\mathbf{x}}} 
\newcommand{\by}{\ensuremath{\mathbf{y}}} 
\newcommand{\bu}{\ensuremath{\mathbf{z}}} 
\newcommand{\bz}{\ensuremath{\mathbf{z}}} 
\newcommand{\bv}{\ensuremath{\mathbf{z}'}} 

\newcommand{\N}{\ensuremath{\mathbb{N}}} 

\newcommand{\fA}{\ensuremath{\mathfrak{A}}} 
\newcommand{\fB}{\ensuremath{\mathfrak{B}}} 
\newcommand{\fC}{\ensuremath{\mathfrak{C}}} 

\newcommand{\ft}{\ensuremath{\mathfrak{t}}} 

\newcommand{\Atp}{\ensuremath{\mbox{\rm Atp}}}  

\newcommand{\st}{\ensuremath{\mbox{\rm st}}}   
\newcommand{\St}{\ensuremath{\mbox{\rm St}}}   

\newcommand{\tl}{\ensuremath{\mbox{\rm tl}}}   

\newcommand{\atp}{\ensuremath{\mbox{\rm atp}}}  
\newcommand{\itp}{\ensuremath{\mbox{\rm itp}}}  
\newcommand{\col}{\ensuremath{\mbox{\rm col}}}  
\newcommand{\dom}{\ensuremath{\mbox{\rm dom}}}  

\newcommand{\restr}{\!\!\restriction\!\!} 
\newcommand{\blank}{\ensuremath{\llcorner \hspace{-1mm} \lrcorner}} 

\newcommand{\val}[1]{\ensuremath{\textsc{val}^{#1}}} 

\newcommand{\sheq}[2]{\ensuremath{\textsc{eq}(#1, #2)}} 

\newcommand{\str}[1]{{\mathfrak{#1}}}

\newcommand{\sig}{\mathsf{sig}} 

\begin{document}

\begin{abstract}
We introduce the \emph{adjacent fragment} $\AF$ of first-order logic, obtained by restricting the sequences of variables occurring as arguments in atomic formulas. 
The adjacent fragment generalizes (after a routine renaming) the two-variable fragment of first-order logic as well as the so-called fluted fragment.
We show that the adjacent fragment has the finite model property, and that the satisfiability problem for its $k$-variable sub-fragment is
in $(k{-}1)$-\NExpTime. Using known results on the fluted fragment, it follows that the satisfiability problem for the whole adjacent fragment
is $\Tower$-complete. 
We~additionally consider the effect of the adjacency requirement on the well-known guarded fragment of first-order logic, whose satisfiability 
problem is \TwoExpTime-complete. We show that the satisfiability problem for the intersection of the adjacent and guarded adjacent fragments remains $\TwoExpTime$-hard. 
Finally, we show that any relaxation of the adjacency condition on the allowed order of variables in
argument sequences yields a logic whose satisfiability and finite satisfiability problems are undecidable.
\end{abstract}

 \maketitle

\section{Introduction}
\label{sec:intro}

The quest to find fragments of first-order logic for which satisfiability is algorithmically decidable has been a central undertaking of mathematical logic since  the 
appearance of D.~Hilbert and W.~Ackermann's {\em Grundz\"{u}ge der theoretischen Logik}~\cite{hilbert, book:ha50} almost a century ago. The best-known 
of these fragments belong to just three families:
(i) \textit{quantifier prefix} fragments~\cite{BorgerGG1997}, where we are restricted to formulas in prenex form with a specified quantifier sequence; 
(ii) \textit{two-variable} logics~\cite{Henkin1967}, where the only logical variables occurring as arguments of predicates are $x_1$ and $x_2$;
and (iii) \textit{guarded} logics, where either quantifiers or negated formulas are relativized by atomic formulas featuring all the free variables in their scope~\cite{ABN98,BaranyCS15}.
There is, however, a fourth family of first-order fragments for which satisfiability is decidable, but which has languished in relative obscurity. 
The fragments of this family are defined by restricting the allowed sequences of variables occurring as arguments in atomic formulas, an idea dating back 
W.~Quine's {\em homogeneous $m$-adic formulas}~\cite{quine69}. Such {\em argument-sequence fragments}, as we might call them, 
include the \emph{fluted fragment}~\cite{purdy96}, the \emph{ordered fragment}~\cite{herzig90} and the \emph{forward fragment}~\cite{Bednarczyk21}. In this paper, we identify a new argument-sequence fragment, the \emph{adjacent fragment}, which includes the fluted, ordered and forward fragments, and subsumes, in a sense we make precise, the two-variable fragment. We show that the satisfiability problem for the adjacent fragment is decidable, and obtain bounds on its complexity. Finally, we
show that the adjacent fragment is maximal among the argument-sequence fragments whose satisfiability and finite satisfiability problems are decidable.

To explain how restrictions on argument sequences work, we consider presentations of first-order logic over purely relational signatures, employing individual variables from the alphabet $\set{x_1, x_2, x_3, \dots}$. Any atomic formula in this logic has the form~$p(\bar{x})$, where~$p$ is a predicate of arity $m$ (possibly 0), and $\bar{x}$ a word of length $m$ over the alphabet of variables. Call a first-order formula $\phi$ \emph{index-normal} if any occurrence of a quantifier binding a variable $x_k$ has as its scope 
a Boolean combination of formulas that either (i) are atomic with free variables among $x_1, \dots, x_k$, or (ii) have as their major connective a quantifier binding~$x_{k+1}$. By re-indexing variables, any first-order formula can easily be written as a
logically equivalent index-normal formula. In the {\em fluted fragment}, as defined by W.~Purdy~\cite[Sec.~3]{purdy96}, we confine attention to index-normal formulas, but additionally insist that any atom occurring in a context in which~$x_k$ is available for quantification have the form $p(x_{k - m +1} \cdots x_k)$, i.e.~$p(\bar{x})$ with $\bar{x}$ being a \emph{suffix} of
$x_1 \cdots x_k$. 
In~the \emph{ordered fragment}, due to A.~Herzig~\cite[Sec.~2]{herzig90}, by contrast,
we insist that $\bar{x}$ be a \emph{prefix} of
$x_1 \cdots x_k$.  In the \emph{forward fragment}, due to B.~Bednarczyk~\cite[Sec.~3.1]{Bednarczyk21}, we insist only that $\bar{x}$ be an \emph{infix} (i.e.~a factor) of
$x_1 \cdots x_k$. 
All these logics have the finite model property, and hence are decidable for satisfiability.

We denote the fluted fragment by \FL, and the sub-fragment of \FL{} involving at most $k$ variables (free or bound) by $\FLv{k}$. It is known that  
the satisfiability problem for $\FLv{k}$ is in $(k{-}2)$-\NExpTime{} for
all $k \geq 3$, and $\lfloor k/2\rfloor$-\NExpTime-hard for all $k \geq 2$~\cite[Thm.~3.2 \& Thm.~4.2]{phst19}.
Thus, satisfiability for the whole of \FL{} is \Tower-complete, in the system of trans-elementary complexity classes due to S.~Schmitz~\cite[Sec.~3.1.1]{schmitz16}. By contrast, the satisfiability problem for the ordered fragment is $\PSpace$-complete~\cite{herzig90}
(see also a related result by R. Jaakkola~\cite[Thm.~13]{Jaakkola21}).
On the other hand, the apparent liberalization afforded by the forward fragment yields no useful increase in expressive power, and
there is a polynomial-time, satisfiability-preserving reduction of the forward fragment to the fluted fragment~\cite[p.~182]{BednarczykJ22}.
The term ``fluted'' originates with Quine~\cite{quine76a}, and presumably invites us to imagine the atoms in formulas aligned in such a way that the variables form columns.  (However, it is unclear that Quine had in mind the fragment now generally referred to as the fluted fragment; for a brief historical discussion see,
\cite[p.~221]{PrattHartmann23}.)
Note that none of these fragments can state that a relation is reflexive or symmetric, as can be easily established using a game-theoretic argument~\cite[Sec.~3]{BednarczykJ22}.

Say that a word $\bar{x}$ over the alphabet $\set{x_1, \dots, x_k}$ is \emph{adjacent} if the indices of neighbouring letters differ by at most 1. For example,
$x_3x_2x_1x_2x_2x_2x_3x_4x_3$ is adjacent, but $x_1x_3x_2$ is not.  The \emph{adjacent fragment}
is analogous to the fluted, ordered and forward fragments, but we allow any atom $p(\bar{x})$
to occur in a context where $x_k$ is available for quantification as long as $\bar{x}$ is an adjacent word over $\set{x_1, \dots, x_k}$
(see Sec.~\ref{sec:preliminaries} for a formal definition).
As a simple example, the formula 
\begin{equation}
	\forall{x_1}\forall{x_2}\forall{x_3}\exists{x_4}\forall{x_5} \ \big(p(x_1 x_2 x_3 x_2 x_3 x_4 x_5) \to p(x_1 x_2 x_3 x_4 x_3x_4x_5)\big).
\label{eq:simpleExample}
\end{equation}
is in the adjacent fragment. (In fact, it is a validity,
as can be seen by assigning~$x_4$ the same value as $x_2$.) 
We denote the adjacent fragment by \AF, and the sub-fragment of \AF{} involving at most $k$ variables (free or bound) by $\AFv{k}$.
Evidently, \AF{} includes the fluted, ordered and 
forward fragments; the inclusion is strict, since the
formulas $\forall x_1\, r(x_1x_1)$ and 
$\forall x_1 x_2 (r(x_1x_2)\rightarrow r(x_2x_1))$, stating that $r$ is reflexive and symmetric,
respectively, are in \AF.
As we show in the sequel (Theorem~\ref{theo:fo2-and-af-over-binary-sig-are-the-same}), any formula of the two-variable fragment may be translated 
to a logically equivalent formula of $\AF$.
Hence, a number of other well-known logics can be translated naturally into the adjacent fragment, including
the system of basic multimodal propositional logic $K$ 
(under the standard translation), 
a great many description logics~\cite{dlbook}, 
and even polyadic extensions~of multimodal logic~\cite[Sec.~1.5]{GorankoO07}.

Our principal result is that \AF{} has the finite model property, and that the satisfiability problem for $\AFv{k}$ is in $(k{-}1)$-\NExpTime{} for all $k \geq 2$.
The proof follows the same basic strategy as employed for \FLv{k} in~\cite{phst19}: the (finite) satisfiability problem for $\AFv{k{+}1}$ is
reduced, with exponential blow-up, to that for~$\AFv{k}$. The
result then follows from the fact that
$\AFv{2}$ is subsumed by the two-variable fragment, $\FOt$, which
has the finite model property, and
for which satisfiability is known to be in \NExpTime~\cite{GradelKV97}. 
On the other hand, $\AFv{k}$ includes $\FLv{k}$, whence the above-mentioned lower bounds for the latter carry over: the satisfiability problem for $\AFv{k}$ is $\lfloor k/2 \rfloor$-\NExpTime-hard for all $k \geq 2$ (and \NP-hard for $k \in \{0,1\}$). We remark that, using techniques similar to those employed in~\cite{phst19},
we can in fact shave one exponential off the upper bounds for $\AFv{k}$ ($k \geq 3$) when the equality predicate is disallowed; in the interests of simplicity, we leave this as an exercise to the interested reader. We additionally consider the \textit{guarded adjacent fragment} $\GA$, defined as the intersection of the adjacent fragment, $\AF$, and the 
guarded fragment, $\GF$, due to H.~Andr\'eka, J.~van~Benthem, and I.~N\'emeti~\cite[Sec.~4.1]{ABN98}. The satisfiability 
problem for $\GF$ is \TwoExpTime-complete, as shown by by E.~Gr\"adel~\cite[Thm.~4.4]{Gradel99}.
We show in the sequel that the satisfiability problem for $\GA$ remains \TwoExpTime-hard. 
We finish with a pair of results on the expressiveness of the adjacent fragment. 
First, we show that this fragment subsumes, in a sense we make precise, the two-variable fragment, with a converse subsumption holding for
signatures featuring predicates of arity at most two. Second, we
consider liberalizations of the adjacent fragment, in which the palette of permitted variable sequences is further extended. 
We show that the fragment \AF{} is a \textit{maximal} argument-sequence fragment for which satisfiability (or finite satisfiability) is decidable.

The structure of the paper is as follows. Sec.~\ref{sec:preliminaries} defines the fragments considered in this paper and establishes the notation used throughout.
Sec.~\ref{sec:primGen} is devoted to the combinatorics of words, and presents a pair of results (Lemmas~\ref{lemma:main} and ~\ref{lemma:main2}), which 
form the basis of the following two sections. 
Sec.~\ref{section:AF-upper-bounds} establishes upper bounds on the complexity of satisfiability for the sub-fragment of $\AFv{k}$ \textit{without equality}; concentrating on the equality-free case simplifies the combinatorics, thus bringing the key technical ideas 
into sharper focus. Sec.~\ref{section:AF-upper-boundsEq} then extends these results to the logic with equality. Sec.~\ref{sec:guarded} gives the
advertised lower complexity bound for the guarded adjacent fragment. Sec.~\ref{sec:extensions-and-future-work} establishes the observations on expressive power
mentioned in the previous paragraph. The results of Secs.~\ref{section:AF-upper-bounds}
and~\ref{sec:guarded}, concerning the adjacent fragment without equality, 
were first announced in the conference paper~\cite{bkp-h23}. The present article provides full proofs, and extends our results to the full 
adjacent fragment with equality. 


\section{Preliminaries}
\label{sec:preliminaries}
Let $m$ and $k$ be non-negative integers.
For any integers $i$ and $j$, we write $[i,j]$ to denote the set of integers $h$ such that $i \leq h \leq j$.
A~function $f \colon [1,m] \rightarrow [1,k]$ is \emph{adjacent}
if $|f(i{+}1) {-} f(i)| \leq 1$ for all $i \in [1,m{-}1]$.
We write $\bA^m_k$ to denote the set of adjacent functions $f \colon [1,m] \rightarrow [1,k]$. Since $[1,0] = \emptyset$, we have $\bA^0_k = \set{\emptyset}$, and $\bA^m_0 = \emptyset$ if $m >0$.
Let $A$ be a non-empty set. Regarding $A$ as an alphabet, a word $\bar{a}$ over the alphabet $A$ is simply
a tuple of elements from $A$; we alternate freely in the sequel between these two ways of
speaking, as the context requires. Accordingly, we take $A^k$ to denote the set of words
over $A$ having length exactly $k$, and $A^*$, the set of all finite words over $A$. 
If $\bar{a} = a_1 \cdots a_k$, we write $|\bar{a}| = k$ for the length of $\bar{a}$, and $\tilde{a} = a_k \cdots a_1$ for the reversal of $\bar{a}$.
Any function
$f\colon [1,m] \rightarrow [1,k]$ (adjacent or not) induces a natural map from $A^k$ to $A^m$ defined by
$\bar{a}^f = a_{f(1)} \cdots a_{f(m)}$.
If $f \in \bA^m_k$ (i.e.~if $f$ is adjacent), we may think of $\bar{a}^f$ as the result of a `going for a stroll' on the tuple $\bar{a}$, starting at the element $a_{f(1)}$, and moving left, right,
or remaining stationary according to the sequence of values $f(i+1) {-} f(i)$ (for $1 \leq i <m$).

For any $k \geq 0$, denote by $\bx_k$ the fixed word $x_1 \cdots x_k$ (if $k=0$, this is the empty word).
A \emph{$k$-atom} is an expression $p(\bx^f_k)$, where $p$ is a predicate of
some arity~$m$, and $f\colon [1,m] \rightarrow [1,k]$. 
Thus, in a $k$-atom, each argument is a variable chosen from $\bx_k$. If $f$ is adjacent, we speak of an
\emph{adjacent $k$-atom}.
Thus,
in an adjacent $k$-atom, the indices of neighbouring
arguments differ by at most one. 
The equality predicate is allowed when $m=2$. 
Proposition letters (predicates of arity $m=0$) count as (adjacent) $k$-atoms
for all $k \geq 0$, taking $f$ to be the empty function.
When $k= 0$, we perforce have $m = 0$, since otherwise, there are no functions from $[1,m]$ to~$[1,k]$; thus the 0-atoms are
precisely the proposition letters. 
When $k \leq 2$, the adjacency requirement is vacuous, and we prefer to speak
simply of {$k$-atoms}.\label{page:Def}

We define the sets of first-order formulas $\AFv{[k]}$ by simultaneous structural induction for all $k \geq 0$:
\begin{enumerate}
	\item every adjacent $k$-atom is in $\AF^{[k]}$;
	\item $\AF^{[k]}$ is closed under Boolean combinations;
	\item if $\phi$ is in $\AFv{[k{+}1]}$, then $\exists x_{k{+}1}\, \phi$ and  
	$\forall x_{k{+}1}\, \phi$ are in $\AFv{[\ell]}$ for all $\ell \geq k$.
\end{enumerate}
Now let $\AF=\bigcup_{k \geq 0} \AF^{[k]}$ and define $\AFv{k}$ to be the set of formulas of $\AF$ featuring no variables other than $x_1, \dots, x_k$, free or bound.
We call
$\AF$ the \emph{adjacent fragment} and $\AFv{k}$ the \emph{$k$-variable adjacent fragment}.
Note that formulas of $\AF$ contain no individual constants or function symbols; however, they may contain equality.
The primary objects of interest here are the languages  $\AF$ and~$\AFv{k}$; the sets of formulas $\AF^{[k]}$ will make only occasional 
appearances in the sequel.
Thus, for example, the formula~\eqref{eq:simpleExample} is in $\AFv{k}$ if and only if $k \geq 5$, but it is in $\AFv{[k]}$ for all $k \geq 0$.
On the other hand, the \textit{quantifier-free} formulas of
$\AFv{[k]}$ and $\AFv{k}$ are the same.  
A simple structural induction establishes that $\AFv{[k]} \subseteq \AFv{[\ell]}$ for all $k \leq \ell$.


We silently assume the variables $\bx_k \coloneqq x_1 \cdots x_k$ to be ordered in the standard way. That is: if
$\phi$ is a formula of $\AF^{[k]}$, $\fA$ a structure interpreting its signature, and $\bar{a} \coloneqq a_1 \cdots a_k \in A^k$,
we say simply that $\bar{a}$ \textit{satisfies} $\phi$ in $\fA$, and write $\str{A} \models \phi[\bar{a}]$ to mean that
$\bar{a}$ satisfies~$\phi$ in $\fA$ under the assignment $x_i \leftarrow a_i$ (for all $1 \leq i \leq k)$. (This does not
necessarily mean that
each of the variables of $\bx_k$ actually appears in $\phi$.)
If $\phi$ is true under  all assignments in all structures,
we write $\models \phi$; the notation
$\phi \models \psi$ means the same as $\models \phi \rightarrow \psi$ (i.e.~variables are consistently instantiated in $\phi$ and $\psi$).
The notation $\phi(\bar{v})$, where $\bar{v} \coloneqq v_1 \cdots v_k$ are variables (chosen from among $x_1, x_2, \dots$), will always be used to denote
the formula that results from substituting $v_i$ for $x_i$ ($1 \leq i \leq k)$~in~$\phi$, rather than to indicate the order in which elements of some structure
are to be assigned to variables. 
If $\phi$ is any formula, $\fv(\phi)$ denotes the set of free variables of $\phi$.
A~{\em sentence} is a formula with no free variables. Necessarily, all formulas
of $\AFv{[0]}$ are sentences.
For a sentence $\varphi$ we write simply $\fA \models \phi$ to mean that $\phi$ is true in $\fA$.
%
We call the set of predicates used in $\varphi$ \emph{the signature of $\varphi$}, denoted~$\sig(\varphi)$.

We adapt the standard notion of (atomic) $k$-types for the fragments studied here. Fix some non-logical relational signature $\tau$ (i.e.~not
containing the equality predicate).
An \emph{adjacent $k$-literal} \emph{over} $\tau$ is an {\em adjacent $k$-atom} or its negation, featuring a predicate in $\tau \cup \set{=}$.
An \emph{adjacent $k$-type} \emph{over} $\tau$ is a maximal consistent set of adjacent $k$-literals over $\tau$.
Reference to $\tau$ is suppressed where clear from context.
We use the letters $\zeta$, $\eta$ and~$\xi$ to range over
adjacent $k$-types for various $k$. We denote by
$\Atp^\tau_k$ the set of all adjacent $k$-types over $\tau$. For finite $\tau$, we identify members of
$\Atp^\tau_k$ with their conjunctions, and treat them as (quantifier-free) $\AFv{k}$-formulas,
writing $\zeta$ instead of $\bigwedge \zeta$.
Given a pair of integers $i, j$ ($1 \leq i \leq j \leq k$), we write $\zeta \restr_{[i, j]}$ for the set $\eta$
obtained by deleting literals in~$\zeta$ that feature variables outside the range $[i, j]$.
It is evident that (after a shift in indices) $\eta$ is a $(j{-}i{+}1)$-type. 
We write $\zeta^+$ for the quantifier-free $\AFv{\ell+1}$-formula $\zeta(x_2, \dots, x_{k{+}1})$ obtained by incrementing the index of each variable. 
When $k \leq 2$, the adjacency requirement is vacuous, and in this case we shall simply speak of \emph{$k$-types}. 
Every quantifier-free $\AFv{k}$-formula $\chi$
is thus logically equivalent to a disjunction of adjacent $k$-types, as may be seen by 
writing $\chi$ in disjunctive normal form. In particular,
if $\chi$ is satisfiable, then there is an adjacent $k$-type which entails it.
If $\fA$ is a $\tau$-structure and $\bar{a}$ a $k$-tuple of elements from $A$, there
is a unique adjacent $k$-type $\zeta$ such that $\fA \models \zeta[\bar{a}]$; we denote
this adjacent $k$-type by $\atp^\fA[\bar{a}]$, and call it the {\em adjacent type of $\bar{a}$ in $\fA$}.
It is not required that the elements of $\bar{a}$ be distinct; note however that any (in)equality literals occurring in adjacent types must themselves be adjacent. For instance,
$x_5 = x_6$ or $x_5 \neq x_6$ may occur in an adjacent type, but not $x_4 = x_6$ or $x_4 \neq x_6$. 

The following derivative notions relating to adjacent types will feature in the sequel. 
Call an adjacent $k$-literal {\em covering} if it features all of the variables in $\bx_k$, i.e.~if it has the form $p(\bx_k^f)$
or $\neg p(\bx_k^f)$ with $p$ of arity $m$ and $f \in \bA^m_k$ {surjective}. Define an {\em incremental $k$-type} over $\tau$ to be a maximal consistent set
$\iota$ of covering adjacent $k$-literals over $\tau$. 
If $\xi$ is an adjacent $k$-type, then the {\em increment of} $\xi$, denoted  $\partial\xi$, is the (unique) incremental $k$-type included in $\xi$.
If $\fA$ interprets $\tau$ and $\bar{c}$ is a $k$-tuple over $\fA$, then the {\em incremental type of} $\bar{c}$ in $\fA$ is the (unique) incremental 
type $\iota$ such that $\fA \models \iota[\bar{c}]$; we write $\itp^\fA[\bar{c}]$ to denote $\iota$. 
Suppose now that $\bar{c} \coloneqq a\bar{a}b$, $\atp^\fA[a\bar{a}]= \zeta$, $\atp^\fA[\bar{a}b]= \eta$ and $\itp^\fA[\bar{c}]= \iota$. It should be
obvious that, writing $\xi$ for $\atp^\fA[\bar{c}]$, we have $\xi = \zeta \cup \eta^+ \cup \iota$, and, moreover, $\iota = \partial \xi$.

The following lemma establishes a {normal form} for \AFv{\ell+1}-sentences, which simplifies the decision procedures discussed in Secs.~\ref{section:AF-upper-bounds} and~\ref{section:AF-upper-boundsEq}. 
\begin{lemma}
Let $\phi$ be a sentence of \AFv{\ell+1}, where $\ell \geq 1$. We can compute, in polynomial time, an \AFv{\ell+1}-formula
$\psi$ satisfiable over the same domains as $\phi$, of the form
\begin{equation}
\psi \coloneqq \bigwedge_{i \in I} \forall \bx_{\ell} \exists x_{\ell+1}\, \gamma_i \wedge 
\forall \bx_{\ell+1}\,
\beta,
\label{eq:anf}
\end{equation}  
where $I$ is a finite index set, 
and the formulas $\gamma_i$ and $\beta$ are quantifier-free; moreover, if $\phi$ is equality-free, then so is $\psi$.
\label{lma:anf}
\end{lemma}
\begin{proof}
If the sentence $\phi$ is quantifier-free, then it is a formula of the propositional calculus, and the result is easily obtained by adding vacuous quantification. 
Otherwise, write $\phi_0 = \phi$, and  let $\theta\coloneqq Q x_{k{+}1}\, \chi$
be a subformula of $\phi$, where $Q \in \set{\forall, \exists}$, such that $\chi$ is quantifier-free. 
Writing $\bar{\exists}= \forall$ and $\bar{\forall} = \exists$,
let $p$ be a new predicate 
of arity $k$, let $\phi_1$ be the result of replacing $\theta$ in $\phi_0$ by the atom $p(\bx_k)$, 
and let  $\psi_1$ be the formula 
\begin{equation*}
\forall \bx_k Q x_{k+1}\big(p(\bx_{k}) \rightarrow \chi\big) \wedge
\forall \bx_k \bar{Q} x_{k+1}\big(\chi \rightarrow p(\bx_{k})\big).
\end{equation*}
It is immediate that $\phi_1 \wedge \psi_1 \models \phi_0$. Conversely, if $\fA \models \phi_0$, then we may expand $\fA$ 
to a model $\fA'$ of $\phi_1 \wedge \psi_1$ by taking $p^{\fA'}$ to be the set of $k$-tuples $\bar{a}$ such that $\fA \models \theta[\bar{a}]
$. Evidently, $\phi_1$ is a sentence of \AFv{\ell+1}. 
Processing $\phi_1$ in the same way, and proceeding similarly, we obtain a set of formulas $\phi_2, \dots, \phi_m$ and $\psi_2, \dots, \psi_m$, with  $\phi_m$ quantifier-free and
$\phi_0$ satisfiable over the same domains
as $\psi_1 \wedge \cdots \wedge \psi_m \wedge \phi_m$. Since $\phi_m$ is a sentence, it is a formula of the
propositional calculus. By moving $\phi_m$ inside one of the quantified formulas, re-indexing variables and re-ordering conjuncts, we obtain a formula
$\psi$ of the form~\eqref{eq:anf}.
\end{proof}
We refer to any \AFv{k}-sentence having the form~\eqref{eq:anf} as a \textit{normal-form} formula.

The following notation will be useful.
If $\chi$ is any quantifier-free $\AFv{\ell+1}$-formula, we denote by $\chi^{-1}$ the formula $\chi(x_{\ell+1}, \dots, x_1)$
obtained by simultaneously replacing each variable $x_h$ by $x_{\ell -h +2}$ ($1 \leq h \leq \ell+1$); and we denote by $\hat{\chi}$ the formula $\chi \wedge \chi^{-1}$. Obviously $\chi^{-1}$ and $\hat{\chi}$ are also in $\AFv{\ell+1}$.


 \section{Primitive generators of words}
\label{sec:primGen}
The upper complexity bounds obtained below depend on an observation concerning the combinatorics of words, which may be of independent interest.
For words $\bar{a}, \bar{c} \in A^*$ 
with $|\bar{a}|=k$ and $|\bar{c}|=m$, 
say that $\bar{a}$ \emph{generates} $\bar{c}$ if $\bar{c} = \bar{a}^f$ for some \textit{surjective} function $f \in \bA^m_k$. 
As explained above, it helps to think of $\bar{a}^f$ as the sequence of letters encountered on an $m$-step `stroll' backwards and forwards on the tuple $\bar{a}$, with $f(i)$ giving the index of 
our position in $\bar{a}$ at the $i$th step. The condition that 
$f$ is adjacent ensures that we never change position by more than one letter at a time; the condition that $f$ is surjective ensures that we visit every position of $\bar{a}$. 
We may depict $f$ as a piecewise linear graph,
with the \mbox{generat{\em{}ed}} word superimposed on the abscissa and the \mbox{generat{\em ing}} word on the ordinate (Fig.~\ref{fig:example}).
\begin{figure}[h]
  \begin{center}
    \begin{tikzpicture}[scale= 0.25]
      \draw[->] (0,0) --  (25,0);
      \coordinate[label={$\bar{c}$}] (wLabel) at (24,0);
      
      \draw[->-=1] (-1,1) to (-1, 10.5);
      \coordinate[label={\rotatebox{90}{$\bar{a}$}}] (wLabel) at(-0.25, 8.5) ;
      
      \coordinate[label=left:{\rotatebox{90}{c}}] (wLabel) at (-1,1);
      \coordinate[label=left:{\rotatebox{90}{b}}] (wLabel) at (-1,2);
      \coordinate[label=left:{\rotatebox{90}{a}}] (wLabel) at (-1,3);
      \coordinate[label=left:{\rotatebox{90}{d}}] (wLabel) at (-1,4);
      \coordinate[label=left:{\rotatebox{90}{e}}] (wLabel) at (-1,5);
      \coordinate[label=left:{\rotatebox{90}{f}}] (wLabel) at (-1,6);
      \coordinate[label=left:{\rotatebox{90}{b}}] (wLabel) at (-1,7);
      \coordinate[label=left:{\rotatebox{90}{a}}] (wLabel) at (-1,8);
      
      \coordinate[label={a}] (wLabel) at (0,-2);
      \coordinate[label={b}] (wLabel) at (1,-2);
      \coordinate[label={c}] (wLabel) at (2,-2);
      \coordinate[label={b}] (wLabel) at (3,-2);
      \coordinate[label={a}] (wLabel) at (4,-2);
      \coordinate[label={a}] (wLabel) at (5,-2);
      \coordinate[label={a}] (wLabel) at (6,-2);
      \coordinate[label={d}] (wLabel) at (7,-2);
      \coordinate[label={e}] (wLabel) at (8,-2);
      \coordinate[label={f}] (wLabel) at (9,-2);
      \coordinate[label={e}] (wLabel) at (10,-2);
      \coordinate[label={d}] (wLabel) at (11,-2);
      \coordinate[label={a}] (wLabel) at (12,-2);
      \coordinate[label={d}] (wLabel) at (13,-2);
      \coordinate[label={e}] (wLabel) at (14,-2);
      \coordinate[label={f}] (wLabel) at (15,-2);
      \coordinate[label={b}] (wLabel) at (16,-2);
      \coordinate[label={a}] (wLabel) at (17,-2);
      \coordinate[label={b}] (wLabel) at (18,-2);
      \coordinate[label={f}] (wLabel) at (19,-2);
      
      \draw (0,3) -- (2,1) -- (4,3) -- (6,3) -- (9,6) -- (12,3) --(17,8) -- (19,6);
    \end{tikzpicture}
  \end{center}
  \caption{Generation of {abcbaaadefedadefbabf} from {cbadefba}.}
  \label{fig:example}
\end{figure}
We refer to any maximal interval $[i,j] \subseteq [1,m]$ over which $f(h{+}1){-}f(h)$ is constant (for $i \leq h < j$) as a {\em leg} of $f$. Thus, 
the legs correspond to 
the straight-line segments in the graph of $f$. A leg is {\em increasing}, {\em flat} or {\em decreasing} according as $f(h{+}1){-}f(h)$ is
1, 0 or -1.

Every word $\bar{a}$ generates both itself and its reversal, $\tilde{a}$. Moreover, if $\bar{a}$ generates~$\bar{c}$, then
$|\bar{c}| \geq |\bar{a}|$, by the surjectivity requirement. In fact, $\bar{a}$ and $\tilde{a}$ are the only words of length $|\bar{a}|$ generated by $\bar{a}$. 
Finally, generation is transitive: if $\bar{a}$ generates $\bar{b}$ and $\bar{b}$ generates $\bar{c}$, then $\bar{a}$ generates $\bar{c}$.
We call $\bar{a}$ \textit{primitive} if it is not generated by any word shorter than itself, equivalently, if it is generated only by itself and its reversal. 
For example, $babcd$ and $abcbcd$ are not primitive, because they are generated by $abcd$; but
$abcbda$ is primitive.  
Note that factors of primitive words need not be primitive; for example, $abcbda$ is primitive, but its factor $bcb$ is not.
Define  a {\em primitive generator} of $\bar{c}$ to be
a generator of $\bar{c}$ that is itself primitive.
It follows from the foregoing remarks 
that every word $\bar{c}$ has some primitive generator $\bar{a}$, and indeed, $\tilde{a}$ as well,
since the reversal of a primitive generator is clearly a primitive generator. The~following observation, on the other hand,
is surprising. 
\begin{lemma}[Thm.~1 of \cite{ph:primGen24}]
  The primitive generator of any word is unique up to reversal.
  \label{lemma:main}
\end{lemma}
For a very similar, though not identical, result, see~\cite{ar90}.

Define the {\em primitive length} of any word $\bar{c}$ to be the length of any primitive generator of~$\bar{c}$. By Lemma~\ref{lemma:main},
this notion is well-defined; it will play a significant role in our analysis of the adjacent fragment.
Clearly, the primitive length of~$\bar{c}$ is at most $|\bar{c}|$, but will be strictly less if $\bar{c}$ is not primitive. 

It is important to realize that, while primitive generators are unique up to reversal, modes of generation are not. 
Indeed, 
$\bar{a} \coloneqq abcbd$ is one of the two
primitive generators of ${\bar{c}} \coloneqq abcbcbd$, but we have $\bar{a}^f= \bar{c}$ for $f \colon[1,7] \rightarrow [1,5]$ given by either of the courses of values $[1,2,3,4,3,4,5]$ or $[1,2,3,2,3,4,5]$.
In the sequel, it will be important to identify those words on which a given pair of surjective adjacent functions yield identical outputs.
A {\em palindrome} is a word equal to its reversal;
a palindrome is \textit{non-trivial} if its length is at least 2.
Let $\bar{a} \coloneqq a_1 \cdots a_k$ be a word of length $k$.
We say that a pair $\langle i, j \rangle$ is a \emph{defect} of $\bar{a}$ if the factor $a_i \cdots a_j$ is a non-trivial palindrome. 
We denote the set of defects of $\bar{a}$ by $D_{\bar{a}}$, and regard it as a 
a binary relation on the set $[1,k]$. If $R$ is any binary relation, we write 
$R^{*}$ for its equivalence closure, i.e.~the smallest reflexive, symmetric and transitive relation that includes $R$.
Now, for any pair of adjacent functions
$f,g \colon [1,m] \rightarrow [1,k]$ and any set of pairs
$D \subseteq \set{\langle i,j \rangle \mid 1 \leq i < j \leq k}$, we write $f \overset{D}{=} g$ if
$\langle f(i), g(i) \rangle \in D^{*}$ for all $i$ ($1 \leq i \leq m$). Evidently, $\overset{D}{=}$ is an equivalence relation. 
\begin{lemma}[Thm.~4 of \cite{ph:primGen24}]
  Let $\bar{a}$ be a primitive word of length $k$ with defect set $D$, and let $f$ and $g$ be surjective functions in $\bA^m_k$ for some
  $m \geq k$.   
  Then $\bar{a}^f = \bar{a}^g$ if and only if $f \overset{D}{=} g$.
  \label{lemma:main2}
\end{lemma}
Of course, given surjective functions $f,g \in \bA^m_k$ and any $D \subseteq \set{\langle i,j \rangle \mid 1 \leq i < j \leq k}$,
it is a simple matter to check whether $f \overset{D}{=} g$.
Lemma~\ref{lemma:main2} allows us to read this condition as stating that $f$ and $g$ yield the same tuples when applied to
any primitive word of length $k$ whose defect set includes $D$.

Any adjacent function $f\colon [1,\ell], \rightarrow [1,k]$ induces a natural map from quantifier-free $\AFv{\ell}$-formulas
to quantifier-free $\AFv{k}$-formulas. Specifically,
if $\chi$ is a\linebreak \mbox{quantifier-free} $\AFv{\ell}$-formula, denote by $\chi^g$ the formula $\chi(x_{g(1)} \cdots x_{g(\ell)})$, obtained by simultaneously
replacing every variable $x_i$ in $\chi$ by the corresponding variable $x_{g(i)}$. 
We claim that $\chi^g \in \AFv{k}$. Indeed, any
atom $\alpha$ appearing in $\chi$ is of the form $p(\bx^f_k)$, 
where $p$ is a predicate of some 
arity $m$ and $f \in \bA^m_\ell$. But then the corresponding atom in $\chi^g$ has the form
$\beta\coloneqq \alpha(x_{g(1)} \cdots x_{g(\ell)}) = p(x_{g(f(1))} \cdots x_{g(f(m))}) = p(\bx^{(g \circ f)}_k)$.
Since the composition of adjacent functions is adjacent, the assertion follows. 
The following (almost trivial) lemma is useful when manipulating adjacent formulas. 
Recall in this regard that any function $g \in \bA^\ell_k$ 
maps a $k$-tuple $\bar{a}$ over some set to an $\ell$-tuple $\bar{a}^g$ 
over the same set.

\begin{lemma}
  Let $\chi$ be a quantifier-free $\AFv{\ell}$-formula, and let $g \in \bA^\ell_k$. For~any structure $\fA$ and any $\bar{a} \in A^k$, we have $\fA \models \chi^g[\bar{a}]$ if and only if $\fA \models \chi[\bar{a}^g]$.~Thus,~the adjacent type of any tuple in $\fA$ is determined 
  by that of its primitive~generator. 
  \label{lma:essentiallyTrivial}
\end{lemma}
\begin{proof}
  We need consider only the case where $\chi$ is atomic: the general case follows by straightforward structural induction.
  Let $\chi \coloneqq p(\bx^f_k)$, with $p$ an $m$-ary predicate, and $f \in \bA^m_\ell$.
  Then, writing $\bar{a} = a_1 \cdots a_m$,  both sides of the 
	bi-conditional amount to the statement $a_{g(f(1))} \cdots a_{g(f(m))} \in p^\fA$. 
	For the second statement, 
	let $\fA$ be a structure, and 
	$\bar{a}$ an $\ell$-tuple from $A$. Then $\bar{a}$ has a primitive generator, say $\bar{b}$ of length $k \leq \ell$, with
	$\bar{a} = \bar{b}^g$ for some (surjective) $g \in \bA^\ell_k$.  Now consider any atomic \AFv{\ell}-formula $\alpha$. 
	Then
	$\fA \models \alpha[\bar{a}]$ if and only if 
	$\fA \models \alpha^g[\bar{b}]$. 
\end{proof}

Let $\fA$ and $\fA'$ be structures interpreting some common signature over a common domain~$A$, and let $\ell \geq 0$. 
We write
$\fA \approx_\ell \fA'$, if, for 
any predicate $p$ of arity $m \geq 0$, and any $m$-tuple $\bar{a}$ from $A$ of primitive length at most $\ell$, we have 
$\bar{a} \in p^\fA$ if and only if $\bar{a} \in p^{\fA'}$. That is, $\fA \approx_\ell \fA'$
just in case, for any predicate $p$ interpreted by~$\fA$, $p^\fA$ and $p^{\fA'}$ agree on all those $m$-tuples whose primitive length~is~at~most~$\ell$. 
The next lemma states that, when evaluating \AFv{\ell}-formulas in structures, we can disregard tuples whose primitive length
is greater than $\ell$. 
\begin{lemma}
Let $\phi$ be an $\AFv{\ell}$-sentence,  and suppose $\fA$ and $\fA'$
are $\sig(\phi)$-structures over a common domain $A$ such that
$\fA \approx_\ell \fA'$. 
Then $\fA \models \phi \Leftrightarrow\fA' \models \phi$.
\label{lma:approx}
\end{lemma}
\begin{proof}
  Let $\psi$ be a formula of $\AFv{\ell}$ (possibly featuring free variables), 
  and let $k$ (bounded by $\ell$) be such that $\psi \in \AFv{[k]}$. 
  We claim that, for any $k$-tuple   of elements $\bar{a}$, $\fA \models \psi[\bar{a}]$ 
  if and only if $\fA' \models \psi[\bar{a}]$. The statement of the lemma is the special case where $\psi$ has no free variables. 
  Again, we need consider only the case where $\psi$ is atomic: the general case follows by straightforward structural induction. 
  Let $\psi \coloneqq p(\bx^f_k)$, with $p$ an $m$-ary predicate, and $f \in \bA^m_k$.
  If $\bar{a}$ is a $k$-tuple of elements from $A$, then $\fA \models \psi[\bar{a}]$ 
  if and only if $\bar{a}^f \in p^\fA$, and similarly for $\fA'$. But the primitive length of $\bar{a}^f$ is certainly at most $k = |\bar{a}|$, and thus $\bar{a}^f \in p^\fA$ if and only
  if $\bar{a}^f \in p^{\fA'}$.
\end{proof}

In
view of Lemma~\ref{lma:approx}, when considering models of \AFv{\ell}-sentences, it will be useful to take the extensions of non-logical predicates (of whatever arity) to be {undefined in respect of tuples whose primitive length is greater than} $\ell$, since these 
cannot affect the outcome of semantic evaluation. That is, where $\ell$ is clear from context, we typically suppose any
model $\fA$ of $\phi$ to 
determine
whether $\bar{a} \in p^\fA$ for any $m$-ary predicate $p$ and any $m$-tuple $\bar{a}$ \textit{of primitive length at most} $\ell$; but with respect to $m$-tuples $\bar{a}$ having greater primitive length, $\fA$ remains agnostic. To make it clear that the structure
$\fA$ need not be fully defined, we refer to it as a {\em layered structure}, of {\em height} $\ell$. Notice that the notion of height is independent of the arities of the predicates interpreted. 
A layered structure $\fA$ may have height, say 3, but still interpret a predicate $p$ of arity, say, 5. In this case, it is
determined whether $\fA \models p[babcc]$, because the primitive generator of $babcc$ is $abc$; however, it is not determined
whether  $\fA \models p[abcab]$, because $abcab$~is~primitive.

This idea enables us to build up models of \AF-formulas layer by layer.
Suppose $\fA$ is a layered structure of height $\ell$, and we wish to
construct a layered structure $\fA^+$ of height $\ell{+}1$ over the same domain $A$, agreeing with the assignments made by $\fA$. Clearly, it suffices to fix the adjacent type of each primitive $(\ell{+}1)$-tuple $\bar{b}$ from $A$. 
Suppose we want to fix the adjacent type
of $\bar{b}$ and hence that of its reversal $\tilde{b}$. 
To do so, we consider each predicate $p$ in turn---of arity, say, $m$---and decide, for any
$m$-tuple $\bar{c}$ from $A$ whose primitive generator is $\bar{b}$, whether $\fA \models p[\bar{c}]$.  
Now repeat this process for all pairs of mutually inverse 
primitive words $(\bar{b}, \tilde{b})$ from $A$ having primitive length~$\ell{+}1$.
Since every tuple $\bar{c}$ considered for inclusion in the extension of some predicate has primitive length~$\ell{+}1$, these assignments will not clash with any previously made
in the original structure~$\fA$. Moreover, since, by Lemma~\ref{lemma:main}, every $m$-tuple $\bar{c}$ assigned in this process has a unique primitive generator
$\bar{b}$ (up to reversal), these assignments will not clash with each other.
Thus, to elevate
$\fA$ to a layered structure of height $\ell{+}1$, one takes each inverse pair $(\bar{b}, \tilde{b})$ of primitive $(\ell{+}1)$-tuples in turn,
and fixes the adjacent type of each $\bar{b}$ consistently with the existing assignments of all tuples generated by
proper infixes of $\bar{b}$, as given in the original structure~$\fA$.

We finish this section with an easy technical  observation  that will be needed in the sequel.
Denote by  $\vec{\bA}^m_k$ the set of all functions $f \in \bA^m_k$ such that $f(m) = k$. We refer to $f$ as a {\em final} adjacent function.
Thus, if $f \in \vec{\bA}^m_k$ is thought of as a stroll of length $m$ on some word $\bar{a}$ of length $k$, then that stroll ends at the final position of $\bar{a}$.
\begin{lemma}
	Let $\bar{c}$ be a word of length $m$ over some alphabet $A$, and $d$ an element of~$A$ that does not appear in $\bar{c}$. If $\bar{c}d$ is not primitive, then neither is~$\bar{c}$. In fact, there is a word~$\bar{a}$
	of length $k < m$ and a function $f \in \vec{\bA}^m_{k}$ such that~$\bar{a}^f = \bar{c}$.
	\label{lma:simplePrimitive}
\end{lemma}
\begin{proof}
  Suppose $\bar{c}d = \bar{b}^g$ for some word $\bar{b}$ of length $k{+}1$ (bounded by $m$) and some surjective map $g \in \bA^{m+1}_{k +1}$.
  Since $d$ does not occur in $\bar{c}$, it is immediate that $d$ occupies either the first or last position in $\bar{b}$, for
  otherwise, it would be encountered again in the entire traversal of $\bar{b}$ (as $g$ is adjacent and surjective). By reversing $\bar{b}$ if necessary,
  assume the latter, so that we may write $\bar{b} = \bar{a}d$, with $g(m+1) = k+1$. By adjacency, $g(m) = k$, so that
  setting $f = g \setminus \set{\langle m+1, k+1\rangle}$, we have the required $\bar{a}$ and $f$.
\end{proof}

We remark that, if $f \in \vec{\bA}^m_{k}$, then the function $f^+ = f \cup \langle m{+}1,k{+}1 \rangle$
satisfies $f^+ \in \vec{\bA}^{m{+}1}_{k{+}1}$. That is, we can extend $f$ by setting $f(m{+}1) = k{+}1$, retaining adjacency.
We utilise this fact as follows.
Let $\phi$ be a normal-form \AFv{\ell+1}-formula as given in~\eqref{eq:anf}.
Recalling that we write $\chi^h = \chi(x_{h(1)} \cdots x_{h(\ell{+}1)})$ for $\chi$ quantifier-free in $\AFv{\ell{+}1}$ and $h \in \bA^{\ell{+}1}_k$, 
we define the {\em adjacent closure of $\phi$}, denoted $\acl(\phi)$, to be:
\begin{equation*}
\bigwedge_{\hspace{2mm}i \in I\phantom{\bigwedge^{\ell+1}_{k+1}}\hspace{-5mm}} \hspace{-3mm}
\hspace{-1mm} \bigwedge_{\hspace{2mm}k=1\phantom{\bigwedge^{\ell+1}_{k+1}}\hspace{-4mm}}^{\ell-1} \hspace{-3mm}
\bigwedge_{f \in \vec{\bA}^{\ell}_{k}} \hspace{-2mm}
\forall \bx_{k} \exists x_{k+1}\, (\gamma_i)^{f^+} 
\wedge
\bigwedge_{\hspace{2mm}k=1\phantom{\bigwedge^{\ell+1}_{k+1}}\hspace{-5mm}}^{\ell}  \hspace{-4mm}
\bigwedge_{\hspace{2mm}g \in \bA^{\ell+1}_{k}\phantom{\bigwedge^{\ell+1}_{k+1}}\hspace{-8mm}} \hspace{-3mm}
\forall \bx_{k}\, \beta^g.
\end{equation*}  
Observe that the
conjunctions for the $\forall^\ell\exists$-formulas range over $f \in \vec{\bA}^{\ell}_{k}$ (so that $f^+ \in \bA^{\ell{+}1}$), while the conjunctions
for the purely universal formula range over the whole of $\bA^{\ell+1}_{k}$.
Up to trivial logical rearrangement and re-indexing of variables, $\acl(\phi)$ is actually a normal-form \AFv{\ell}-formula. 
In effect, $\acl(\phi)$ is the result of identifying various universally quantified variables in $\phi$ in a way which preserves adjacency.
The following~lemma~is~therefore~immediate. 
\begin{lemma}
  Let $\phi \in \AFv{\ell+1}$ be in normal-form. Then $\phi \models  \acl(\phi)$. 
  \label{lma:adjacentClosure}
\end{lemma}
Important notation mentioned in this and the previous sections is recapitulated in Table~\ref{table:prelim} for future reference.
\begin{table}[t]
	\begin{tabular}{c | l}
		\multicolumn{2}{l}{Functions $f: [1, m] \to [1, k]$ and tuples $\bar{a} = a_1 a_2 \cdots a_k$}\\
		\hline
		$\bA^m_k$ & the set of all adjacent functions $f: [1, m] \to [1, k]$ \\
		$\vec{\bA}^m_k$ & the set of all final adjacent functions $f: [1, m] \to [1, k]$ \\
		$f^+$ & $f \cup \{ (m+1, k+1) \}$ only if $f \in \vec{\bA}^m_k$ \\
		$\bx_k$ & $x_1 x_2 \cdots x_k$ \\
		$\tilde{a}$ & reversal of $\bar{a}$\\ 
		$\bar{a}^f$ & $a_{f(1)} a_{f(2)} \cdots a_{f(m)}$ \\
		\multicolumn{2}{l}{}\\ 
		\multicolumn{2}{l}{Formulas $\chi \in \AFv{[m]}$ with $f: [1, m] \to [1, k]$}\\
		\hline
		$\sig(\chi)$ & the signature of $\chi$ \\
		$\chi^{-1}$ & $\chi(x_m x_{m-1} \cdots x_1)$ \\
		$\hat{\chi}$ & $\chi \wedge \chi^{-1}$ \\
		$\chi^f$ & $\chi(\bx_k^f)$ \\
		
		\multicolumn{2}{l}{}\\ 
		\multicolumn{2}{l}{Normal form $\varphi$ (in $(\ell+1)$-variables)}\\
		\hline
		$\varphi$ & $\bigwedge_{i \in I} \forall \bx_{\ell} \exists x_{\ell+1}\, \gamma_i \wedge  \forall \bx_{\ell+1}\, \beta$ \\
		$\acl(\varphi)$ & 
		$
		\bigwedge_{i \in I}
		\bigwedge_{k=1}^{\ell-1}
		\bigwedge_{f \in \vec{\bA}^{\ell}_{k}}
		\forall \bx_{k} \exists x_{k+1}\, \gamma_i^{f^+} 
		\wedge
		\bigwedge_{k=1}^{\ell}
		\bigwedge_{g \in \bA^{\ell+1}_{k}}
		\forall \bx_{k}\, \beta^g
		$ \\
		\multicolumn{2}{l}{}\\ 
		\multicolumn{2}{l}{Adjacent $k$-types $\chi$ and the $k$-type of a $k$-tuple $\bar{a}$ in $\fA$}\\
		\hline
		$\Atp_k^\tau$ & the set of all adjacent $k$-types over $\tau$ \\
		$\atp^\fA_k(\bar{a})$ & the adjacent $k$-type of $\bar{a}$ in $\fA$ \\
		$\itp^\fA_k(\bar{a})$ & the incremental $k$-type of $\bar{a}$ in $\fA$ \\
		$\partial \xi$ & the incremental $k$-type included in $\xi$\\
		$\chi^+$ & $\chi(x_2 x_3 \cdots x_{k+1})$ \\
	\end{tabular}
	\caption{Quick reference guide for Sec~\ref{sec:preliminaries}~and~\ref{sec:primGen}.}
  \label{table:prelim}
	\end{table}
\section{Upper bounds for $\AFv{\ell}$ without equality}\label{section:AF-upper-bounds}
Our goal in the following two sections is to establish a small model property for each of the fragments $\AFv{\ell}$, for $\ell \geq 2$. 
We proceed by induction on $\ell$. The base case ($\ell = 2$) involves no work: it was shown in~\cite[Thm.~5.1]{GradelKV97}
that each satisfiable $\FOt$-sentence $\phi$ has a model of size $2^{O(\sizeof{\varphi})}$; 
since $\AFv{2} \subseteq \FOt$, the result follows.
For the inductive step, we reduce the case $\ell{+}1$, with exponential blow-up, to the case~$\ell$.
More precisely, we compute, for a given formula $\phi \in \AFv{\ell{+}1}$, 
an equisatisfiable formula $\psi \in \AFv{\ell}$, over an exponentially larger signature, together with bounds on the relative sizes of their
respective models. We thereby show that any satisfiable
$\AFv{\ell +1}$-sentence $\phi$ is satisfiable in a structure of size $\ft(\ell{-}1,O(\sizeof{\phi}))$, where $\ft$ 
is defined inductively by $\ft(0,n) = n$ and $\ft(k{+}1,n) = 2^{\ft(k,n)}$. 
To illustrate the proof strategy as perspicuously as possible, we confine our attention in this section to the
sub-fragment of $\AFv{\ell}${} \textit{without equality}. In the next section we generalize the result to the full fragment $\AFv{\ell }${}, at
the cost of some additional combinatorics. Returning to current affairs,
Table~\ref{table:noeq} provides a summary of important symbols that will be defined throughout this section. (There is no need to read it just~yet!)

\begin{table}
	\begin{tabular}{c | l}
		\multicolumn{2}{l}{Surjective $f, g \in \bA^m_k$, incremental $k$-types $\iota$, quantifier-free $\chi \in \AFv{[k]}$}\\
		\hline
		$\bD^{\mathrm{o}}_{k}$ & all pairs $\langle i, j \rangle$ for $1 \leq i < j \leq k$ with $j {-} i {+} 1 \geq 2$ and odd \\
		$D^+$ & $\{ \langle i+1, j+1 \rangle \mid \langle i, j \rangle \in D \}$ \\
		$R^*$ & the equivalence closure of a binary relation $R$\\
		$f \overset{D}{=} g$ & $\langle f(i),g(i) \rangle \in D^*$ for all $i \in [1,m]$\\
		$\iota$ is $D$- & \multirow{2}{*}{$f \overset{D}{=} g$ implies $\iota \models p(\bx^f_k) \leftrightarrow p(\bx^g_k)$ for all $p {\in} \sig(\chi)$, all $f, g \in \bA^m_k$} \\
		compatible       & \\
		$\chi$ is $D$- & \multirow{2}{*}{ there exists an $\xi \in \Atp_k^\tau$ s.t. $\xi \models \chi$ and $\partial \xi$ is $D$-compatible } \\
		consistent  & \\
		\multicolumn{2}{l}{}\\ 
		\multicolumn{2}{l}{Construction of $\psi$ from $\phi$}\\
		\hline
		$d_{k}(\bx_k)$ & atom implying that $\bx_k$ is a palindrome \\
		$\delta_D$($\bx_\ell$)  & $\bigwedge \set{d_{j{-}i{+}1}(x_i \cdots x_j) \mid \langle i,j \rangle \in D}$ \\ 
		$p_\zeta(\bx_{\ell-1})$ & atom implying there is some $x$ s.t. $x\bx_{\ell-1}$ realizes $\zeta$ \\
		$\fB \times H$  & the model obtained by cloning elements of $\fB$ for each $h \in H$ \\
	\end{tabular}
	\caption{Quick reference guide for Sec~\ref{section:AF-upper-bounds}.}
	\label{table:noeq}
	\end{table}

For the next few lemmas (\ref{lma:coherencyAndConsistency}--\ref{lma:reductionDirection2}), fix an equality-free, normal-form \AFv{\ell{+}1}-formula $\phi$ over some signature $\tau$, as given in~\eqref{eq:anf}, with $\ell \geq 2$.
We construct an equality-free, normal-form formula $\psi \in \AFv{\ell}$ such that: (i) if $\phi$ is satisfiable over some domain $A$, then so is $\psi$; and (ii) if $\psi$ is satisfiable
over a domain $B$, then $\phi$ is satisfiable over some domain $C$, with $|C|/|B|$ 
bounded by an exponential~function~of~$\sizeof{\phi}$. For the remainder of this section, all formulas will silently be assumed to be equality-free, and likewise
for adjacent types.

We take $\psi$ to have the form
\begin{equation*}
\psi\coloneqq \acl(\phi) \wedge \psi_1  \wedge \psi_2 \wedge \psi_3 \wedge \psi_4,
\end{equation*}
where $\acl(\phi)$ is the adjacent closure of $\phi$ (featured in Lemma~\ref{lma:adjacentClosure}) and 
$\psi_1, \dots, \psi_4$ are $\AFv{\ell}$-formulas over an expanded signature.  
To motivate this construction of $\psi$, we first suppose $\phi$ is satisfiable, and consider any model $\fA \models \phi$.
We then introduce the conjuncts $\psi_1, \dots, \psi_4$ one by one,
simultaneously defining an expansion $\fA^+$ of $\fA$ by interpreting the new predicates so as to satisfy these conjuncts. (Of course, the $\psi_1, \dots, \psi_4$ depend only on $\phi$, and not on $\fA^+$.) Since,
by Lemma~\ref{lma:adjacentClosure}, $\fA \models \acl(\phi)$, we have $\fA^+ \models \psi$. Our main task is then to show that, given any finite layered structure $\fB \models \psi$ of height $\ell$, we can
construct a finite layered structure $\fC \models \phi$ of height $\ell{+}1$, with
$|C|/|B|$ bounded by some exponential~function~of~$\sizeof{\phi}$.
This establishes the finite model property for $\AFv{\ell+1}$, and reduces its
satisfiability problem to that for $\AFv{\ell}$, though with exponential blow-up. In effect,
$\psi$ specifies just the right amount of information concerning the primitive $\ell$-tuples occurring in any of its models
to ensure that the adjacent types of primitive $(\ell{+}1)$-tuples can be assigned in such a way as to build a model of $\phi$. 

%
%

We now proceed to the definition of the conjuncts $\psi_1, \dots, \psi_4$, and the construction of the expansion $\fA^+$.
Turning first to $\psi_1$, 
for each $s$ ($3 \leq 2s + 1 \leq \ell$) we introduce a fresh $(2s{+}1)$-ary predicate $d_{2s{+}1}$, and declare that a tuple $\bar{b} \in A^{2s+1}$
satisfies $d_{2s+1}$ in $\fA^+$ just in case $\bar{b}$ is a palindrome. It is then easy to verify that $\fA^+$
is a model of the sentence $\psi_1$ given as
\begin{equation*}
	\psi_1 \coloneqq \bigwedge_{3 \leq 2s + 1 \leq \ell} \forall \bx_{k+1}\ d_{2s{+}1}(x_1 \cdots x_s x_{s + 1} x_s \cdots x_1).
\end{equation*}
Conversely, if $\fB$ is any structure such that $\fB \models \psi_1$, and $\bar{c} \in B^{2s{+}1}$ is a palindrome
($3 \leq 2s + 1 \leq \ell$), 
then $\fB \models d_{2s+1}[\bar{c}]$.
Observe that $\psi_1$ employs at most $(\ell+1)/2 < \ell$ variables. 

The predicates $d_s$ just introduced will be useful for specifying that various tuples of objects exhibit certain sets of \textit{defects}, in
the sense of Sec.~\ref{sec:primGen}. Recall that a defect of an $\ell$-tuple $a_1 \cdots a_\ell$ is a pair of integers $\langle i, j \rangle$ such that 
$a_i \cdots a_j$ is a non-trivial palindrome. (Note that the length of this factor is $j{-}i{+}1 \geq 2$.)
In the sequel, for any $k \geq 2$, we denote by $\bD^{\mathrm{o}}_{k}$ the set of all pairs $\langle i, j \rangle$ for $1  \leq i < j \leq k$ such that
$j{-}i{+}1$ is greater than two and odd. (Thus, $\bD^{\mathrm{o}}_{2}= \emptyset$.)
Now, for any $D \subseteq \bD^{\mathrm{o}}_{\ell{-}1}$, we write $\delta_D$ for the formula
\begin{equation*}
\delta_D:= \bigwedge \set{d_{j{-}i{+}1}(x_i \cdots x_j) \mid \langle i,j \rangle \in D}.
\end{equation*}
Intuitively, $\delta_D$ says that any satisfying tuple has a defect set which includes $D$. Note that $\delta_D$ is defined only for 
$D \subseteq \bD^{\mathrm{o}}_{\ell{-}1}$: for technical reasons, we are not interested in defects
corresponding to even-length palindromic factors.

Turning now to $\psi_2$, we introduce, for each adjacent $\ell$-type $\zeta \in \Atp_\ell^\tau$, an $(\ell-1)$-ary predicate $p_\zeta$
intended to identify the tails of $\ell$-tuples satisfying $\zeta$. 
Specifically, we declare that $\fA^+ \models p_\zeta[\bar{b}]$ just in case there is some 
$a \in A$ such that $\fA \models \zeta[a\bar{b}]$.
It is then easy to verify that $\fA^+$
is a model of the sentence $\psi_2$ given as 
\begin{equation*}
\psi_2 \coloneqq	\bigwedge_{\zeta \in \Atp^\tau_\ell} \forall \bx_{\ell} \Big(\zeta \rightarrow p_\zeta(x_2 \cdots x_{\ell})\Big).
\end{equation*}
Conversely, if $\fB$ is a structure interpreting the signature $\tau$, and $\fB^+$ is an expansion of $\fB$ such that $\fB^+ \models \psi_2$, then 
$\atp^\fB[a\bar{b}] = \zeta$ implies $\fB^+ \models p_\zeta[\bar{b}]$.
Observe that $\psi_2$ employs $\ell$ variables.

The construction of $\psi_3$ and $\psi_4$ requires some preliminary work concerning palindromes, as tuples containing palindromic factors pose a
technical challenge in the construction of our formula $\psi$. To explain why,
let $\bar{a}$ be an $(\ell {+} 1)$-tuple over some set $A$ and $\xi$ an adjacent
$(\ell {+} 1)$-type, and imagine that we wish to define a structure
$\fA$ on $A$ in such a way that $\atp^\fA[\bar{a}] = \xi$. Is this at all possible? Suppose, for example, that $\bar{a} = abcbd$, and  that $\xi$ 
contains the literals $p(x_1 x_2 x_3 x_2 x_3 x_4 x_5)$ and $\neg p(x_1 x_2 x_3 x_4 x_3 x_4 x_5)$. A moment's thought shows that we cannot
have $\atp^\fA[\bar{a}] = \xi$, since that would require both $\fA \models p[abcbcbd]$ and $\fA \models \neg p[abcbcbd]$. We need to be able to identify 
this sort of situation using only the resources of $\AF$. The next lemma explains how.

Recall the apparatus of incremental types introduced in Sec.~\ref{sec:preliminaries}: if $\bar{c}$ is a $k$-tuple over some structure $\fA$,
then $\itp^\fA[\bar{c}]$ is the set of covering adjacent $k$-literals (i.e.~those featuring all variables in $\bx_k$) satisfied by $\bar{c}$ in $\fA$.
Recall also the notation introduced in Lemma~\ref{lemma:main2}: 
for any pair of adjacent functions
$f,g \colon [1,m] \rightarrow [1,k]$ and any set
$D \subseteq \set{\langle i,j \rangle \mid 1 \leq i < j \leq k}$, we
write $f \overset{D}{=} g$ if
$\langle f(i), g(i) \rangle \in D^{*}$ for all $i$ ($1 \leq i \leq m$), 
where $D^{*}$ denotes the equivalence closure of $D$.
Let us say that an incremental $k$-type  $\iota$ is $D$-\textit{compatible} if
$f \overset{D}{=} g$ implies $\iota \models p(\bx^f_k) \leftrightarrow p(\bx^g_k)$  for all predicates $p$ in the
signature of $\iota$ and all
surjective adjacent functions $f,g \colon [1,m] \rightarrow [1,k]$ where
$m$ is the arity of $p$.

\begin{lemma}
Let $\bar{c}$ be a primitive $k$-tuple over $A$, and $D$ the defect set of $\bar{c}$.
If $\fA$ is a structure with domain $A$, then $\iota = \itp^\fA[\bar{c}]$ is $D$-compatible.
Conversely, if $\iota$ is a $D$-compatible incremental $k$-type over some signature $\tau$, 
then there is a structure $\fA$ interpreting $\tau$ over $A$ such that $\itp^\fA[\bar{c}] = \iota$. 
\label{lma:coherencyAndConsistency}
\end{lemma}
\begin{proof}
The first statement of the lemma is almost immediate. Fix
a predicate $p$ of arity $m$ interpreted by $\fA$, and suppose $f, g \in \bA^m_k$ are surjective. 
We must show that $f \overset{D}{=} g$ implies $\iota \models p(\bx^f_k) \leftrightarrow p(\bx^g_k)$.
But by Lemma~\ref{lemma:main2}, if $f \overset{D}{=} g$, then
$\bar{c}^f = \bar{c}^g$, hence
$\fA \models p[\bar{c}^f] \Leftrightarrow \fA \models p[\bar{c}^g]$. 
Thus $\iota \models p(\bx^f_k) \leftrightarrow p(\bx^g_k)$.
	
For the second statement, define $\fA$ over the domain $A$ by setting, for any predicate $p$ of $\tau$ having arity, say, $m$:
\begin{equation*}
p^{\fA} = \set{\bar{a}^f \mid \text{there exists a surjective $f \in \bA^m_k$ such that $\iota \models p(\bx^f_k)$}}.
\end{equation*}
To show that $\fA \models \iota[\bar{a}]$, fix any $p \in \tau$ with arity $m$. If $\iota$ contains the (covering, adjacent) atom $\alpha\coloneqq p(\bar{x}^f)$, where
then it is immediate
from the construction of $\fA$ that $\fA \models \alpha[\bar{a}]$. It remains to show that
if $\iota$ contains the negated atom $\nu\coloneqq \neg p(\bx^f_k)$, then $\fA \models \nu[\bar{a}]$. 
Suppose otherwise. 
From the construction of $\fA$, we have $\iota \models p(\bx^g_k)$ for some surjective adjacent $g\colon [1,m] \rightarrow [1,k]$ such that $\bar{a}^f = \bar{a}^g$. 
By Lemma~\ref{lemma:main2}, $f \overset{D}{=} g$. Yet 
$\iota$ is by assumption $D$-compatible, whence $\iota \models p(\bar{x}^f) \leftrightarrow p(\bar{x}^g)$, contradicting the
fact that $\xi$ contains both $\neg p(\bx^f_k)$ and $p(\bx^g_k)$.
\end{proof}

To explain the significance of this lemma, let us return to our example of the 5-tuple $\bar{a} = abcbd$, and the adjacent 5-type $\xi$ 
containing literals $p(x_1 x_2 x_3 x_2 x_3 x_4 x_5)$ and $\neg p(x_1 x_2 x_3 x_4 x_3 x_4 x_5)$. 
Notice that these literals are covering, and thus are contained in $\partial \xi$, the incremental $5$-type included in $\xi$. 
As we have seen, $\bar{a}$ cannot be assigned the  adjacent-type $\xi$, because
it makes inconsistent demands in respect of the 7-tuple $abcbcbd$. Observe, however, that
the defect set of $\bar{a}$ is $D = \set{\langle 2,4 \rangle}$, and when we
write the argument sequences $x_1 x_2 x_3 x_2 x_3 x_4 x_5$ and
$x_1 x_2 x_3 x_4 x_3 x_4 x_5$ in the form
$\bx_5^f$ and $\bx_5^g$, respectively, it is easily checked that $f\overset{D}{=} g$, whence $\partial \xi$ is not $D$-compatible.
On the other hand, setting $\ell \coloneqq 4$, and supposing the formula $\psi_1$ to hold, we see that 
$bcb$ satisfies the predicate $d_3$, whence the quadruple
$abcb$ satisfies the $\AFv{4}$-formula $\delta_D = d_3(x_2x_3x_4)$. Crucially,
the fact that $abcb$ satisfies $\delta_D$ allows us to
detect, using resources available only in $\AFv{4}$, that 
$\bar{a}= abcbd$ cannot be assigned the adjacent 5-type $\xi$, because $\partial \xi$ is not $D$-compatible.

Returning now to the construction of $\psi_3$ and $\psi_4$, and recalling that
$\bD^{\mathrm{o}}_{\ell{-}1}$ denotes the set of all pairs $\langle i, j \rangle$ for $1  \leq i < j \leq \ell{-}1$ such that
$j{-}i{+}1$ is greater than 2 and odd, consider any subset $D \subseteq \bD^{\mathrm{o}}_{\ell{-}1}$. Denote by $D^+$ the result of adding 1 to both components of every pair in $D$, i.e.~$D^+ \coloneqq \set{\langle i{+}1, j{+}1 \rangle \mid \langle i,j \rangle \in D}$.
(Thus, $D^+  \subseteq \bD^{\mathrm{o}}_{\ell}\subseteq \bD^{\mathrm{o}}_{\ell{+}1}$.)
The intuition here is that if $D$ represents a set of odd-length non-trivial palindromic factors of
an $(\ell{-}1)$-tuple $\bar{a}$, and $a$ and $b$ are some elements, 
then $D^+$ represents a set of odd-length non-trivial  palindromic factors of the $\ell$-tuple
$a\bar{a}$ and indeed also of the $(\ell{+}1)$-tuple
$a\bar{a}b$. Say that a quantifier-free \AFv{\ell+1}-formula $\chi$ is $D^+$-\textit{consistent} if there exists an adjacent $(\ell{+}1)$-type $\xi$ over the relevant signature such
that $\xi \models \chi$ and $\partial \xi$ is $D^+$-compatible. Finally, define the formulas $\psi_3$ and $\psi_4$ as:
\begin{align*}
\begin{split}
\psi_3\coloneqq \smash{\bigwedge_{i \in I} \bigwedge_{\smash{\zeta \in {\Atp}^\tau_\ell}} \bigwedge_{D \subseteq \bD^{\mathrm{o}}_{\ell{-}1}} \hspace{-0.45cm} \forall \bx_{\ell{-}1}} 
   & \exists x_\ell \Big(\big(\delta_D \wedge p_\zeta(\bx_{\ell{-}1}) \big) \to\\
   & \bigvee \set{\eta \in \Atp^\tau_\ell \mid 
				\text{$(\zeta \wedge \hat{\beta} \wedge \gamma_i\wedge \eta^+)$ is $D^+$-consistent}}\Big)
	\end{split}
\\
\psi_4 \coloneqq \begin{split}
	\smash{\bigwedge_{\smash{\zeta \in \Atp^\tau_\ell}} \bigwedge_{D \subseteq \bD^{\mathrm{o}}_{\ell{-}1}} \forall \bx_{\ell}}
	 \Big(\big(&\delta_D \wedge p_\zeta(\bx_{\ell{-}1}) \big) \to\\
	& \bigvee \set{\eta \in \Atp^\tau_\ell \mid 
		\text{$(\zeta \wedge \hat{\beta} \wedge \eta^+)$ is $D^+$-consistent}}\Big).
\end{split}	
\end{align*}

We claim that $\fA^+ \models \psi_3$. Suppose $\bar{a}$ is an $(\ell{-}1)$-tuple such that $\fA^+ \models \delta_D[\bar{a}]$ and
$\fA^+ \models p_\zeta[\bar{a}]$. By construction of $\fA^+$, there exists $a \in A$ such that $\atp^\fA[a\bar{a}] = \zeta$.
Now, fix an index $i \in I$ of some witness requirement in $\phi$. 
Since $\fA \models \phi$, there exists an element $b \in A$ such that
$\fA \models \gamma_i[a\bar{a}b]$,
and in fact such that $\fA \models (\hat{\beta} \wedge \gamma_i)[a\bar{a}b]$.
Letting 
$\xi \coloneqq \atp^\fA[a\bar{a}b]$ and $\eta \coloneqq \atp^\fA[\bar{a}b]$, we have
$\xi \models \zeta \wedge \hat{\beta} \wedge \gamma_i \wedge \eta^+$.
On the other hand, since $\fA^+ \models \delta_D[\bar{a}]$ we have, by construction of $\fA^+$, that the defect set of $\bar{a}$ includes $D$;
and hence that the defect set of $a\bar{a}b$ includes $D^+$. By the first statement of Lemma~\ref{lma:coherencyAndConsistency}
we have that $\partial \xi$ is $D^+$-compatible.
Hence, $b$ is a witness for the quantifier $\exists x_\ell$ in $\psi_3$ with respect to $\bar{a}$, so that 
$\fA^+ \models \psi_3$ as required. By similar reasoning, $\fA^+ \models \psi_4$.

To motivate these formulas, consider first $\psi_3$, and suppose $\psi_1 \wedge \psi_2 \wedge \psi_3$ has a model, say, $\fB$. We may assume, by Lemma~\ref{lma:approx}, that 
$\fB$ is a layered structure of height $\ell$, and we may take $\fB^-$ to be the reduct of $\fB$ to the signature $\tau$ (i.e.~we forget the predicates $p_\zeta$ and $d_{2s+1}$). 
Suppose, now $a$ is an element of $B$ and $\bar{a}$ an $(\ell{-}1)$-tuple over $B$, and let
$\zeta= \atp^{\fB^-}[a\bar{a}]$. Let us further imagine that the defect set of $\bar{a}$ is $D \subseteq \bD^{\mathrm{o}}_{\ell{-}1}$. (We shall not worry
about even-length palindromic factors for the present.) Since $\fB \models \psi_1 \wedge \psi_2$, the $(\ell{-}1)$-tuple $\bar{a}$ will satisfy $\delta_D$ and $p_\zeta$
in $\fB$. Of course, $\fB$ does not in general define the adjacent types of $(\ell{+}1)$-tuples (unless they are non-primitive), but
we would like to be able to find, for each index $i \in I$, an element $b \in B$ such that the $(\ell{+}1)$-tuple $a\bar{a}b$ \textit{could} be assigned an adjacent
type satisfying the existential requirement $\gamma_i$ of $\phi$ without violating the universal constraints $\beta$.
Anticipating Lemma ~\ref{lma:reductionDirection2}, this is what $\psi_3$ does: the witnesses it provides for $x_\ell$ have the required properties. 
Formula $\psi_4$ plays an analogous role with respect to the universal requirements of $\phi$.
But we are getting ahead of ourselves. For the present, we have proved:
\begin{lemma}
	Suppose $\fA \models \phi$. Then we can expand $\fA$ to a model $\fA^+ \models \psi$.
	\label{lma:reductionDirection1}
\end{lemma}

Having defined $\psi$ and established Lemma~\ref{lma:reductionDirection1}, we wish to argue for the converse direction.
We begin with three straightforward, technical lemmas.
When dealing with equality-free formulas, we may freely duplicate elements in their models. 
Let $\fB$ be a $\tau$-structure, and $H$ a non-empty set of indices. We define 
the structure $\fB \times H$ over the Cartesian product $B \times H$ as follows: for any
$p \in \tau$ of arity $k$, and any $k$-tuples $b_1 \cdots b_k$ over $B$ and $h_1 \cdots h_k$ over $H$, set
$\fB \times H \models p[\langle b_1,h_1\rangle  \cdots \langle b_k, h_k\rangle]$ if and only if 
$\fB \models p[b_1 \cdots b_k]$.
By routine structural induction, we have: 
\begin{lemma}
	Let $\psi$ be an equality-free first-order formula all of whose free variables occur in $\bx_k$, $\fB$
	a structure interpreting the signature of $\psi$, and $H$ a non-empty set. Then, for any tuples 
	$\bar{b} = b_1 \cdots b_k$ from $B$ and $h_1 \cdots h_k$ from $H$,
	$\fB \models \psi[\bar{b}]$ if and only if $\fB \times H \models \psi[\langle b_1, h_1\rangle \cdots \langle b_k, h_k\rangle]$.\label{lma:cartesian}
\end{lemma}

The following very simple lemma will be used in conjunction with Lemma \ref{lma:cartesian}.
\begin{lemma}\label{lma:defectInheritance}
	Suppose $\bar{b}$ is a $k$-tuple over some set $B$ and $H$ a non-empty set.
	Then, for any $h_1, \dots, h_k \in H$, if $\bar{c}$ is the tuple 
	$\langle b_1, h_1\rangle \cdots \langle b_k, h_k\rangle$ over $B \times H$, then 
	every defect of $\bar{c}$ is a defect of $\bar{b}$.
\end{lemma}
\begin{proof}
	Pick some defect $\langle i, j \rangle$ of $\bar{c}$.  Thus, $c_i \cdots c_j$ is a non-trivial palindrome,
	i.e.~$\langle b_{i + m}, k_{i + m} \rangle = \langle b_{j - m}, k_{j - m} \rangle$ for all $0 \leq m \leq \lfloor(j-i)/2\rfloor$.
	But then $b_i \cdots b_j$ is certainly a non-trivial palindrome, and hence $\langle i, j \rangle$ a defect of $\bar{b}$ as required.
\end{proof}

Finally, the following combinatorial lemma allows us to extend the technique of `circular witnessing' 
(see, e.g.~\cite[p. 62--64]{GradelKV97})
frequently used in the analysis of two-variable logics, to the languages $\AFv{k}$.
\begin{lemma}\label{lma:simpleComb}
	For any integer $k>0$ there is a set $H$ with $|H| = (k^2 + k+1)^{k+1}$
	and a function $g\colon H^k \rightarrow H$ such that,
	for any tuple $\bar{t} \in H^k$ consisting of the elements $t_1, \dots, t_k$ in some order:
	{\em (i)} $g(\bar{t})$ is not in $\bar{t}$;
	{\em (ii)} if $\bar{t}' \in H^k$ consists of the elements $\set{t_2, \dots, t_k, g(\bar{t})}$ in some order, then
	$g(\bar{t}')$  is not in $\bar{t}$ either.
\end{lemma}
\begin{proof}
	Let $z = k^2{+}k{+}1$ and $H = [1,z]^{k{+}1}$.
	Thus, $|H|$~is~as required.
	Writing any element of $H$ as the word $i\bar{s}$, where $i \in [1,z]$ and $\bar{s} \in [1,z]^k$,
	let $g$ be defined by $g(i_1\bar{s}_1, \dots, i_k\bar{s}_k) = i_0 i_1 \cdots i_k$, where $i_0$ is the smallest positive integer
	not in the set $S = \set{i_1, \dots , i_k} \cup \bar{s}_1 \cup \cdots \cup \bar{s}_k$. (For
	brevity, we are here identifying words over the integers with the sets of their members.) Note that $i_0 \in [1,z]$ since $|S| <z$.
	Now let some tuple $\bar{t} \in H^k$ be given,
	consisting of words $t_1, \dots, t_k$ in some order, and write $t_h = i_h \bar{s}_h$ (for all $1 \leq h \leq k$), where $i_h \in [1,z]$ and $\bar{s}_h \in [1,z]^k$.
	Thus, $t = g(\bar{t})$ is a word of the form $i_0\bar{s}$, where $\bar{s}$ consists of $i_1, \dots, i_{k}$ in some order,
	and $i_0$ does not occur anywhere in~$\bar{t}$.
	Condition (i) is then immediate because of the choice of $i_0$. For condition (ii), we observe that
	$t' = g(\bar{t}')$ is a word of the form $i' \bar{s}'$, where $\bar{s}'$ consists of $i_0, i_2, \dots, i_{k}$ in some order, and $i'$ does not occur
	in any of the words in $\bar{t}'$. By the choice of $i'$, it is immediate that $t' \not \in \set{t_2, \dots, t_k}$. But the value $i_1$ (the first letter of $t_1$)
	occurs in $g(\bar{t})$ (which belongs to the tuple $\bar{t}'$), whence $i' \neq i_1$. It follows that $t' \neq t_1$, whence~$t'$ is not in $\bar{t}$, as required.
\end{proof}

With the help of these three technical lemmas, we provide the advertised converse to Lemma~\ref{lma:reductionDirection1}.
\begin{lemma}
	Suppose $\fB \models \psi$. Then we can construct a model $\fC^+ \models \phi$ such that $|C^+|/|B| \leq |I| \cdot (\ell^2+\ell+1)^\ell$. 
	\label{lma:reductionDirection2}
\end{lemma}
\begin{proof}
	We repeat here for convenience the formula $\phi$ as defined in~\eqref{eq:anf}:
	\begin{equation*}
		\bigwedge_{i \in I} \forall \bx_{\ell} \exists x_{\ell+1}\, \gamma_i \wedge 
		\forall \bx_{\ell+1}\, \beta.
	\end{equation*} 
In the sequel, we refer to the conjunct $\forall \bx_{\ell} \exists x_{\ell+1}\, \gamma_i$ as the $i$th \textit{witness requirement of} $\phi$.
Our first step is to define a collection of `pseudo-witnesses' in the structure $\fB$ for the various witness requirements of $\phi$. These will be used later
to select actual witnesses in the model $\fC^+$ of $\phi$ that we eventually construct.
Let $\fB^-$ be the reduct of $\fB$ to the signature $\tau$ (i.e.~we forget the predicates $p_\zeta$ and $d_{2s+1}$).
Consider any $\ell$-tuple $\bar{b} = b_1 \cdots b_\ell$: let $\zeta = \atp^{\fB^-}[\bar{b}]$ and let $D$ be the set of defects of $b_2 \cdots b_\ell$ corresponding to odd-length palindromes, i.e.,~those defects $\langle i, j \rangle$ for which $j{-}i{+}1$ is odd.
Since $\fB \models \psi_1 \wedge \psi_2$, we have $\fB \models \delta_{D}[b_2 \cdots b_\ell]$ and
$\fB \models p_\zeta[b_2 \cdots b_\ell]$. And since $\fB \models \psi_3$, given any $i \in I$, we may select $b \in B$ 
such that $\eta= \atp^\fB[b_2 \cdots b_{\ell} b]$ is one of the disjuncts in the consequent
\begin{equation*}
\bigvee \set{\eta \in \Atp^\tau_\ell \mid 
	\text{$(\zeta \wedge \hat{\beta} \wedge \gamma_i\wedge \eta^+)$ is $D^+$-consistent}}
\end{equation*}
of the implication occurring in $\psi_3$. It follows that
there exists some $\xi \in \Atp^\tau_{\ell+1}$ such that $\xi \models (\zeta \wedge \hat{\beta} \wedge \gamma_i\wedge \eta^+)$
and $\partial \xi$ is $D^+$-compatible. 
We identify $b$ as the `pseudo-witness' for the $i$th witness requirement of $\phi$, and write $b_{\bar{b},i} \coloneqq b$ and $\xi_{\bar{b},i} \coloneqq \xi$ to make the dependencies on $i$ and $\bar{b}$ explicit. 
	
	Next we inflate the model $\fB$ using the construction of Lemma~\ref{lma:cartesian}. Let $H$ be a set of cardinality $\ell^2 + \ell + 1$ and let $g \colon H^\ell \rightarrow H$ be a function satisfying conditions 
	(i) and (ii) guaranteed by
	Lemma~\ref{lma:simpleComb}. 
	Define  $\fC=\fB^- \times (I \times H)$, and define the layered structure $\fC'$ (with the same domain, $C$) to be the result of disregarding any tuples in $\fC$ of primitive length greater than $\ell$. Thus, we may write elements of~$C$ as triples $(b,i,j)$, where $b \in B$, $i \in I$ and
	$j \in H$; and of course, $\fC'$ has height~$\ell$. By Lemma~\ref{lma:cartesian}, $\fC \models \psi$, and bearing in mind that $\psi \in \AFv{\ell}$, by Lemma~\ref{lma:approx}, $\fC' \models \psi$.
	We may now define a collection of `witness
	functions' $w_{i} \colon C^{\ell}\rightarrow C$, where $i$ ranges over $I$. 
	Consider any $\ell$-tuple  $\bar{c} \coloneqq c_1 \cdots c_\ell$ of elements in $C$, with $c_h= (b_h,i_h,j_h)$ for each $h$ ($1 \leq h \leq \ell$). 
	Writing $\bar{b} = b_1 \cdots b_\ell$, and recalling the pseudo-witness $b_{\bar{b},i}$ identified in the previous paragraph,
	we define $w_{i}(\bar{c})$ to be the element $(b_{\bar{b},i}, i, g(j_1 \cdots j_\ell))$.
	Thus, $w_i(\bar{c})$ is a judiciously chosen copy of
	$b_{\bar{b},i}$, able, potentially,  to serve as an $i$th witness for $\phi$. Most of the work will be done by its first component, $b_{\bar{b},i}$; the remaining components, 
	$i$ and $g(j_1 \cdots j_\ell)$, simply ensure that the eventually selected witnesses do not, as it were, tread on each others' toes. 
	Note that the functions $w_i$ satisfy the following properties:
	\begin{description}
	\item[(w1)] for fixed $\bar{c}$, the $w_{i}(\bar{c})$
	are distinct as $i$ varies over $I$;
	\item[(w2)] $w_{i}(\bar{c})$ does not occur in $\bar{c}$; 
	\item[(w3)] 
	if $\bar{c} = c_1 \cdots c_\ell$ and $\bar{c}'$
	is an $\ell$-tuple consisting of the elements
	$c_2, \dots, c_\ell, w_{i}(\bar{c})$ in some order, then 
	$w_{i'}(\bar{c}')$ does not occur in $\bar{c}$ for any $i' \in I$.
	\end{description}
	Indeed, {(w1)} is immediate from the fact  $w_i(\bar{c})$ contains $i$ as its second element; {(w2)} and {(w3)}
	follow, respectively, from conditions (i) and 
	(ii) on $g$ guaranteed in
	Lemma~\ref{lma:simpleComb}. 

	We are now ready to elevate the layered structure $\fC'$ to a layered structure $\fC^+$ of height $\ell{+}1$ such that $\fC^+ \models \phi$.  
	We first ensure that $\fC^+$ provides witnesses for the 
	various witness requirements of $\phi$.
	Fix any $\ell$-tuple $\bar{c} = c_1 \cdots c_\ell$ and any $i \in I$. 
	We have two cases, depending on whether the $(\ell{+}1)$-tuple
	$\bar{c}\, w_{i}(\bar{c})$ is primitive. Suppose first that it is not. By {(w2)}, $w_{i}(\bar{c})$ is not an element of $\bar{c}$, whence
	by Lemma~\ref{lma:simplePrimitive}
	there is some $k$-tuple $\bar{d}$ ($k < \ell$) and $f \in \vec{\bA}^{\ell}_{k}$ such that $\bar{d} = \bar{c}^f$. 
	As before, define $f^+ \in \vec{\bA}^{\ell+1}_{k+1}$ extending $f$ by setting $f(k+1) = \ell+1$.
	Since $k < \ell$, and $\fC' \models \acl(\phi)$, there exists $c' \in C$ such that $\fC' \models \gamma^{f^+}_i[\bar{d}c']$.
	By Lemma~\ref{lma:essentiallyTrivial}, $\fC' \models \gamma_i[(\bar{d}c')^{f^+}]$, or in other words, $\fC' \models \gamma_i[\bar{c}c']$,
	so that a witness $c'$ is already present in respect of the tuple $\bar{c}$ and
	the index $i$. (Notice that we are throwing the element $w_{i}(\bar{c})$ away, as an $i$th witness is already present for the tuple $\bar{c}$.)
	Suppose on the other hand that $\bar{c}\, w_{i}(\bar{c})$ is primitive. 
	Recall that, for all $h$ ($1 \leq h \leq \ell$), the element $c_h$ has the form $(b_h,i_h,j_h)$, and that 	
	$w_{i}(\bar{c})= (b_{\bar{b},i},i,g(j_1 \cdots j_\ell))$.
	Recall in addition 
	the adjacent $(\ell{+}1)$-type $\xi_{\bar{b},i}$
	identified earlier.	
	To reduce notational clutter, we 
	write $c$ for $w_{i}(\bar{c})$ and 
	$\xi$ for $\xi_{\bar{b},i}$.
	Setting $\zeta = \atp^{\fB^-}[\bar{b}]$ and $\eta = \atp^{\fB^-}[b_2 \cdots b_\ell b_{\bar{b},i}]$, Lemma~\ref{lma:cartesian} guarantees that
	$\zeta = \atp^{\fC'}[\bar{c}]$ and $\eta = \atp^{\fC'}[c_2 \cdots c_\ell\, c]$. 
	Moreover, by Lemma~\ref{lma:defectInheritance}, any defect of $c_2 \cdots c_\ell$ is evidently
	a defect of $b_2 \cdots b_\ell$. In fact, since $\bar{c}$ is primitive and
	$c = w_i(\bar{c})$ does not occur in $\bar{c}$, the tuple $\bar{c}c$ 
	cannot have any defects corresponding to even-length palindromes, and cannot have any defects of the forms $\langle 1, j\rangle$ or $\langle j, \ell+1\rangle$. Thus, letting
	$D$ denote the set of defects of $b_2 \cdots b_\ell$ corresponding to odd-length palindromes, we see that the defect set of $\bar{c}c$ is $D^+$. 
	But we have already argued that $\xi \models (\zeta \wedge \hat{\beta} \wedge \gamma_i\wedge \eta^+)$, and moreover that
    $\partial \xi$ is $D^+$-compatible.
	By Lemma~\ref{lma:coherencyAndConsistency},  
	there is a structure over the elements of $\bar{c}c$ in which the incremental type of $\bar{c}c$ is $\partial \xi$. Thus we may consistently 
	fix $\itp^{\fC^+}[\bar{c}c]$ to be $\partial \xi$. (Observe that this also fixes the incremental type of the reversed tuple
	$c\tilde{c}$.)
	Since $\zeta = \atp^{\fC'}[\bar{c}]$, $\eta = \atp^{\fC'}[c_2 \cdots c_\ell\, c]$, and $\xi = \zeta \cup \eta^+ \cup \partial \xi$, we have
	$\atp^{\fC^+}[\bar{c}c]= \xi$. Moreover, since $\xi \models \gamma_i$, $c$
	is an $i$th witness for $\phi$ with respect to the tuple $\bar{c}$. Finally, since $\xi \models \hat{\beta}$ the adjacent $(\ell{+}1)$-types
	assigned to the tuples $\bar{c}c$ and $c\tilde{c}$ satisfy $\beta$.
%
	Still keeping
	$\bar{c}$ fixed for the moment, we may carry out the above procedure for all $i \in I$. 
	To see that these assignments do not interfere with each other, we invoke property {(w1)} of the functions~$w_i$.

	Now make these assignments as just described for \textit{each} $\ell$-tuple $\bar{c}$ over $C$. To ensure that these assignments do not interfere with each other,
	we make use of properties {(w1)} and {(w3)} of the functions $w_i$. If $\bar{d}$ is an $m$-tuple
	that has been assigned (or not) to the extensions of various predicates by the process described above,
	then the two primitive generators of $\bar{d}$ must be of the form $\bar{c}c$ and $c\tilde{c}$, where $c = w_{i}(\bar{c})$ for some $i \in I$.
	Since primitive generators are unique up to reversal by Theorem~\ref{lemma:main}, it suffices to show that, for distinct pairs $(\bar{c},i)$ and
	$(\bar{c}',i')$, the corresponding $(\ell{+}1)$-tuples $\bar{c}\, w_i(\bar{c})$ and $\bar{c}'\, w_{i'}(\bar{c}')$ are not
	the same up to reversal. Now $\bar{c}\, w_i(\bar{c}) = \bar{c}'\, w_{i'}(\bar{c}')$ implies $\bar{c} = \bar{c}'$,  whence $i$ and $i'$ are distinct,
	whence $w_i(\bar{c}) \neq w_{i'}(\bar{c}')$ by ({w1}), a contradiction.
	On the other hand $\bar{c}'w_{i'}(\bar{c}') = w_i(\bar{c})\tilde{c}$ implies $\bar{c}'=w_{i}(\bar{c}), c_\ell \cdots c_2$,
	whence $w_{i'}(\bar{c}')$ does not occur in $\bar{c}$ (or therefore in $\tilde{c}$) by {(w3)}, again a contradiction.

	At this point, we have assigned a collection of tuples with primitive length $\ell{+}1$ to the extensions of predicates in $\tau$ so as to guarantee that  
	$\fC^+ \models \forall \bx_{\ell} \exists x_{\ell+1}\, \gamma_i$
	for all $i \in I$. In addition, no adjacent $(\ell{+}1)$-types thus defined violate $\forall \bx_{\ell+1}\, \beta$.
	It remains to complete the specification of $\fC^+$ by defining the adjacent types of all remaining primitive $\ell{+}1$-tuples, 
	and showing that, in the resulting structure, every $(\ell{+}1)$-tuple (primitive or not) satisfies $\beta$.
	Let $\bar{c}= c_1 \cdots c_{\ell+1}$ be a primitive $(\ell{+}1)$-tuple whose adjacent type has not yet been defined. Let
	$\zeta = \atp^\fC[c_1 \cdots c_{\ell}]$ 
	and $\eta= \atp^\fC[c_2 \cdots c_{\ell+1}]$. 
	Since $\bar{c}$ is primitive, it has no defects corresponding to even-length palindromes, and indeed
	no
	defects of the forms $\langle 1, j\rangle$ or $\langle j, \ell+1\rangle$. Letting $D$ be
	the defect set of $c_2 \cdots c_\ell$, therefore, we have $D \subseteq \bD^{\mathrm{o}}_{\ell{-}1}$; moreover,
	the defect set of $\bar{c}$ is $D^+$. 
	Further, since $\fB \models \psi_1 \wedge \psi_2$, we have $\fB \models \delta_{D}[c_2 \cdots c_\ell]$ and 
	$\fB \models p_\zeta[c_2 \cdots c_{\ell}]$. And since $\fB \models \psi_4$, we see that $\eta = \atp^{\fC'}[c_2 \cdots c_{\ell{+}1}]$ 
	is one of the disjuncts in the consequent of $\psi_4$. Thus, there exists an adjacent $(\ell{+}1)$-type $\xi$ such that 
	$\xi \models (\zeta \wedge \hat{\beta} \wedge \eta^+)$
	and $\partial \xi$ is $D^+$-compatible.
	By Lemma~\ref{lma:coherencyAndConsistency}, therefore,  
	there is a structure over the elements of $\bar{c}$ in which $\bar{c}$ has the incremental type $\partial \xi$. Thus we may consistently 
	fix $\itp^{\fC^+}[\bar{c}]$ to be $\partial \xi$. 
	Since $\zeta = \atp^{\fC'}[c_1 \cdots c_\ell]$, $\eta = \atp^{\fC'}[c_2 \cdots c_{\ell{+}1}]$, and $\xi = \zeta \cup \eta^+ \cup \partial \xi$, we have
	$\atp^{\fC^+}[\bar{c}]= \xi$. And since $\xi \models \hat{\beta}$, the adjacent $(\ell{+}1)$-types assigned to the tuples $\bar{c}$ and $\tilde{c}$ satisfy $\beta$. 
	Carrying this procedure out for all remaining primitive $(\ell{+}1)$-tuples,
	we obtain a layered structure $\fC^+$ of depth $\ell{+}1$. Let $\bar{d}$ be any $(\ell{+}1)$-tuple
	of elements from $C$.
	If $\bar{d}$ is primitive, then we have just ensured that $\fC^+ \models \beta[\bar{d}]$. If, on the other hand, 
	$\bar{d} = \bar{e}^{f}$ for some $k$-tuple $\bar{e}$ and some $f \in \bA^{\ell+1}_{k}$, where $k \leq \ell$, then, since
	$\fC \models \acl(\phi)$, we have $\fC \models \beta^f[\bar{e}]$ and hence, by Lemma~\ref{lma:essentiallyTrivial}, 
	$\fC \models \beta[\bar{d}]$. 
	This completes the construction of $\fC^+$. We have shown that $\fC^+ \models \phi$.
\end{proof}
Let us take stock. 
Since $\AFv{2}$ is included in $\FOt$, we know that 
each satisfiable $\AFv{2}$-sentence $\phi$ has a finite model of size $2^{O(\sizeof{\phi})}$.
Further, given any equality-free normal-form $\AFv{\ell{+}1}$-formula $\phi$ ($\ell \geq 2$),
we can construct an equisatisfiable equality-free $\AFv{\ell}$-formula $\psi$, also---modulo trivial logical re-arrangement---in normal form. We showed in Lemma~\ref{lma:reductionDirection1} that,
if $\phi$ is satisfiable over some domain, then $\psi$ is satisfiable over the same domain. 
And we showed in Lemma~\ref{lma:reductionDirection2} that,
if $\psi$ is satisfiable over some finite domain, then $\phi$ is satisfiable over a domain larger by a factor of at most $|I| \cdot (\ell^2+\ell+1)^\ell$.
Moreover, a simple check shows that $\sizeof{\psi}$ is $2^{O(\sizeof{\phi})}$. 
It follows by a routine (if slightly fiddly)  induction that any equality-free, satisfiable normal-form formula of $\AFv{\ell}$ ($\ell \geq 2$)
has a finite model of size $\ft(\ell{-}1,O(\sizeof{\varphi}))$. We omit a detailed proof, because, 
in the next section, we strengthen the above results in two ways.
First, we consider normal-form $\AFv{\ell{+}1}$-formulas $\phi$ involving equality. Using a more elaborate definition of $\psi \in \AFv{\ell}$, we prove
analogues of Lemmas~\ref{lma:reductionDirection1} and ~\ref{lma:reductionDirection2}. Secondly, as a by-product, we obtain tighter size bounds for the latter, showing that,
if $\psi$ is satisfiable over some domain, then $\phi$ is satisfiable over the same domain. This simplifies the proof of the eventual bound on model-sizes,
which we give in full in Theorem~\ref{theo:upperBoundEquality}.


\section{Adding Equality}
\label{section:AF-upper-boundsEq}

\begin{table}[b]
	\begin{tabular}{c | l}
		\multicolumn{2}{l}{Construction of $\psi$ from $\phi$ in the presence of $=$}\\
		\hline
		$G_{\bar{a}}$ & directed graph $(A, D_{\bar{a}} {\cup} E_{\bar{a}})$ of witnesses in $\fA$ around $\bar{a}$ \\
		$\sigma$ & a star; a function mapping witness $(\ell+1)$-types to colours $C$ \\
		$\atp(\sigma)$ & the underlying $\ell$-type of $\sigma$ \\
		$\St$ & the set of all stars $\sigma: \Atp_{\ell+1}^\tau \hookrightarrow C$ \\
		$\tl(\xi)$ & the $\ell$-type $\eta$ s.t. $\xi \models \eta^+$ \\
		$s_\sigma(\bx_\ell)$ & atom implying $\bx_\ell$ realises $\sigma$ \\
		$q_{\sigma, c}(\bx_{\ell-1})$ & atom implying there is some $x$ s.t. $x\bx_{\ell-1}$ realise $\sigma$ and colour $c$ \\
		$r_{\sigma, c, \xi}(\bx_\ell)$ & atom implying $x_\ell$ is the $\xi$-witness for $x\bx_{\ell-1}$ realising $\sigma$ and $c$
	\end{tabular}
	\caption{Quick reference guide for Sec~\ref{section:AF-upper-boundsEq}.}
	\label{table:equality}
	\end{table}

In this section, we complete the task of establishing a small model property for each of the fragments $\AFv{\ell}$, for $\ell \geq 2$, using the same strategy as employed in Sec.~\ref{section:AF-upper-bounds} for the equality-free sub-fragments.
The difference is that, when equality is present,
Lemma~\ref{lma:cartesian} becomes invalid. This lemma was used in the proof of Lemma~\ref{lma:reductionDirection2}, which constructed a 
model of the equality-free normal-form $\AFv{\ell+1}$-formula $\phi$ from a model of the equality-free $\AFv{\ell}$-formula $\psi$: by duplicating elements in the model of $\psi$, we could
easily avoid clashes when selecting elements to serve as witnesses for the existential requirements of $\phi$. When equality is present,
such duplication is no longer available, thus necessitating some additional combinatorial man\oe{}uvres.
Table~\ref{table:equality} provides a guide to important notions in this section. (Do note that they are, as of yet, undefined.)

Suppose $\fA \models \varphi$, where $\varphi$ is an normal form $\AFv{\ell+1}$-formula ($\ell \geq 2$) over some signature $\tau$. 
We fix $\phi$ and $\tau$ for the remainder of this section, 
writing the former as in~\eqref{eq:anf}, again repeated here for convenience:
\begin{equation*}
	\bigwedge_{i \in I} \forall \bx_{\ell} \exists x_{\ell+1}\, \gamma_i \wedge 
	\forall \bx_{\ell+1}\,
	\beta.
\end{equation*} 
We employ the letters $\ell$, $I$, $\beta$ and $\gamma_i$ with these denotations, and we assume without loss of generality that $I$
is non-empty. We proceed to define an 
expansion $\fA^+$ of $\fA$, and simultaneously, a normal-form $\AFv{\ell}$-formula $\psi$ over the expanded signature, such that $\fA^+ \models \psi$;
we later show that any layered model of $\psi$ (having height $\ell$) has a $\tau$-reduct that can be elevated to a model of $\phi$. 

We take $\psi$ to have the form
\begin{equation*}
	\psi\coloneqq \acl(\phi) \wedge \psi_0 \wedge \cdots \wedge \psi_5,
\end{equation*}
where $\acl(\phi)$ is the adjacent closure of $\phi$ (featured in Lemma~\ref{lma:adjacentClosure}), and $\psi_0, \dots, \psi_5$ are $\AFv{\ell}$-formulas over an expanded signature, defined below. We proceed to consider the conjuncts $\psi_0, \dots, \psi_5$ in turn. 

\subsubsection*{The conjunct $\psi_0$}
The initial step in this process is rather elaborate, and has no analogue in Sec.~\ref{section:AF-upper-bounds}.
For any element $a \in A$ and any tuple $\bar{a} \in A^{\ell{-}1}$, let $B_{a\bar{a}}$ be a minimal set such that, for each $i \in I$, there is some
$b \in B_{a\bar{a}}$ with $\fA \models \gamma_i[a\bar{a}b]$. Since $\fA \models \phi$, such a set exists, and moreover $|B_{a\bar{a}}| \leq |I|$.
We call the elements of $B_{a\bar{a}}$ the \textit{witnesses with respect to} $a\bar{a}$.
By assumption of minimality of $B_{a\bar{a}}$, there are no two distinct elements $b, b' \in B_{a\bar{a}}$
such that $\atp^\fA_{\ell{+}1}(a\bar{a}b) = \atp^\fA_{\ell{+}1}(a\bar{a}b')$.
Thus, given an adjacent $(\ell{+}1)$-type $\xi$ such that $\atp^\fA_{\ell{+}1}(a\bar{a}b) = \xi$ for some $b \in B_{a\bar{a}}$,
we will call this $b$ \textit{the $\xi$-witness with respect to $a\bar{a}$}. (Of course, this notion depends on our particular choice of the
set  $B_{a\bar{a}}$.)

Having chosen the witnesses with respect to the various $\ell$-tuples over $A$, consider now any $(\ell{-}1)$-tuple $\bar{a}$ over $A$.  
Defining the sets of ordered pairs over $A$
\begin{align*}
	& D_{\bar{a}} \coloneqq \set{ (a, b) \mid a \neq b \text{ and } b  \in B_{a\bar{a}} }\\
	& E_{\bar{a}} \coloneqq \set{ (a, a') \mid a \neq a' \text{ and there is some } b \in B_{a\bar{a}} \text{ s.t. } a' \in B_{b\tilde{a}}}
\end{align*}
we let $G_{\bar{a}}$ be the directed graph with vertices $A$ and edges $D_{\bar{a}} \cup E_{\bar{a}}$.
Thus, in this directed graph, there is an edge from $a$ to every witness $b$ with respect to $a\bar{a}$ (except $a$ itself), and
an edge from any element of $a$ to any other element $a'$ if there exists a witness $b$ with respect to $a\bar{a}$ such that 
$a'$ is a witness with respect to $b\tilde{a}$.
Since no $\ell$-tuple has more than $I$ witnesses,  
the out-degree of any vertex in $G_{\bar{a}}$ is at most $|I|^2 + |I|$. 
Recall that a $k$-\textit{colouring} of a directed graph $G_{\bar{a}} = (V, E)$ is a function $f\colon V \rightarrow [0,k{-}1]$ satisfying $f(u) \neq f(v)$ for all $(u,v) \in E$. It is well-known that any directed graph with maximum out-degree $d$ has a $(2d+1)$-colouring 
(see e.g. \cite[p.~612]{PrattHartmann23}). Thus,
we may colour the directed graph $G_{\bar{a}}$ with colours from a set $C$ of cardinality $2(|I|^2 + |I|) + 1$. For
every $(\ell{-}1)$-tuple $\bar{a}$, then, let some such colouring 
of $G_{\bar{a}}$ be fixed.

The importance of this colouring will become clearer as the proof unfolds. For the present, however, we note the following:
(i) if $b$ is one of the witnesses with respect to $a\bar{a}$, and is distinct from $a$, then $a$ and $b$ are differently coloured in $G_{\bar{a}}$;
(ii) if, in addition, $a'$ is one of the witnesses with respect to $b\tilde{a}$, and is distinct from $a$, then $a$ and $a'$ are differently coloured in $G_{\bar{a}}$.
It does not follow that the various witnesses with respect to $a\bar{a}$ will be differently coloured from each other in $G_{\bar{a}}$.
Notice also that, for distinct $(\ell{-}1)$-tuples $\bar{a}$ and $\bar{a}'$, an element $a$ might be coloured differently in the graphs $G_{\bar{a}}$ and $G_{\bar{a}'}$; this is true even if $\bar{a}' = \tilde{a}$.
%
%

Now let us treat the elements of $C$ as $\ell$-ary predicates, which we interpret in our expansion $\fA^+$ of $\fA$. Specifically, for all 
$a \in A$, $\bar{a} \in A^{\ell{-}1}$ and $c \in C$, we declare that
$\fA^+ \models c[a\bar{a}]$ just in case $a$ is assigned the colour $c$ in 
the colouring of $G_{\bar{a}}$, and we  write $\smash{\col^{\fA^+}[a\bar{a}]}$
to denote $c$. 
Note that, defining
\begin{equation*}
	\psi_0:= \forall \bx_\ell \left( \bigvee_{c \in C} c(\bx_\ell) \wedge \bigwedge_{c, c' \in C}^{c \neq c'} \Big( \neg c(\bx_\ell) \vee \neg {c'}(\bx_\ell) \Big)\right),
\end{equation*}
we have $\fA^+ \models \psi_0$. Conversely, if $\fB$ is any model of $\psi_0$, every $\ell$-tuple over $B$ is assigned a unique colour from $C$ in the obvious way. 

\subsubsection*{The conjunct $\psi_1$}
Here we can simply repeat material from Sec.~\ref{section:AF-upper-bounds}.
For each $s$ ($2 < 2s + 1 \leq \ell$) we introduce a fresh $(2s{+}1)$-ary predicate $d_{2s{+}1}$, and declare that a $(2s{+}1)$-tuple over $A$ 
satisfies $d_{2s+1}$ in the expansion $\fA^+$ just in case it is a palindrome. In addition, we define the $\AFv{\ell}$-formula
\begin{equation*}
	\psi_1 \coloneqq \bigwedge_{2 < 2s + 1 \leq \ell} \forall \bx_{s+1}\ d_{2s{+}1}(x_1 \cdots x_s x_{s + 1} x_s \cdots x_1).
\end{equation*}
Thus, $\fA^+ \models \psi_1$, and, conversely, if $\fB$ is any structure such that $\fB \models \psi_1$, and $\bar{c} \in B^{2s{+}1}$ is a palindrome
($2 < 2s + 1 \leq \ell$), then $\fB \models d_{2s+1}[\bar{c}]$. As before, we write 
$\delta_D:= \bigwedge \set{d_{j{-}i{+}1}(x_i \cdots x_j) \mid \langle i,j \rangle \in D}$ 
where $D$ is any set of pairs of integers $\langle i, j\rangle$
for $1 \leq i < j \leq \ell$ with $j{-}i{+}1$ odd and greater than 2.

\subsubsection*{The conjunct $\psi_2$}
The treatment of this conjunct is more elaborate than in Sec.~\ref{section:AF-upper-bounds}, and involves the colouring of $\ell$-tuples encountered above. 
However, the essential function of securing witnesses and imposing universal constraints is the same.
We remind the reader that, if $\bar{a}$ is a tuple, then $\tilde{a}$ denotes its reversal.
Observe 
that each $\ell$-tuple $a\bar{a}$ over $A$ gives rise to a partial function $\sigma \colon \Atp^{\tau}_{\ell{+}1}
\hookrightarrow C$ mapping adjacent $(\ell{+}1)$-types over $\tau$ to colours, namely,
\begin{equation*}
\sigma(\xi) = 
\begin{cases}
\col^{\fA^+}[b\tilde{a}] & \text{if $b$ is the (unique) $\xi$-witness in $\fA$ with respect to $a\bar{a}$,}\\
\text{undefined} & 	\text{if $a\bar{a}$ does not have a $\xi$-witness in $\fA$}.
\end{cases}
\end{equation*}
Thus, for every adjacent $(\ell{+}1)$-type $\xi$ for which
$a\bar{a}$ has a $\xi$-witness $b$ in $\fA$, the value $\sigma(\xi)$ tells us the colour of the `reversed' $\ell$-tuple $b\tilde{a}$ in $\fA^+$.
We denote the domain of $\sigma$ (i.e. the set of $\xi \in \Atp^\tau_{\ell{+}1}$ for which $\sigma(\xi)$ is defined) by $\dom(\sigma)$. Since
$I$ is non-empty, the $\ell$-tuple $a\bar{a}$ has at least one witness, whence $\dom(\sigma)$ is also non-empty.
We call $\sigma$ the \textit{star of} $a\bar{a}$ in $\fA^+$, and denote it $\st^{\fA^+}[a\bar{a}]$. Always remember that the elements
of $\Atp^{\tau}_{\ell{+}1}$ are adjacent types over the {\em original} signature $\tau$ of $\fA$, not over any expanded signature.

A simple check with reference to the formula $\phi$ given above verifies that $\sigma$ satisfies the following conditions:
\begin{enumerate}
	\item there exists a (unique) adjacent $\ell$-type $\zeta$ over $\tau$ such that, for all $\xi \in \dom(\sigma)$, we have $\xi \models \zeta$.
	\item for every $\xi \in \dom(\sigma)$ there exists $i \in I$ such that $\xi \models \gamma_i$;
	\item for every $i \in I$ there is exactly one $\xi \in \dom(\sigma)$ such that $\xi \models \gamma_i$; and
	\item for every $\xi \in \dom(\sigma)$, $\xi \models \hat{\beta}$ (remember that $\hat{\beta}:= \beta \wedge \beta^{-1}$).
\end{enumerate}
Condition 1 is verified by setting $\zeta = \atp^\fA[a\bar{a}]$; it is obvious that this is the unique adjacent $\ell$-type with the required properties;
we call $\zeta$ the \textit{underlying adjacent type} of $\sigma$ and denote it $\atp(\sigma)$.
The remaining conditions are immediate from the properties of witnesses.
Accordingly, we call any partial function $\sigma \colon \Atp^{\tau}_{\ell{+}1}
\hookrightarrow C$ satisfying conditions 1--4 above a {\em star}, and we denote the set of all such stars as $\St$. Thus, for any $(\ell{+}1)$-tuple over 
$A$, we have  $\st^{\fA^+}[a\bar{a}] \in \St$. We remark that the notion of a {star} depends on the formula $\phi$ (and on the parameters $\ell$, $C$,  $\beta$, $I$, $\gamma_i$ and $\tau$ associated with $\phi$). Since, however, $\phi$ may be considered fixed for the present, we suppress these parameters to avoid notational clutter.

Now we are ready to fix some additional predicates in the expansion $\fA^+$. 
Since we have settled the interpretations of the predicates $C$ in $\fA^+$, the adjacent star-type
$\st^{\fA^+}[a\bar{a}]= \sigma$ of any $\ell$-tuple $a\bar{a}$
is unaffected by these additional
predicates: in particular, the domain of $\sigma$
consists of adjacent $(\ell{+}1)$-types $\xi$ over the \textit{original signature} $\tau$; and the values  
$\sigma(\xi)$ are determined by the interpretations of 
the predicates $C$ in $\fA^+$.
With this in mind, for every $\zeta \in \Atp^\tau_{\ell}$, we introduce the $(\ell{-}1)$-ary predicate $p_\zeta$
familiar from Sec.\ref{section:AF-upper-bounds}, declaring $\fA^+ \models p_\zeta[\bar{a}]$ just in case, for some $a \in A$,
$\atp^\fA[a\bar{a}] = \zeta$. Thus,
$\fA^+ \models \psi_{2,0}$, where
\begin{align*}
	\psi_{2,0} &:= \bigwedge_{\zeta \in \Atp^\tau_{\ell}} \forall \bx_\ell \left( \zeta \rightarrow p_\zeta(x_2 \cdots x_\ell) \right).
\end{align*}
In addition, for every $\sigma \in \St$, we
introduce a new $\ell$-ary predicate $s_\sigma$,
and declare $\fA^+ \models s_\sigma[a\bar{a}]$ if and only if $\st^{\fA^+}[a\bar{a}] = \sigma$,
for any $\ell$-tuple $a\bar{a}$ over $A$. 
It is then easy to verify that 
$\fA^+ \models \psi_{2,1} \wedge \psi_{2,2}$, where
\begin{align*}
\psi_{2,1} &:= \forall \bx_\ell \bigvee_{\sigma \in \St} s_\sigma(\bx_\ell)  \wedge \forall \bx_\ell
 	\bigwedge_{\sigma, \sigma' \in \St}^{\sigma \neq \sigma'} \left( \neg s_\sigma(\bx_\ell) \vee \neg s_{\sigma'}(\bx_\ell) \right)\\
\psi_{2,2} & := \forall \bx_\ell \bigvee_{\sigma \in \St} \left(s_\sigma(\bx_\ell) \rightarrow \atp(\sigma)\right).
\end{align*}
Indeed, the first of these formulas states that any $\ell$-tuple satisfies exactly one of the predicates $s_\sigma$, and the second, that its adjacent type in $\fA$ is given by $\atp(\sigma)$.
(Recall in this connection that $\atp(\sigma)$ is an $\ell$-type, and thus a formula with variables $\bx_\ell$.)
Further, for every $\sigma \in \St$ and every $c \in C$, we introduce
a new $(\ell{-}1)$-ary predicate $q_{\sigma,c}$, and declare $\fA^+ \models q_{\sigma,c}[\bar{a}]$ just in case there is some $a \in A$ such that  $\st^{\fA^+}[a\bar{a}] = \sigma$ and $\col^{\fA^+}[a\bar{a}] = c$.
Thus, $q_{\sigma,c}$ identifies tails of $\ell$-tuples whose star in $\fA^+$ is $\sigma$ and whose colour in $\fA^+$ is $c$. It is thus immediate that $\fA^+ \models \psi_{2,3}$, where
\begin{equation*}
	\psi_{2,3}:= \forall \bx_\ell \bigwedge_{\sigma \in \St}
	\bigwedge_{c \in C}
	\Big( \big( s_\sigma(\bx_\ell) \wedge c(\bx_\ell) \big) \rightarrow q_{\sigma, c}(x_2 \cdots x_\ell) \Big).
\end{equation*}

Still
proceeding with the construction of $\psi_2$, for every $\sigma \in \St$, every $c \in C$, and every $\xi \in \dom(\sigma)$, we introduce
a new $\ell$-ary predicate $r_{\sigma,c,\xi}$, and fix its interpretation in $\fA^+$ as follows. Take any $(\ell{-}1)$-tuple $\bar{a}$ over $A$. If, on the one hand,
$\fA^+ \models q_{\sigma,c}[\bar{a}]$, then select some $a \in A$ such that $\st^{\fA^+}[a\bar{a}] = \sigma$ and $\col^{\fA^+}[a\bar{a}] = c$. By the interpretation of $q_{\sigma,c}$ in
$\fA^+$, this is possible. Now, for each $\xi \in \dom(\sigma)$, let $b$ be the $\xi$-witness for $a\bar{a}$, and set  
$\fA^+ \models r_{\sigma,c,\xi}[\bar{a}b]$. If, on the other hand, $\fA^+ \not \models q_{\sigma,c}[\bar{a}]$, then do nothing in respect of the tuple $\bar{a}$.
By carrying out this procedure for every $(\ell{-}1)$-tuple $\bar{a}$ over $A$, we thus fix the extensions of the predicates $r_{\sigma,c,\xi}$ in $\fA^+$.
Informally, it helps to read the atom $r_{\sigma, c, \xi}(\bx_\ell)$ as ``$x_\ell$ wants to be the $\xi$-witness with respect to the tuple
$x'\bx_{\ell{-}1}$ {\em for a particular} $x'$ such that $x'\bx_{\ell{-}1}$ has star-type $\sigma$ and colour $c$''. 
Since $\xi \in \dom(\sigma)$ implies
that a $\xi$-witness exists, we have $\fA^+ \models \psi_{2,4}$, where
\begin{equation*}
	\psi_{2,4} := \forall \bx_{\ell{-}1} \exists x_\ell
	\bigwedge_{\sigma \in {\St}}
	\bigwedge_{c \in C}
	\bigwedge_{\xi \in {\dom}(\sigma)}
	\left( q_{\sigma, c}(\bx_{\ell-1}) \to r_{\sigma, c, \xi}(\bx_{\ell})  \right).
\end{equation*}

Let us return to our $\ell$-tuple $a\bar{a}$ with $\st^{\fA^+}[a\bar{a}] = \sigma$, $\col^{\fA^+}[a\bar{a}] = c$ and $\xi$-witness $b$.
Let $\tl(\xi)$ denote the unique adjacent $\ell$-type $\eta$ over $\tau$ such that $\xi \models \eta^+$. Thus,
$\atp^{\fA}[\bar{a}b] = \tl(\xi)$.
Moreover, from the fact that $\st^{\fA^+}[a\bar{a}] = \sigma$ and $b$ is the $\xi$-witness with respect to $a\bar{a}$, we have 
$\col^{\fA^+}[b\tilde{a}] = \sigma(\xi)$. Thus, $\fA^+ \models \psi_{2,5}$, where
\begin{equation*}
		\psi_{2,5} \coloneqq \forall \bx_{\ell{-}1}
		\bigwedge_{\sigma \in \St} \bigwedge_{c \in C} \bigwedge_{\xi \in \dom(\sigma)}
		     \left( r_{\sigma, c, \xi}(\bx_\ell) \to \tl(\xi) \wedge (\sigma(\xi))(x_\ell \cdots x_1)\right).
\end{equation*}
Recall in this regard that $\tl(\xi)$ is an adjacent $\ell$-type (hence a formula with free variables $\bx_\ell$), and $\sigma(\xi) \in C$ is an $\ell$-ary predicate of $\tau$.

More can be said considering the $\ell$-tuple $\bar{a}b$ and its reversal, $b\tilde{a}$.
This latter tuple has some star in $\fA^+$, say $\st^{\fA^+}[b\tilde{a}] = \sigma'$. Consider, then, any adjacent $(\ell+1)$-type $\xi' \in \dom(\sigma')$.
Thus, $b\tilde{a}$ has a $\xi'$-witness in $\fA$, say $b'$.
If $b'=a$, then $\atp^\fA[a\bar{a}b]$ and $\atp^\fA[b\tilde{a}b']$ are mutually inverse $(\ell{+}1)$-types, and we have $\xi' = \xi^{-1}$. If, on the other hand, 
$b' \neq a$, then $\langle a,b'\rangle$ is an edge of the directed graph $G_{\bar{a}}$, whence $\col^{\fA^+}[a\bar{a}] \neq \col^{\fA^+}[b'\bar{a}]$, by construction of $\fA^+$. Thus, $\fA^+ \models \psi_{2,6}$, where
\begin{equation*}
	\begin{split}
		\psi_{2,6} \coloneqq \forall \bx_\ell
		\bigwedge_{\sigma  \in \footnotesize{\St}} & \bigwedge_{c \in C} \bigwedge_{\xi \in \footnotesize{\dom}(\sigma)}
		\Big(r_{\sigma, c, \xi}(\bx_\ell) \to  \bigvee\{ s_{\sigma'}(x_\ell \cdots x_1) \mid \\ & \sigma' \in \St \text{ and, for all $\xi' \in \dom(\sigma')$, }
		\xi' \neq \xi^{-1} \Rightarrow \sigma'(\xi') \neq c \} \Big). 
	\end{split}
\end{equation*}
That is to say: if $b$ is the $\xi$-witness with respect to $a\bar{a}$ and $b'$ the $\xi'$-witness with respect to $b\tilde{a}$, then
either $\xi$ and $\xi'$ are mutually inverse types in $\fA$ or else the tuples $a\bar{a}$ and $b'\bar{a}$
are differently coloured in $\fA^+$. 

We complete the construction of $\psi_2$ with a somewhat simpler constraint concerning the predicates $r_{\sigma, c, \xi}$. Consider again any $(\ell{-}1)$-tuple $\bar{a}$ over $A$, and
suppose that, for some $\sigma \in \St$, some $c \in C$ and some pair of distinct $(\ell{+}1)$-types $\xi, \xi' \in \dom(\sigma)$, there exist elements $b$ and $b'$
such that $\fA^+ \models r_{\sigma, c, \xi}[\bar{a}b]$ and $\fA^+ \models r_{\sigma, c, \xi'}[\bar{a}b']$. From the construction of $\fA^+$, $b$ and $b'$ 
must have been chosen as $\xi$- and $\xi'$-witnesses, respectively, with respect to a tuple $a\bar{a}$ for some particular element $a$. It follows that
$b$ and $b'$ must be distinct. That is, $\fA^+ \models \psi_{2,7}$, where
\begin{equation*}
\psi_{2,7} \coloneqq \forall \bx_\ell  \bigwedge_{\sigma  \in \footnotesize{\St}} \bigwedge_{c \in C} \bigwedge^{\xi \neq \xi'}_{\xi, \xi' \in \footnotesize{\dom}(\sigma)}
		\big(r_{\sigma, c, \xi}(\bx_\ell) \to \neg r_{\sigma, c, \xi'}(\bx_\ell) \big).
\end{equation*}
Conversely, if $\fB$ is any structure making $\psi_{2,7}$ true, and $\bar{a}$ an $(\ell{-}1)$-tuple over $B$, then we cannot have the same element $b$ such that
$\fB \models r_{\sigma, c, \xi}[\bar{a}b]$ and $\fB \models r_{\sigma, c, \xi'}[\bar{a}b]$ for different $\xi$ and $\xi'$ in the domain of $\sigma$.

Setting $\psi_2 \coloneqq \psi_{2,0} \wedge \cdots \wedge \psi_{2,7}$, 
we have thus established that $\fA^+ \models \psi_2$.

\subsubsection*{The conjuncts $\psi_3$ and $\psi_4$}
Here, we can again recapitulate the ideas of Sec.~\ref{section:AF-upper-bounds}, though
in a slightly different guise. As before, we take
$\bD^{\mathrm{o}}_{k}$ to denote the set of all pairs $\langle i, j \rangle$ for $1  \leq i < j \leq k$ such that
$j{-}i{+}1$ is greater than 2 and odd. 
Fix some subset $D \subseteq \bD^{\mathrm{o}}_{\ell{-}1}$, and 
suppose $\bar{a}$ is an $(\ell{-}1)$-tuple over $A$ such that $\fA^+ \models \delta_D[\bar{a}]$. By construction of $\fA^+$, the defect set of $\bar{a}$ includes $D$, whence, for any elements $a, b \in A$, the defect set of $a\bar{a}b$ certainly includes $D^+$. It follows that, if $\sigma \in \St$, $c \in C$ and
$\xi \in \dom(\sigma)$, but with $\partial \xi$ not $D^+$-compatible,
then there cannot exist $b \in A$ such that $\fA \models r_{\sigma,c,\xi}[\bar{a}b]$. For otherwise,
by construction of $\fA^+$, there exists $a$ such that $\st^{\fA^+}[a\bar{a}] = \sigma$, and $b$ is the $\xi$-witness with respect to $a\bar{a}$, whence
$\atp^\fA[a\bar{a}b] = \xi$, contradicting the first statement of Lemma~\ref{lma:coherencyAndConsistency}. We have thus proved $\fA^+ \models \psi_3$, where
\begin{equation*}\label{eq:equalreduct5}
	\begin{split}
		\psi_3 \coloneqq
		\bigwedge_{D \subseteq \bD^{\mathrm{o}}_{\ell{-}1}} \hspace{-4mm} & \forall \bx_\ell\,
		   \big( \delta_{D} \rightarrow\\  
		    &  \bigwedge \set{\neg r_{\sigma, c, \xi}(\bx_\ell) \big)  \mid \text{$\sigma \in \St$, $c \in C$, $\xi \in \dom(\sigma)$, $\partial \xi$ not $D^+$-compatible}} \Big).
	\end{split}
\end{equation*} 
For $\psi_4$, we recapitulate the formula from Sec.~\ref{section:AF-upper-bounds}, namely
\begin{equation*}
	\psi_4 \coloneqq \begin{split}
	\smash{\bigwedge_{\smash{\zeta \in \Atp^\tau_\ell}} \bigwedge_{D \subseteq \bD^{\mathrm{o}}_{\ell{-}1}}}
	  \forall \bx_\ell\, \Big( & \big(\delta_D \wedge p_\zeta(\bx_{\ell{-}1}) \big) \to\\
	       & \bigvee \set{\eta \in \Atp^\tau_\ell \mid 
		\text{$(\zeta \wedge \hat{\beta} \wedge \eta^+)$ is $D^+$-consistent}}\Big).
\end{split}
\end{equation*}	
And by the same reasoning as in
Sec.~\ref{section:AF-upper-bounds}, we have $\fA^+ \models \psi_4$.

\subsubsection*{The conjunct $\psi_5$}
Our final conjunct again has no analogue in Sec.~\ref{section:AF-upper-bounds}, and concerns extra conditions we need to impose on certain \textit{non-primitive} $(\ell{+}1)$-tuples in models of $\psi$. We need to ensure that, when reconstructing a model of $\phi$ from such structures, we do not attempt to assign these
tuples incompatible adjacent types.
We require the following simple lemma regarding words.
Recall that $\vec{\bA}^\ell_k$ denotes the set of adjacent functions $f \colon [1,\ell] \rightarrow [1,k]$ such that $f(\ell) =k$.
We say that an $\ell$-tuple $\bar{b}$ is \textit{terminal} if it can be written $\bar{b} = \bar{d}^f$ for some $k$-tuple
$\bar{d}$ with $k < \ell$, and some $f \in \vec{\bA}^\ell_k$. As explained in Sec.~\ref{sec:preliminaries}, for any $f \in \bA^\ell_k$,
a leg of $f$ is a maximal interval $[i,j] \subseteq [1,\ell]$ such that $f(h{+}1){-}f(h)$ is constant for all $h$ ($i \leq h < j)$; such an interval
corresponds to a 
straight-line segment of the graph of $f$.
\begin{lemma} \label{lma:rightFoldedOrterminal}
	Let $\bar{b}$ be an $\ell$-tuple over some set and $b$ an element of that set. Then 
	at least one of the following holds: \textup{(i)}
	$\bar{b}b$ is primitive; \textup{(ii)} $b$ is the last element of $\bar{b}$; \textup{(iii)} $\bar{b}b$ has a suffix which is an odd-length,
	non-trivial palindrome;
	\textup{(iv)}
	$\bar{b}$ is terminal.
\end{lemma}
\begin{proof}
	Suppose that $\bar{b}b$ is not primitive. 
	If $\bar{b}$ has an immediately repeated letter, then it is certainly terminal: indeed $\bar{b}= \bar{a}cc\bar{d}$ is generated from $\bar{a}c\bar{d}$ 
	via a function $f \in \vec{\bA}^\ell_{\ell{-}1}$ which pauses for one step on the letter $c$.  Hence
	we may assume that there are no immediately repeated letters in $\bar{b}$. Furthermore, if $b$ is not the last element of $\bar{b}$, then the whole of $\bar{b}b$ has no immediately repeated letters. 
	Since $\bar{b}b$ is not primitive, we have $\bar{b}'b=\bar{c}^f$ for some $k$-tuple $\bar{b}'$ and some $f \in \bA^m_k$ with at least two legs (i.e. maximal strictly ascending or descending intervals). If the final leg of $f$ is shortest, then some suffix of $\bar{b}b$ is a non-trivial palindrome, and this palindrome must have odd length, since there are no immediately repeated letters. Otherwise 
	$\bar{b}b$ has either of the forms $c\bar{d}d\tilde{d}c\bar{e}b$ or $\bar{a}c\bar{d}d\tilde{d}c\bar{d}d\bar{e}b$, depending on whether the shortest leg is initial
	or internal. In the former case, $\bar{b} = (d\tilde{d}c\bar{e})^g$ for some final adjacent function $g$; in the latter, $\bar{b} =
	(\bar{a}c\bar{d}d\bar{e})^h$  for some final adjacent function $h$ In both cases, $\bar{b}$ is terminal. 
\end{proof}

We remark that, in the case were $\bar{b}b$ has a suffix which is an odd-length non-trivial palindrome, that palindrome may be the whole of $\bar{b}b$.
The cases of Lemma~\ref{lma:rightFoldedOrterminal} correspond to conditions on the adjacent type of the $(\ell{+}1)$-tuple in question.
Say that an
adjacent $(\ell{+}1)$-type $\xi$ is \textit{palindromic} if, for any $\AFv{[\ell{+}1]}$-atom $\alpha$ (over the relevant signature),  $\xi \models \alpha$ implies $\xi \models \alpha(x_{\ell{+}1} \cdots x_1)$. 
Evidently, if an $(\ell{+}1)$-tuple is a palindrome, then its adjacent type in any structure is palindromic. Say that $\xi$ is \textit{blunt} if, for any $\AFv{[\ell{+}1]}$-atom $\alpha$, $\xi \models \alpha$ implies $\xi \models \alpha(\bx_{\ell}x_{\ell})$. If an $(\ell{+}1)$-tuple has the same last two elements, then its adjacent type in any structure is blunt. 
Say that $\xi$ is $s$-{\em hooked} (for $s$ satisfying $2 <  2s+1 < \ell{+}1$) if, for any $\AFv{[\ell{+}1]}$-atom $\alpha$ we have that
$\xi \models \alpha$ implies $\xi \models \alpha(x_1 \cdots x_{\ell{+}1{-}s} x_{\ell{-}s} \cdots x_{\ell{+}1{-}2s})$.
If an $(\ell{+}1)$-tuple has a proper suffix that is a non-trivial palindrome of length $2s+1$, then its adjacent type in any structure is $s$-hooked. Observe the strict inequality $2s+1 < \ell{+}1$ governing the parameter $s$ in this last definition: if
$\xi$ is palindromic (and $\ell+1$ is odd), we do not say that $\xi$ is $(\ell/2)$-hooked.

Now let $\bar{a}$ be an $(\ell{-}1)$-tuple over $A$ and $b \in A$, and suppose $\fA^+ \models r_{\sigma,c,\xi}[\bar{a}b]$.
By the construction of $\fA^+$, there exists $a \in A$ such that $b$ is the $\xi$-witness with respect to $a\bar{a}$, and, moreover, $\col^{\fA^+}[a\bar{a}] = c$.
If $\xi$ is palindromic, blunt or $s$-hooked for some $s$ ($2 < 2s+ 1 < \ell{+}1$), then the $(\ell{+}1)$-tuple $a\bar{a}b$ exhibits certain properties, which we proceed to describe.
Consider first the case where $\ell{+}1$ is odd and $\xi$ is not palindromic, and suppose in addition that $\fA^+ \models d_{\ell-1}[\bar{a}]$. By the construction of $\fA^+$ again, $\bar{a}$ is a palindrome, i.e.~$\bar{a}= \tilde{a}$. It follows that $a \neq b$,
since otherwise, $a\bar{a}b$ is a palindrome, contradicting the assumption that $\xi = \atp^\fA[a\bar{a}b]$ is not palindromic. 
And since $b$ is a witness with respect to $a\bar{a}$ with $a \neq b$, the ordered pair $\langle a,b \rangle$ is an edge of the directed graph $G_{\bar{a}} = G_{\tilde{a}}$,
so that by the construction of $\fA^+$,  we have $c = \col^{\fA^+}[a\bar{a}] \neq \col^{\fA^+}[b\bar{a}] = \col^{\fA^+}[b\tilde{a}]$. 
Thus we have shown that $\fA^+ \models \psi_{5,1}$, where, for $(\ell{+}1)$ odd,
\begin{equation*}
	\psi_{5,1} := \forall \bx_\ell \bigwedge_{\sigma \in \St}
	\bigwedge_{c \in C} \bigwedge_{\xi \in \dom(\lambda)}^{\xi \text{ not palindromic}}
	\Big(r_{\lambda, c, \xi}(\bx_\ell) \wedge d_{\ell{-}1}(\bx_{\ell{-}1}) \to \neg c(x_\ell \cdots x_1) \Big),
\end{equation*}
and for $(\ell{+}1)$ even, $\psi_{5, 1} := \top$. We remark  that even-length palindromes have immediately repeated letters in the middle, which 
obviates---as we shall see later---the need for an analogue of the odd-length case.

Second, 
consider the case where $\xi$ is not blunt. Here, we do not need to add any conjuncts to $\psi$, since the consistency of $\xi$ requires that 
it contains the inequality literal
$x_{\ell} \neq x_{\ell{+}1}$. Hence, $\atp^\fA[\bar{a}b]$ contains the inequality literal
$x_{\ell{-}1} \neq x_{\ell}$, which is all the information we shall require.

Third, 
consider the case where, for any $s$  in the range $2 <  2s+1 < \ell{+}1$, $\xi$ is not $s$-hooked. Since $\xi = \atp[a\bar{a}b]$, it follows that $a\bar{a}b$ has no proper suffix of length $2s+1$ that is a palindrome, and
therefore, $\bar{a}b$ has no suffix of the same length that is a palindrome. Thus we have shown that $\fA^+ \models \psi_{5,2}$, where
\begin{equation*}
	\psi_{5,2} := \forall \bx_\ell \bigwedge_{\sigma \in \St}
	\bigwedge_{c \in C} \bigwedge_{s= 1}^{s \leq (\ell{-}1)/2} \, \,\bigwedge_{\xi \in \dom(\lambda)}^{\xi \text{ not $s$-hooked}}
	\Big(r_{\lambda, c, \xi}(\bx_\ell)  \to  \neg d_{2s{+}1}(x_{\ell{-}2s} \cdots x_{\ell}) \Big).
\end{equation*}
Writing $\psi_5:= \psi_{5,1} \wedge \psi_{5,2} \wedge \psi_{5,2}$, we have shown that $\fA^+ \models \psi_5$. This completes the definition of the formula $\psi$.

Summarizing the above discussion, and recalling that, by Lemma~\ref{lma:adjacentClosure}, $\phi \models \acl(\phi)$, 
we have:
\begin{lemma}
	Suppose $\fA \models \phi$. Then we can expand $\fA$ to a model $\fA^+ \models \psi$.
	\label{lma:reductionDirection1eq}
\end{lemma}
Having defined $\psi$ and established Lemma~\ref{lma:reductionDirection1eq}, we establish a converse
in the form of the following lemma.
\begin{lemma}
Suppose $\fA \models \psi$. Then we can construct a model $\fB \models \phi$ over the same domain. 
\label{lma:reductionDirection2eq}
\end{lemma}
\begin{proof}
Since $\psi \in \AFv{\ell}$, by Lemma~\ref{lma:approx} we may take $\fA$ to be
a layered structure of the height $\ell$.
Setting $\fA^-$ to be the contraction of $\fA$ to the signature $\tau$ (the original signature of $\phi$), we proceed to elevate $\fA^-$
to a model $\fB \models \phi$ (with the same domain, $A$) by setting the interpretations of the predicates in $\tau$ with respect to the primitive $(\ell{+}1)$-tuples over $A$.

Let $a$ be any element of $A$ and $\bar{a}$ any $(\ell{-}1)$-tuple over $A$. By $\psi_0$, $a\bar{a}$ has a unique colour, say $\smash{c = \col^{\fA}[a\bar{a}]}$, 
and  by $\psi_{2,1}$, there is a unique $\sigma \in \St$ such that $\fA \models s_\sigma[a\bar{a}]$. Let $\zeta = \atp(\sigma)$ be the underlying $\ell$-type of $\sigma$, and let us fix, for the moment, any $\xi \in \dom(\sigma)$. Thus, $\zeta= \xi\restr_{[1,\ell]}$.  By $\psi_{2,2}$, $\atp^{\fA^-}[a\bar{a}] = \zeta$.
Moreover, by $\psi_{2,3}$, $\fA \models q_{\sigma,c}[\bar{a}]$, and hence, by the fact that  $\fA \models \psi_{2,4}$, there exists $b \in A$ such that
$\fA \models r_{\sigma,c,\xi}[\bar{a}b]$. Let $\eta = \tl(\xi)$; it follows from $\psi_{2,5}$ that $\atp^{\fA^-}[\bar{a}b]= \eta$ and
$\col^{\fA}[b\tilde{a}] = \sigma(\xi)$. The intention is that we should set the interpretations of 
the predicates in the structure $\fB$ in such a way that the $(\ell{+}1)$-tuple $a\bar{a}b$ is assigned the adjacent type $\xi$. The various other conjuncts of $\psi$ will ensure that this can be done consistently.

Suppose on the one hand that the $(\ell{+}1)$-tuple $a\bar{a}b$ is primitive. Thus, no prefix or suffix of
$a\bar{a}b$ is a non-trivial palindrome, and 
$a\bar{a}b$ contains no immediately repeated letters.
Hence, any defect $\langle i, j \rangle$ of $a\bar{a}b$
satisfies $2 \leq i < j \leq \ell$, with $j{-}i{+}1$ greater than 2 and odd. Letting $D$ now be the set of 
defects $\langle i, j \rangle$ of the $(\ell{-}1)$-tuple $\bar{a}$ 
(so that $D \subseteq \bD^{\mathrm o}_{\ell{-}1}$), we see that the set of defects of $a\bar{a}b$ is given by $D^+ = \set{\langle i+1, j+1 \rangle \mid \langle i,j \rangle \in D}$. By $\psi_1$, we have that $\fA \models \delta_D[\bar{a}]$; and by $\psi_3$, bearing in mind that
$\fA \models r_{\sigma,c,\xi}[\bar{a}b]$, we have that $\partial \xi$ is $D^+$-compatible. 
Lemma~\ref{lma:coherencyAndConsistency} thus ensures that
that it is meaningful to set $\itp^{\fB}_{\ell+1}(a\bar{a}b) = \partial \xi$. 
Moreover, since 
$\atp^{\fA^-}[a\bar{a}] = \zeta$,
$\atp^{\fA^-}[\bar{a}b] = \eta$, and $\xi = \zeta \cup \eta^+ \cup \partial \xi$, we have $\atp^{\fB}_{\ell+1}(a\bar{a}b) = \xi$.  
We remark that the adjacent type of the primitive $(\ell{+}1)$-tuple $b\tilde{a}a$ is also set in this process, but no other primitive $(\ell{+}1)$-tuples have their adjacent types set. Observe that, since
$\xi \models \hat{\beta}$, 
the newly-set adjacent types of $a\bar{a}b$ and $b\tilde{a}a$ will not violate $\beta$.

Suppose on the other hand that 
that $a\bar{a}b$ is not primitive. 
We show that, in that case, $\fA$ already provides a $\xi$-witnesses with respect to $a\bar{a}$, so that there is nothing to do.
Here, we make use of the fact that $\fA \models \acl(\phi) \wedge \psi_5$.
By Lemma~\ref{lma:rightFoldedOrterminal}, either $b$ is the last element of $a\bar{a}$, or $a\bar{a}b$ has a suffix 
that is a non-trivial odd-length palindrome, or $a\bar{a}$ is terminal.
The case where $a\bar{a}$ is terminal is easily dealt with. There exists a $k$-tuple $\bar{d}$ and $f \in \vec{\bA}^{\ell}_{k}$
such that $a\bar{a} = \bar{d}^{f}$ for some $k$ ($2 \leq k < \ell$). Defining, as before, $f^+ = f \cup \langle k{+}1, \ell{+}1\rangle$, 
by  $\acl(\phi)$ there exists, for each $i \in I$, some $b' \in A$ such that
$\fA \models \gamma_i[(\bar{d}b')^{f^+}]$, i.e.~$\fA \models \gamma_i[a \bar{a} b']$.
In effect, we are discarding our originally chosen element $b$, since the required witness is already present.  

Consider next the case where $b$ is the last element of $\bar{a}$. We have already argued that
$\atp^{\fA^-}[\bar{a}b]= \eta = \tl(\xi)$, whence $\tl(\xi) \models x_{\ell{-}1} = x_\ell$,
whence $\xi \models x_{\ell} = x_{\ell{+}1}$. Thus, $\xi$ is blunt, by consistency of $\xi$. But in that case, writing
$\bar{a} = a_1 \cdots a_{\ell{-}1}$, we have 
$\fA^- \models \zeta[a a_1 \cdots a_{\ell{-}1}]$ implies
$\fA^- \models \xi[a a_1 \cdots a_{\ell{-}1}a_{\ell{-}1}]$, that is, $\fA^- \models \xi[a\bar{a}b]$. Thus, 
our chosen element $b$ is already a $\xi$-witness with respect to $a\bar{a}$, without our having to do anything.

Consider finally the case where $a\bar{a}b$ is not terminal and has a suffix that is an odd-length, non-trivial palindrome.
Suppose, on the one hand, that the suffix in question is the \textit{whole} of $a\bar{a}b$---that is, $\ell{+}1$ is odd, and $a\bar{a}b$ has the form
$a a_1 \cdots a_{\ell/2} a_{\ell/2{-}1} \cdots a_1 a$.
Thus, $a = b$ and $\bar{a} = \tilde{a}$.
Hence, $a\bar{a} = b\tilde{a}$, whence $\col^{\fA^-}[b\tilde{a}] = \col^{\fA^-}[a\bar{a}] = c$. 
It follows by $\psi_1$ that $\fA \models d_{\ell{-}1}[\bar{a}]$, and thence, by $\psi_{5,1}$\linebreak
that $\xi$ is palindromic. 
Thus,
$\atp^{\fA^-}(a a_1 \cdots a_{\ell/2}) = \zeta\restr_{[1,\ell/2{+}1]} = \xi\restr_{[1,\ell/2{+}1]}$, whence 
$\atp^{\fA^-}(a a_1 \cdots a_{\ell/2} a_{\ell/2{-}1} \cdots a_1 a) = \xi$.
Again, our chosen element $b= a$ is already a $\xi$-witness with respect to $a\bar{a}$.

On the other hand, suppose that $a\bar{a}b$ has a \textit{proper} suffix that is a non-trivial, odd-length palindrome.
Thus, $a\bar{a}b$ is of the form $a a_1 \cdots 
a_{\ell{-}s} a_{\ell{-}(s{+}1)} \cdots a_{\ell{-}2s}$ (for some $s$ with
$2 <  2s+1 \leq \ell$).  
But by $\psi_1$, $\fA \models d_{s}[a_{\ell{-}2s} \cdots a_{\ell{-}s} a_{\ell{-}(s{+}1)} \cdots a_{\ell{-}2s}]$.
Hence, by $\psi_{5,2}$, we have that $\xi$ is $s$-hooked. 
As in the previous case, it follows that $b=a_{(\ell{-}2s)}$ is already a $\xi$-witness with respect to $a\bar{a}$. 
This completes the process of finding a $\xi$-witness with respect to $a\bar{a}$. 

Keeping $a$ and $\bar{a}$ fixed for the moment, carry out the above assignments for all $\xi \in \dom(\sigma)$, choosing, for each such $\xi$,
a $\xi$-witness $b_\xi$ with respect to $a\bar{a}$. It follows from $\psi_{2,7}$ that
the various elements $b_\xi$ must be distinct, so no clashes can arise.
This completes the process of finding all the required witnesses with respect to $a\bar{a}$. 
From the properties of the elements of $\St$, we see that, however $\fB$ is completed, the $\ell$-tuple $a\bar{a}$ will satisfy $\exists x_{\ell{+}1}\, \gamma_i$
 in $\fB$ 
for every $i \in I$. Moreover, each of the newly-fixed primitive $(\ell{+}1)$-tuples (and their reversals) satisfies $\beta$.

Now vary $a$ and $\bar{a}$ freely. We must show that, for $a'$ and $\bar{a}'$ with $a\bar{a} \neq a'\bar{a}'$, the chosen witnesses
$b$ and $b'$ can never lead to a clash. Suppose then that some $m$-tuple $\bar{d}$ is assigned (or not) to the extension of an $m$-ary predicate $p$ 
when defining the adjacent $(\ell{+}1)$-types of both of the primitive
tuples $a\bar{a}b$ and $a'\bar{a}'b'$.
Then $\bar{d}$ is generated by both $a\bar{a}b$ and $a'\bar{a}'b'$, and there exist adjacent $(\ell{+}1)$-types $\xi$ and $\xi'$ such that $b$ is selected as the $\xi$-witness with respect to $a\bar{a}$ and
$b'$ is selected as the $\xi$-witness with respect to $a'\bar{a}'$. Since, by assumption, $a\bar{a} \neq a'\bar{a}'$, Theorem~\ref{lemma:main} implies
$a\bar{a}b$ is the reversal of $a'\bar{a}'b'$, so that $a' = b$, $b' = a$ and $\bar{a}' = \tilde{a}$; henceforth, we shall write $b\tilde{a}$ in preference to $a'\bar{a}'$. It suffices to show that $\xi' = \xi^{-1}$, because in that case there can be no clash in the assignment of $\bar{d}$ to the extension of the predicate $p$.

Suppose, for contradiction that $\xi' \neq \xi^{-1}$. 
By $\psi_{2,1}$, let $\sigma$ and $\sigma'$ be the unique elements of $\St$ such that $\fA \models s_\sigma[a\bar{a}]$ and 
$\fA \models s_{\sigma'}[b\tilde{a}]$; and by $\psi_{0}$, let $c$ and $c'$ be the unique elements of $C$ such that $\fA \models c[a\bar{a}]$ and 
$\fA \models c'[b\tilde{a}]$. Since $b$ was chosen as a $\xi$-witness for $a\bar{a}$, we have $\fA \models r_{\sigma,c,\xi}[\bar{a}b]$, and similarly
$\fA \models r_{\sigma',c',\xi'}[\tilde{a}a]$. Recalling that $\fA \models r_{\sigma,c,\xi}[\bar{a}b]$ and $\xi' \neq \xi^{-1}$, we have by $\psi_{2,6}$ that
$\sigma'(\xi') \neq c$. 
Moreover, since $\fA \models r_{\sigma',c',\xi'}[\tilde{a}a]$, by $\psi_{2,5}$ we have $\fA \models (\sigma'(\xi'))[a\bar{a}]$,
contradicting $\fA \models c[a\bar{a}]$.
Thus $\xi' = \xi^{-1}$ as required.

This completes the process of finding witnesses for all $\ell$-tuples from $A$, in the course of which 
we have partially defined a layered structure $\fB$, such that, however it is completed,
$\fB \models \forall \bx_\ell \exists x_{\ell{+}1}\, \gamma_i$ for every $i \in I$.
Moreover, each of the newly-fixed 
primitive $(\ell{+}1)$-tuples (and their reversals) satisfies $\beta$.

All that remains is to complete the construction of $\fB$ by
assigning adjacent types to all primitive $(\ell{+}1)$-tuples whose adjacent types have not yet been fixed, without violating the condition
$\forall \bx_{\ell{+}1} \, \beta$. Suppose, then $\bar{c} = a\bar{a}b$ is such a primitive $(\ell{+}1)$-tuple.
By primitiveness, any defect $\langle i,j \rangle$ of $\bar{c}$ is such that $2 < i < j \leq \ell$ and $j{-}i{+}1$ is odd.
Thus, writing $D$ for the set of defects of $\bar{a}$, we have that 
$D \subseteq \bD^{\mathrm o}_{\ell{-}1}$, and
$D^+$ is exactly the defect set of $\bar{c}$.
By $\psi_1$ we have $\fA \models \delta_D[\bar{a}]$ and, 
writing $\zeta = \atp^{\fA^-}[a\bar{a}]$,
we have that $\fA \models p_\zeta[\bar{a}]$ by $\psi_{2, 0}$. 
And writing $\eta = \atp^{\fA^-}[\bar{a}b]$, by 
$\psi_4$ there exists an adjacent $(\ell{+}1)$-type $\xi$ such that $\xi \models \zeta \wedge \hat{\beta} \wedge \eta^+$
and $\partial \xi$ is $D^+$-compatible.   
By the second statement of Lemma~\ref{lma:coherencyAndConsistency}, then,
we can consistently assign $\itp^\fB[\bar{c}] = \partial \xi$.
Since $\zeta = \atp^{\fA^-}[a\bar{a}]$, $\eta = \atp^{\fA^-}[\bar{a}b]$, and $\xi = \zeta \cup \eta^+ \cup \partial \xi$, we have 
$\atp^\fB[\bar{c}] = \xi$.
Additionally, since $\xi \models \hat{\beta}$, both $\bar{c}$ and $\tilde{c}$ satisfy the universal requirements $\beta$ in $\fB$.
Repeated applications of this procedure result in all primitive $(\ell{+}1)$-tuples having their adjacent types defined in
such a way that
$\fB \models \forall \bx_{\ell{+}1} \, \beta$. Hence $\fB \models \phi$. 
\end{proof}

Taken together,
Lemmas~\ref{lma:reductionDirection1eq} and~\ref{lma:reductionDirection2eq} reduce the
satisfiability problem for $\AFv{\ell+1}$ to that for $\AFv{\ell}$~($\ell \geq 2$), though with exponential blow-up.
We thus obtain the decidability of satisfiability for the whole of $\AF$.
More precisely:
\begin{theorem}
	If $\phi$ is a satisfiable \AFv{\ell+1}-formula, with $\ell \geq 1$, then
	$\phi$ is satisfied in a structure of size at most $\ft(\ell, \sizeof{\phi}^{O(1)})$.
	Hence the satisfiability problem for \AFv{k} is in $(k-1)$-\NExpTime{} for all $k \geq 2$, and the adjacent fragment is $\Tower$-complete.
	\label{theo:upperBoundEquality}
\end{theorem}
\begin{proof}
	Fix $\ell \geq 2$ and suppose $\phi$ is a satisfiable $\AFv{\ell+1}$-formula
	over a signature $\tau$.
	By Lemma~\ref{lma:anf}, we may assume that $\phi$ is in normal form.
	Writing $\phi_{\ell+1}$ for $\phi$, let $\phi_{\ell}$ now denote the formula $\psi \coloneqq \acl(\phi) \wedge \psi_0 \wedge  \cdots \wedge \psi_5$
	as defined above. 
	Repeating this process, we obtain a
	sequence of formulas $\phi_{\ell+1}, \dots, \phi_2$. 
	By Lemma~\ref{lma:reductionDirection1eq}, $\phi_2$ is then also satisfiable.
	For all $k$, ($2 \leq k \leq \ell+1$), let $\phi_k$ have signature $\tau_k$, and for $k \leq \ell$, consider the construction of $\phi_{k}$ from
	$\phi_{k{+}1}$. Since $\sum_{k'=1}^{k{+}1} |\bA^{k{+}1}_{k'}|$ 
	is bounded by a constant, we see that $\sizeof{\acl(\phi)_{k{+}1}}$ is $O(\sizeof{\phi_{k{+}1}})$. 
	In regard to~$\psi_0, \dots, \psi_5$ we may,
	when considering the adjacent $(k{+}1)$-types over $\tau_{k{+}1}$, disregard all adjacent atoms whose argument sequence is not
	a substitution instance of some argument sequence $\bx_{k{+}1}^g$ occurring in an atom of $\phi_{k{+}1}$, as these cannot affect the
	evaluation of $\phi_{k{+}1}$. And since $k \leq \ell$, the 
	number of functions from $\bx_{k{+}1}$ to itself is again bounded by a constant,
	and the number of adjacent $(k{+}1)$-atoms over $\tau_{k{+}1}$ 
	that impact the satisfaction of $\varphi_{k{+}1}$ is $O(\sizeof{\phi_{k{+}1}})$.
	Thus, the number of adjacent $(k{+}1)$-types over $\tau_{k{+}1}$ that we need to consider
	is $2^{O(\sizeof{\phi_{k{+}1}})}$. (This cardinality bounds the number of predicates $p_\zeta$ as featured in $\psi_{2, 0}$ and $\psi_{4}$).
	Keeping the number of $(k{+}1)$-types we need to consider in mind recall that a star is a partial function mapping
	$(k{+}1)$-types to a set of colours $C$.
	Pick any such star $\sigma$ and let $m$ be the cardinality of $\dom(\sigma)$.
	Clearly, $m$ is at most $\sizeof{\varphi_{k{+}1}}$ as no tuple needs more witnesses than there are existential requirements.
	On the other hand, $|C|$ is also bounded by the number of existential requirements in $\varphi_{k{+}1}$ and is
	at most $2(\sizeof{\varphi_{k{+}1}}^2 + \sizeof{\varphi_{k{+}1}}) + 1$.
	Noting that there are at most $|C|^{m}$ functions over a domain of size $m$ and image $C$,
	we have that the total number of stars does not exceed $\sum_{m=0}^{\sizeof{\varphi_{k{+}1}}} (|C|^{m} \cdot (2^{O(\sizeof{\phi_{k{+}1}})})^{m})$
	which is clearly bounded by $2^{\sizeof{\phi_{k{+}1}}^{O(1)}}$.
	Taking these considerations together and noting that the set of defects $\bD^{\mathrm{o}}_{k{-}1}$ is bounded by the constant $(k-1)^2$,
	we have that each of the sentences $\psi_0, \dots, \psi_5$ contain at most $2^{\sizeof{\phi_{k{+}1}}^{O(1)}}$ conjuncts.
	Some care is needed when calculating the sizes of these conjuncts themselves, as they feature disjunctions over stars as in $\psi_{2, 6}$ or
	that over tailing ends of $D^+$-compatible $(k{+}1)$-types
	obtained by processing $\zeta \wedge \hat{\beta}$
	for some set of defects $D \subseteq \bD^{\mathrm{o}}_{k{-}1}$ as in $\psi_{4}$.
	However, in the first case, the disjunction rages over a maximum of $2^{\sizeof{\phi_{k{+}1}}^{O(1)}}$ predicates $s_{\sigma'}$, and in the second case
	the formulas are, in effect, in disjunctive normal form over atoms contained in $\phi_{k{+}1}$, and hence
	have cardinality $2^{O(\sizeof{\phi_{k{+}1}})}$. Hence, $\sizeof{\phi_k}$
	is $2^{\sizeof{\phi_{k{+}1}}^{O(1)}}$.
	By an easy induction, $\sizeof{\phi_{2}}$ is then $\ft(\ell-1, \sizeof{\phi_{\ell+1}}^{O(1)})$. 
	
	Since $\phi_2$ is a sentence in the two-variable fragment of first-order logic, we have, by \cite{GradelKV97},
	that $\phi_2$ has a model $\fA_2$ of cardinality $2^{O(\sizeof{\varphi_2})}$. 
	Moreover, by Lemma~\ref{lma:reductionDirection2eq}, each of the formulas $\phi_{k{+}1}$ ($2 \leq k \leq \ell$)  
	have a model over the same domain as $\fA_2$.
	Since $|\tau_k| \geq |\tau_{k{+}1}|$, we conclude that the expansion of $\fA_2$ (and subsequent models) into models
	$\fA_3, \dots, \fA_{\ell{+}1}$ of $\phi_3, \dots, \phi_{\ell+1}$
	remains bounded by $2^{O(\sizeof{\varphi_2})}$ and thus is of size $\ft(\ell, \sizeof{\phi_{\ell+1}}^{O(1)})$ as required.
\end{proof}

We note that one can do better if the equality predicate is disallowed.
As shown in \cite[Lemma~4.6]{bkp-h23}, $\AFv{3}$-sentences have a model exponential size in regard to the input.
Thus, by reducing variables of $\varphi_{\ell+1}$ as in Theorem~\ref{theo:upperBoundEquality}
until an equisatisfiable $\AFv{3}$-sentence $\varphi_3$ is reached, we
have that $\varphi_3$ (and thus also $\varphi_{\ell+1}$) has a model of size $\ft(\ell-1, \sizeof{\phi_{\ell+1}}^{O(1)})$
thus establishing the following:
\begin{theorem}
The satisfiability problem for the sub-fragment of \AFv{k} without equality is in $(k{-}2)$-\NExpTime{} for all $k \geq 3$.
\end{theorem}

It may not have escaped the reader's attention that the equality predicate barely features in the proof of Theorem~\ref{theo:upperBoundEquality}. This
is  because the reduction from the satisfiability problem for $\AFv{\ell{+}1}$ to that for $\AFv{\ell}$ effected in
Lemmas~\ref{lma:reductionDirection1eq}
and~\ref{lma:reductionDirection2eq}
does not interfere with predicates of arity 
less than $(\ell{+}1) \geq 3$, and does not interfere with the domain of quantification. The special status of the equality predicate is thus, in a sense, pushed back into the two-variable fragment.

This observation suggests a generalization.
Recent decades have witnessed concerted attempts to investigate the decidability of the satisfiability problem for $\FO^2$ over various classes of structures, where certain distinguished predicates are required to be interpreted as relations subject to various restrictions, e.g. as linear orders~\cite{KieronskiT09, SchwentickZ10}, trees~\cite{CharatonikKM14, BednarczykCK17}, equivalences~\cite{KieronskiT09, KieronskiO12} and more (see~\cite{KieronskiPT18} for a survey).
If such an extension of $\FOt{}$ is decidable for (finite) satisfiability, then the reduction outlined in this section applies to the corresponding 
extension of $\AF$---it being understood that the distinguished predicates in question have arity at most 2. On the other hand,
if the extension of $\FOt{}$ is undecidable for (finite) satisfiability, then---anticipating Theorem~\ref{theo:fo2-and-af-over-binary-sig-are-the-same}, which states that $\AF$ subsumes $\FOt$---this undecidability result likewise immediately transfers to $\AF$.


\newcommand{\atmM}{\mathcal{M}}
\newcommand{\nofM}{n}
\newcommand{\statesQ}{\mathrm{Q}}
\newcommand{\statesQexists}{\statesQ_{\exists}}
\newcommand{\statesQforall}{\statesQ_{\forall}}
\newcommand{\states}{\mathrm{s}}
\newcommand{\stateq}{\mathrm{q}}
\newcommand{\initialstate}{\states_{I}}
\newcommand{\acceptingstate}{\states_{A}}
\newcommand{\rejectingstate}{\states_{R}}
\newcommand{\transreldelta}{\mathrm{T}} 
\newcommand{\transitiont}{\mathrm{t}}
\newcommand{\letterzero}{\mathtt{0}}
\newcommand{\letterone}{\mathtt{1}}
\newcommand{\lettera}{\mathtt{a}}
\newcommand{\letterb}{\mathtt{b}}
\newcommand{\letterc}{\mathtt{c}}
\newcommand{\letterx}{\mathtt{x}}
\newcommand{\tapeword}[1]{\mathrm{#1}}
\newcommand{\tapewordw}{\bar{w}}
\newcommand{\tapewordu}{\bar{u}}
\newcommand{\tapewordv}{\bar{v}}
\newcommand{\runT}{\mathcal{T}}
\newcommand{\confC}{\mathcal{C}}
\newcommand{\confD}{\mathcal{D}}
\newcommand{\phiconfM}{\varphi_{\atmM}^{\textrm{conf}}}
\newcommand{\phitransM}{\varphi_{\atmM}^{\textrm{tran}}}


\renewcommand{\bx}{\ensuremath{\mathbf{x}}} 
\renewcommand{\by}{\ensuremath{\mathbf{y}}} 
\renewcommand{\bu}{\ensuremath{\mathbf{u}}} 
\renewcommand{\bz}{\ensuremath{\mathbf{z}}} 
\renewcommand{\bv}{\ensuremath{\mathbf{v}}} 

\newcommand{\genright}{\ensuremath{\overrightarrow{\mathrm{allSeq}}}}
\newcommand{\genbi}{\ensuremath{\overleftrightarrow{\mathrm{allSeq}}}}

\section{The Guarded Subfragment}\label{sec:guarded}


We next shift our attention to the \emph{guarded adjacent fragment}, denoted~$\GA$, defined as the intersection of $\AF$ with the guarded fragment, $\GF$~\cite[Sec.~4.1]{ABN98}. 
In $\GF$, quantification is relativized by atoms, i.e.~all quantification takes the form $\forall{\bar{x}} (\alpha \rightarrow \psi)$ and $\exists{\bar{x}} (\alpha \land \psi)$, where $\alpha$ (a \emph{guard}) is an atom featuring all the variables in $\bar{x}$ and all the free variables~of~$\psi$. 
The satisfiability problem for $\GF$ is $\TwoExpTime$-complete~\cite[Thm.~4.4]{Gradel99}. 
We show that the satisfiability problem for $\GA$ remains $\TwoExpTime$-complete, and thus is as hard as for full $\GF$.
This contrasts with the $\ExpTime$-completeness of the $k$-variable guarded fragment $\GF^k$~\cite[Cor.~4.6]{Gradel99} and of the guarded forward fragment~\cite[Thm.~4]{Bednarczyk21}.
Our proof follows the same reduction strategy as the $\TwoExpTime$-hardness proof for $\GF$ by E.~Gr\"adel~\cite[Thm.~4.4]{Gradel99}.
Because we are working in the guarded \emph{adjacent} fragment, however, Gr\"adel's reduction is not directly available, and some new techniques are required. 
To aid readability, we employ the variable names $w, x, y, z, \dots$ in adjacent formulas, rather than the official $x_1, x_2, x_3, x_4, \dots$. 
The reader may easily check that all formulas in question are indeed in $\AF$ modulo variable renaming.


\subsection{Generating Words.}\label{subsec:generating-words-for-GAF}
We start with a combinatorial exercise concerning the generation of words by the recurrent application of certain adjacent functions. 
Let $m \in \N$ and consider the adjacent functions $\lambda_1, \lambda_2, \lambda_3 \colon [1,m+2] \rightarrow [1,m+1]$ defined by the following
courses of values $(\lambda_i(1) \ \ \cdots \ \ \lambda_i(m+2))$:
\begin{align*}
& \lambda_1\colon \qquad
\begin{pmatrix}
	1 & 2 & 2 & 3 & 4 & \dots & m{+}1
\end{pmatrix}\\
& \lambda_2\colon  \qquad
\begin{pmatrix}
	1 & 2 & 1 & 2 & 3 & \dots & \quad \ m
\end{pmatrix}\\
& \lambda_3\colon  \qquad
\begin{pmatrix}
	1 & 2 & 3 & 3 & 4 & \dots & m{+}1
\end{pmatrix}.
\end{align*}
Using the imagery of Sec.~\ref{sec:primGen}, we can think of these functions as describing strolls on a word of length $(m{+}1)$, in each case starting at the 
left-most position, and proceeding generally rightwards: $\lambda_1$ pauses at the second time step before continuing; 
$\lambda_2$ returns to the beginning after the second time step, but then resumes its rightward journey (without quite reaching the end); and $\lambda_3$ pauses at the third time step before continuing.
%
%

We show that repeated application of these functions to the bit-string $\letterzero \letterone \letterone^m$ yields the whole of the ({exponentially large}) language $\letterzero \letterone \{ \letterzero, \letterone \}^m$.
\begin{lemma} \label{lemma_GA_generation_1}
    Let $W_0 \subseteq \{ \letterzero, \letterone \}^*$ contain $\letterzero \letterone \letterone^m$ and $W_i \coloneqq W_{i{-}1}$
    $\cup \{ \bar{w}^{\lambda_1}, \bar{w}^{\lambda_2}, \bar{w}^{\lambda_3} \mid \bar{w} \in W_{i{-}1} \}$.
    Setting $W \coloneqq \textstyle \bigcup_{i\geq0} W_i$, we have $\letterzero \letterone \{ \letterzero, \letterone \}^m \subseteq W$.
\end{lemma}   
\begin{proof}
    We establish by induction on $i \in [0,m]$ that, for any word $\bar{c} \in \{ \letterzero, \letterone \}^i$, 
    the word $\letterzero \letterone \bar{c} \letterone^{m{-}i}$ is in $W$; the case $i = m$ then yields the statement of the lemma.
    The base case, $i = 0$ follows from the assumption that $W_0$ contains $\letterzero\letterone\letterone^m$. Suppose now $i > 0$. 
    We show that, for any $x\in \{ \letterzero, \letterone \}$, the word $\bar{u} = \letterzero \letterone x \bar{c} \letterone^{m{-}i{-}1}$ is in $W$.
    Letting $c_1$ be the first character of $\bar{c}$, and writing $\bar{c}= c_1\bar{d}$, it follows by 
    inductive hypothesis that the words $\bar{v} = \letterzero \letterone \bar{c}\letterone^{m{-}i{-}1} \letterone$ and
    $\bar{w} = \letterzero \letterone \bar{d} \letterone^{m{-}i{-}1} \letterone \letterone$ are in $W$.
    Consider cases: 
    (i) if $x = \letterone$ then $\bar{u} = \bar{v}^{\lambda_1}$, 
    (ii) if both $x = \letterzero$ and $c_1 = \letterzero$ then $\bar{u} = \bar{v}^{\lambda_3}$,
    and otherwise, (iii) $x=\letterzero$, $c_1 =  \letterone$ and $\bar{u} = \bar{w}^{\lambda_2}$.
    Thus~$\bar{u} \in W$.
\end{proof} 

We now apply Lemma~\ref{lemma_GA_generation_1} to $\GA$. For $m \geq 0$,
let $G_{m}$ be an $(m{+}2)$-ary predicate. We proceed to write, for any 
binary predicate $p$, a $\GA$-sentence $\genright_m(p)$ ensuring that, if~$p$ is satisfied by a pair of objects, say $ab$, then $G_{m}$ is satisfied by any $(m+2)$-tuple $ab\bar{w}$ for $\bar{w} \in \{ a, b \}^m$.
By Lemma~\ref{lemma_GA_generation_1}, it suffices so take $\genright_m(p)$ to be
\[
    \forall x y \; \Big( p(x y) \to
        G_{m}(x y \underbrace{y \cdots y}_m) \Big) \land
    \bigwedge_{i \in [1,3]} \forall \bu_{m+2} \; 
        \Big( G_{m}(\bu_{m+2}) \to
        G_{m}(\bu_{m+2}^{\lambda_i}) \Big).
\]
Thus, if $\fA \models \genright_m(p)$, and $\fA \models p[ab]$,
we may freely quantify over words of length $(m{+}2)$ over the alphabet $\set{a,b}$ as long as they have the prefix $ab$, since $G_m$ can always be used as a guard. 

We shall require a `mirrored' version of the above device, this time involving a \emph{pair} of 2-element alphabets and a \emph{pair} of words of length $m$ over these alphabets. 
For $m \geq 0$, let $F_m$ be a $(2m + 4)$-ary predicate.  We proceed to write, for any
quaternary predicate $r$, a $\GA$-sentence $\genbi_m(r)$ ensuring that, 
if $r$ is satisfied by a quartet of objects $bacd$, then, for any $m$-tuple $\bar{u}$ over the alphabet $\set{a, b}$, and any $m$-tuple~$\bar{v}$ over the alphabet $\set{c,d}$, the predicate $F_{m}$ is satisfied by $\bar{u}bacd\bar{v}$. 
Let $\lambda_0$ denote the identity function on $[1,m+2]$, and take $\lambda_1, \lambda_2, \lambda_3$ as defined above. 
In addition, for any word $\bar{w}$, we write $\bar{w}^{\mathrm R}$ for the reversal of $\bar{w}$. (Thus, $\bar{w}^{\mathrm R} = \tilde{w}$; but the new notation
is more readable in what follows.) 
By two applications of Lemma~\ref{lemma_GA_generation_1}, it suffices so take $\genbi_m(r)$ to be:
\begin{multline*}
\forall{yxzt} \Big( r(yxzt) \to
F_m(\underbrace{y \dots y}_m yxzt \underbrace{t \dots t}_m) \Big) \land\\ 
\bigwedge_{i, j \in [0,3]} \forall{{\bu}^{\mathrm R}_{m{+}2}\bv_{m{+}2}}
\Big( F_m({\bu}^{\mathrm R}_{m{+}2} \bv_{m{+}2}) \to F_m((\bu_{m{+}2}^{\lambda_i})^{\mathrm R}\, \bv_{m{+}2}^{\lambda_j}) \Big).
\end{multline*}
Again, if $\fA \models \genbi_m(r)$, and $\fA \models r[bacd]$,
we may freely quantify over words of the language $\set{a,b}^mbacd\set{c,d}^m$ in $\GA$, since $F_m$ can always be used as a guard. 
Note that the variables of $\bu_{m{+}2} = u_1 \cdots u_{m+2}$ are quantified above in `reverse order'.
This ensures, after renaming, the adjacency of the formula~$\genbi_m(r)$.

%

\subsection{ATMs.} 

An \emph{alternating Turing machine} (ATM)~\cite{Chandra81} is a tuple\linebreak
$\atmM \coloneqq \langle Q, \Sigma, q_0, \Delta \rangle$,
where $Q$ is a non-empty finite set (the {\em states} of $\atmM$),
$\Sigma$ a non-empty finite {alphabet}, $q_0$ an element of $Q$ (the \emph{initial} state), and $\Delta$
a set \emph{transitions} $\delta$, defined presently.
We imagine $\atmM$ to operate on an 1-way infinite
tape by means of a read-write head as usual, but
we take $Q$ to be partitioned into the sets $Q_\exists$ (\emph{existential states}), $Q_\forall$ (\emph{universal states}) and
$\{ q_a, q_r\}$ (the \emph{accepting} and \emph{rejecting state}, respectively). 
Writing $\Sigma'$ for the alphabet $\Sigma$ augmented with the blank cell symbol `$\blank$',
we define a \textit{transition} $\delta \in \Delta$ to be a relation
\begin{equation*}
    \delta \in Q \times \Sigma' \times Q \times \Sigma' \times \{ -1, 0, 1 \}.
\end{equation*}
The transition $\delta = \langle q, s, q', s', k' \rangle$ is enabled when the machine is in state $q$ and
the read-write head is positioned over a tape square containing the letter $s$. On execution of $\delta$, the current tape square is overwritten with $s'$, the head is moved by $k'$ squares, and 
the current state is updated to $q'$. We denote the set of $\delta \in \Delta$ enabled by state $q$ and letter $s$ by $\Delta(q,s)$.

A \emph{configuration} $\confC$ of $\atmM$ is a triple $\langle q, \omega, h \rangle$,
where $q \in Q$, $\omega$ is an infinite word over $\Sigma'$ and $h$ a non-negative integer.
We read the triple $\confC$ as stating that the machine is in state $q$, the tape contents are given by $\omega$, and the head is situated over the $h$th tape square (counting from 0).
If $\bar{w}_0$ is a finite 
word over $\Sigma$ representing the input to the machine, the
{\em initial configuration} is $\langle q_0, \bar{w}_0\blank^*, 0 \rangle$, where
$\blank^*$ represents an infinite series of blanks.
The \textit{successors} of $\confC$ are defined in the usual way via transitions in $\Delta$.
A \textit{halting} configuration is one which is in state $q_a$ or $q_r$. We assume that halting configurations
have no enabled transitions, and non-halting configurations always have at least one enabled transition.

We shall be interested in the case where every computation  
of $\atmM$ 
(understood as a sequence of enabled transitions starting in the initial configuration)  is of finite length, so that 
there is a function $f$ such that, when $\atmM$ runs on input $w_0$, the read-write head never reaches positions beyond
$f(|w_0|){-}1$.
In that case, we may as well take a configuration to have the form $\langle q_0, \bar{w}, h \rangle$ where $\bar{w}$ is a (finite) word over $\Sigma'$ of length 
$f(|w_0|){-}1$ and $h$ is an integer in the range $[0,f(|w_0|){-}1]$. The notions of acceptance and rejection may then be defined as follows: a halting configuration is \textit{accepting} if it 
is in state $q_a$; an existential configuration is \textit{accepting} if it has an accepting successor;
a universal configuration is \textit{accepting} if all its successors are accepting; a configuration which is not accepting is \textit{rejecting}.
We take $\atmM$ to accept $\bar{w}_0$ if the initial configuration is accepting.

We witness acceptance of an input $\bar{w}_0$ by $\atmM$ using an \textit{acceptance tree} $\mathcal{T}$.
This is a finite tree with vertices labelled by (accepting) configurations. 
The root is labelled by an \emph{initial configuration}; and for any vertex labelled with a particular configuration, its children are 
labelled with the results of executing enabled transitions in that configuration.
Vertices labelled with existential configurations have \textit{at least one} child corresponding to an enabled transition;
those labelled with universal configurations
have a
child corresponding to {\em every} enabled transition; the leaves of the tree are labelled with accepting configurations.

We are interested in the case where the function $f$ bounding the space required by $\atmM$ is of the form $f(n) = 2^{n}$.
Thus, $\atmM$ accesses at most $2^{|\bar{w}_0|}$ tape squares in the course of any computation on input $\bar{w}_0$.
%
We now fix such a machine $\atmM$, and show how, for a given input $\bar{w}_0$, we can manufacture 
a $\GA$-formula $\varphi_{\atmM, \bar{w}_0}$ satisfiable if and only if $\atmM$ accepts $\bar{w}_0$. The computation of 
$\varphi_{\atmM, \bar{w}_0}$ runs in polynomial time (in fact, in logarithmic space) as a function of $|\bar{w}_0|$. Since
there are problems in \complexityclass{ASpace}$(2^n)$ that are complete for $\AExpSpace$,
we can thereby reduce any problem in 
$\AExpSpace$ to the the satisfiability problem for $\GA$. Hence the latter problem is $\AExpSpace$-hard. Using the well-known 
equation $\AExpSpace= \TwoExpTime$, this achieves our goal.


\subsection{Encoding numbers.} \label{sec:ga_enc_numb}
In the sequel, we will consider structures $\str{A}$ interpreting a unary predicate $O$. 
Whenever 
$\fA \models \alpha[ab]$, where $\alpha(xy)$ is the formula $\lnot O(x) \land O(y)$,
we say that $a$ and $b$ \emph{act as zero and unit bits} ($O$ for ``One''), 
and for any word $\bar{u}$ over $\{a,b\}$,
we write $\val{\fA}(\bar{u})$ to denote the integer represented by $\bar{u}$, considered as a bit-string,
with $a$ standing for $\letterzero$ and $b$ for $\letterone$
(most significant bit first).
Notice that there may be other elements, say $c$ and $d$, such that 
$\fA \models \alpha[cd]$, in which case may write $\val{\fA}(\bar{w})$ for the integer represented by any word $\bar{w}$ over $\{c,d\}$.
Clearly, the $\GA$-formula $\textsc{eq}(\bu_\nofM, \bv_\nofM) \coloneqq \textstyle \bigwedge_{i = 1}^{\nofM} O(u_i) \leftrightarrow O(v_i)$ satisfies 
$\fA \models \textsc{eq}[\bar{c}, \bar{d}]$ if and only if $\textstyle\val{\fA}(\bar{c}) = \textstyle\val{\fA}(\bar{d})$.
Other arithmetical properties can be expressed similarly.
In particular, for each $k \in \{ {-}1, {+}1\}$, we may easily write a $\GA$-formula 
$\sheq{\bu_\nofM}{\bv_\nofM{+}k}$ satisfying the following property:
\begin{equation*}
	\fA \models \textsc{eq}[\bar{c}, \bar{d}{+}k] \text{ iff } \val{\fA}(\bar{c}) = \val{\fA}(\bar{d}) + k.
\end{equation*}
The details are routine (see~\cite[p.~15]{bkp-h23}). Observe that the formula $\textsc{eq}(\bar{x},\bar{y})$ is not
atomic! Indeed, all the predicates appearing in this formula are unary; thus, this formula may appear in adjacent formulas whatever the order of quantification
among the variables $\bar{x}$ and $\bar{y}$. Similar remarks~apply~to~$\sheq{\bar{x}}{\bar{y}{+}k}$.


\subsection{Encoding configurations.}
We proceed to 
describe a method of encoding, within a certain class of structures, configurations of an ATM $\atmM = \langle Q, \Sigma, q_0, \Delta \rangle$ 
that never accesses more than $f(n) = 2^n$ tape squares on an input of size $n$. Recall that we agreed to regard configurations in this case as triples $\langle q, \bar{w}, h \rangle$,
where $q \in Q$, $\bar{w}$ is a word over $\Sigma'$ of length $N = 2^n$, and $0 \leq h < N$.  
In doing so, we take the various states $q \in Q$  to be \textit{binary} predicates: a configuration is then represented by
an ordered pair of (distinct) elements $ab$ satisfying any of these predicates. If $\fA \models
q[ab]$, we take it that the represented configuration has state $q$. Since we shall want words over the alphabet $\set{a,b}$ to represent integers,
we shall insist that, in this case, $ab$ satisfies the formula $\alpha$ defined in Sec.~\ref{sec:ga_enc_numb}. This we ensure by writing the
$\GA$-sentence
$\bigwedge_{q \in Q} \forall x y \left(q(xy) \rightarrow \neg O(x) \wedge O(y)\right)$. (Thus, we think of $a$ as a 0 and $b$ as a 1.)
To represent further aspects of the configuration $ab$, we treat each symbol $s \in \Sigma'$ as an $n$-ary predicate, and additionally employ an
$n$-ary predicate $H$.
Specifically,  for any $\bar{w} \in \{ a,b \}^n$, we read
$\fA \models H[\bar{w}]$ as ``the head of the configuration represented by $ab$ is at position $\val{\fA}(\bar{w})$'',
and we read $\fA \models s[\bar{w}]$ as ``the tape square $\val{\fA}(\bar{w})$ of the configuration represented by $ab$ contains the symbol $s$''.
Of course, these  interpretations are only meaningful if:
\begin{enumerate}
    \item there is at most one string in $\{ a, b \}^n$ satisfying $H$ and thus encoding the head position;
    \item each bit-string over $\{ a, b \}^n$ satisfies at most one predicate $s \in \Sigma'$, thus ensuring that a tape cell contains at most one symbol; and
    \item $ab$ satisfies at most one $q \in Q$, thus ensuring that the configuration is in at most one state.
\end{enumerate}
For any state $q \in Q$ we construct a $\GA$-formula $\psi_{q} \coloneqq \psi_{q, 1} \wedge \psi_{q, 2} \wedge \psi_{q, 3}$
enforcing these conditions for any configuration whose state is $q$. 
The first conjunct, $\psi_{q, 1}$, may be given as follows:
\begin{multline*}
    \psi_{q, 1} \coloneqq
    \genright_{2n}(q) \wedge \forall{xy\bu_{\nofM}\bv_{\nofM}}
        \Big( G_{2\nofM}(xy\bu_{\nofM}\bv_{\nofM}) \to \\
            \big(( H(\bu_{\nofM}) \land H(\bv_{\nofM}) ) \to \sheq{\bu_{\nofM}}{\bv_{\nofM}} \big) \Big).
\end{multline*}
Note that this sentence is (adjacent and) guarded, with the atom $G_{2\nofM}(xy\bu_{\nofM}\bv_{\nofM})$ acting as a guard. Of course, thanks to
the conjunct $\genright_{2n}(q)$, this guard is, as it were, \textit{semantically inert},
since, assuming that $q$ is satisfied by $ab$, $G_{2\nofM}$ is satisfied by 
\textit{every}
word of $ab\{ a, b \}^{2n}$.
%
The remaining conjuncts $\psi_{q, 2}$ and $\psi_{q, 3}$ can be easily formulated.

\subsection{Encoding instances of transitions}
For each transition $\delta \in \Delta$ we employ a quaternary predicate $E_\delta$ and read $\fA \models E_\delta[bacd]$ as
``the configuration encoded by $ab$ enables a transition $\delta$ thus producing the configuration encoded by $cd$''.
(Do note the reversal of $ab$ in $E_\delta[bacd]$!)
We call $bacd$ a \textit{$\delta$-transition instance} in $\fA$ with $ab$ corresponding to the predecessor configuration and $cd$ to the successor configuration.
To make sure that each $\delta$-transition instance is indeed a result of transitioning via $\delta = \langle q, s, q', s', k' \rangle$
we write the following~in~$\GA$:
\begin{enumerate}
    \item the successor configuration is in state $q'$;
    \item the head of the successor configuration is moved by $k'$ relative to the predecessor configurations head;
    \item the $h$-th tape cell on the successor configuration is occupied by $s'$, where $h$ is the position of the predecessor's head;
    \item all tape cells that the predecessor's head does not point to are inherited by the successor.
\end{enumerate}
For any transition $\delta \in \Delta$, we construct a $\GA$-formula $\psi_{\delta} \coloneqq \psi_{\delta, 1} \wedge \cdots \wedge \psi_{\delta, 4}$
enforcing these conditions.  We write $\psi_{\delta,4}$ in detail as an example; the preceding three conjuncts are handled similarly or more easily:
\begin{multline*}
    \psi_{\delta,4} \coloneqq \genbi_n(E_\delta) \wedge 
    \forall{\bu_{\nofM} y x z t \bv_{\nofM}} \bigg( F_{\nofM}(\bu_{\nofM} y x z t \bv_{\nofM}) \to  \\
        \Big( \big( E_\delta(yxzt) \wedge \neg H(\bu_{\nofM}) \land \sheq{\bu_\nofM}{\bv_{\nofM}} \big) \to \bigwedge_{s \in \Sigma} \big( s_{\lettera}(\bu_{\nofM}) \leftrightarrow s_{\lettera}(\bv_{\nofM}) \big)  \Big)
    \bigg).
\end{multline*}
Suppose now that $\fA \models \psi_\delta$, and moreover, that
$\fA \models E_\delta[bacd]$ for some $bacd \in A^4$, where
$ab$ and $cd$ both encode configurations.
By $\genbi_n(E_\delta)$, we have that $\fA \models F_n[\bar{u}bacd\bar{v}]$ for all $\bar{u} \in \{ a, b \}^n$ and $\bar{v} \in \{ c, d \}^n$.
By picking any $0 \leq i < 2^n$ that is not the head position of the configuration encoded by $ab$, and $\bar{u} \in \{ a, b \}^n$, $\bar{v} \in \{ c, d \}^n$
with $\val{\fA}(\bar{u}) = \val{\fA}(\bar{v}) = i$, we are guaranteed that, for each symbol $s \in \Sigma$,
$\fA \models s[\bar{u}]$ if and only if $\fA \models s[\bar{v}]$.

\subsection{Encoding acceptance trees}
The last step in our reduction is to write a formula whose models contain configurations arranged as an acceptance tree witnessing the fact that
$\atmM$ accepts some input $w_0 = s_1 \cdots s_n$. Recall that, the root of this tree is labelled with the initial configuration, namely
$\confC = \langle q_0, \bar{w}_0 \blank^{\ell} , 0 \rangle$,
where $\ell = 2^n - |\bar{w}_0|$. We ensure the existence of such a root configuration in a structure with the following $\GA$-sentence:
\begin{multline*}
    \genright_n(q_0) \wedge \exists xy \bigg( q_0(xy) \wedge H(\text{``}0\text{''}) \wedge
    \bigwedge_{i = 0}^{i < |\bar{w}_0|} s_{i{+}1}(\text{``}i\text{''}) \wedge \\
    \forall \bar{u} \Big(G(xy\bar{u}) \to
        \big( (\bigwedge_{i = 0}^{i < |\bar{w}_0|} \neg \sheq{\text{``}i\text{''}}{ \bar{u}} ) \to \blank(\bar{u}) \big) \Big)
    \bigg),
\end{multline*}
where $\text{``}i\text{''}$ is the binary encoding of $i$ using $x$ as a zero bit and $y$ as a unit~bit.
We  ensure the existence of successor configurations required by existential configurations as follows. Suppose that 
$q$ is an existential state, and $s$ a symbol. The following sentence ensures that any configuration in state 
$q  \in Q_\exists$ with the head reading symbol  $s \in \Sigma'$ has a child in the acceptance tree:
%
%
\begin{multline*}
    \genright_n(q) \wedge \forall{\tilde{\bu}_{\nofM} yx}
    \Big( G_{\nofM}(xy\bu_{\nofM}) \to \\
     \big((q(xy) \land H(\bu_{\nofM}) \land s(\bu_{\nofM})) \to \bigvee_{\delta \in \Delta(q,s)} \exists zt\, E_\delta(yxzt)  \big) \Big).
\end{multline*}
In case $q$ is universal, the disjunction over $\Delta(q,s)$ is replaced by a conjunction.
Lastly, we write the sentence $\neg \exists xy\,  q_r(xy)$ to ensure
that a rejecting configuration is never encoded by any pair of elements in any structure.

For any input string $w_0$, let $\varphi_{\atmM, \bar{w}_0}$ be the conjunction of all the $\GA$-sentences given above. We remark that 
$\varphi_{\atmM, \bar{w}_0}$ does not feature the equality predicate.
Bearing in mind that $\atmM$ is guaranteed to terminate on $\bar{w}_0$ in a finite number of steps, accessing no more than $2^{|w_0|}$ tape squares, 
we see that any model of $\varphi_{\atmM, \bar{w}_0}$ embeds an acceptance tree for $\atmM$ on input $\bar{w}_0$.
Conversely, any acceptance trees for $\atmM$ on input $\bar{w}_0$ can be expanded to a model of $\varphi_{\atmM, \bar{w}_0}$ by
interpreting the relevant predicates as suggested above.
We conclude:
\begin{theorem}
    The finite and general satisfiability problem for the \textup{(}equality-free\textup{)} guarded adjacent fragment of first-order logic is $\TwoExpTime$-hard. 
\label{theo:hard}
\end{theorem}

\section{Extending the adjacent fragment}\label{sec:extensions-and-future-work}
In this final section, we consider the prospects for extending the logic $\AF$ within the family of argument-sequence fragments. As a preliminary, we compare the expressive power of \AF{} to that of the two-variable fragment of first-order logic.

Denote by $\FOt$ the two-variable fragment of first-order logic, namely, the set of first-order formulas over a purely relational signature (i.e.~no individual constants or function symbols) in which the only logical variables occurring are $x_1$ and $x_2$. We allow the equality predicate in $\FOt$.
%
\begin{theorem}\label{theo:fo2-and-af-over-binary-sig-are-the-same}
	Every $\FO^2$-formula is logically equivalent to an~$\AF$-formula. Conversely,
	every $\AF$-sentence featuring predicates of arity at most two is logically equivalent to an $\FO^2$-sentence.
\end{theorem}
\begin{proof}
	Let $Q$ stand for either of the quantifiers $\forall$ or $\exists$.
	For the first statement of the theorem, 	
	suppose that $\phi$ is an $\FOt${}-formula. Let $k$ be the highest index of any free variable of 
	$\phi$, and $k=0$ if $\phi$ is a sentence. (Thus, $0 \leq k \leq 2$.) We claim that $\phi$ is logically equivalent to an $\AFv{[k]}$-formula, proceeding by structural induction on $\phi$.
	For the base case, $\phi$ is an atom of $\FOt$, 
	and therefore certainly in $\AFv{[k]}$. The case of Boolean operators is routine (bearing in mind that $\AFv{[0]} \subseteq \AFv{[1]} \subseteq \AFv{[2]}$). Now consider the case where $\phi$ has the form $Q x_1 \, \psi$. If $x_1$ does not occur free in $\psi$, then $\phi$ and $\psi$ are logically equivalent, and the result follows immediately by inductive hypothesis. If $x_1$ does occur free in $\psi$, but $x_2$ does not, then, by inductive hypothesis,
	let $\psi' \in \AFv{[1]}$ be logically equivalent to $\psi$. Then $\phi' = Q x_1\, \psi' \in \AFv{[0]}$ is logically equivalent to $\phi$ as
	required. 
	Otherwise, both $x_1$ and $x_2$ occur free in $\psi$: let $\phi^*$ be the result of transposing the variables $x_1$ and $x_2$ everywhere in $\phi$, and similarly for $\psi$. Thus 
	$\phi$ is satisfied in a structure $\fA$ under an assignment in which $x_2 \leftarrow a$ if and only if  
	$\phi^*$ is satisfied in $\fA$ under an assignment in which $x_1 \leftarrow a$. Moreover,  	
	$\phi^* = Q x_2 \, \psi^*$, with $x_2$ free in $\psi^*$. By inductive hypothesis, $\psi^*$ is equivalent to a formula $\psi' \in \AFv{[2]}$,
	and hence $\phi^*$ is equivalent to $Q x_2 \psi' \in \AFv{[1]}$.
	Now let $\phi' \in \AFv{[2]}$  be the result of incrementing the indices of all variables (free or bound) in $Q x_2 \psi'$. Thus
	$\phi'$ is satisfied in $\fA$ under an assignment in which $x_2 \leftarrow a$ if and only if 
	$\phi^*$ is satisfied in $\fA$ under any assignment in which $x_1 \leftarrow a$. That is to say, $\phi'$ is logically equivalent to $\phi$ as required.
	The remaining case is where $\phi$ is of the form $Q x_2 \psi$. If $x_2$ does not occur free in $\psi$, then $\phi$ and $\psi$ are logically equivalent, and the result follows immediately by inductive hypothesis. 
	If $x_1$ does not occur free in $\phi$ (i.e.~$\phi$ is a sentence), then $\phi$ is logically equivalent to $\phi^* = Q x_1 \psi^*$, and
    by the previous case, $\phi^*$ is equivalent to a formula of $\AFv{[0]}$. 
	Otherwise, if both $x_1$ and $x_2$ occur free in $\psi$, by inductive hypothesis, $\psi$ is logically equivalent to a formula $\psi' \in \AFv{[2]}$, whence 
	$\phi$ is logically equivalent to $Q x_2\, \psi' \in \AFv{[1]}$. Since $x_1$ is free in $\phi$, this is what is required, completing the induction.
	
	Turning to the second statement of the theorem,
	denote by $\FOts$ the fragment of first-order logic consisting of those formulas with the property that no sub-formula contains
	more than two free variables. A simple variable re-naming procedure shows that any sentence of $\FOts$ is logically equivalent to a sentence of~$\FOt$.
	We prove the following claim: if $\phi$ is a formula of $\AF^{k}$ featuring predicates of arity at most 2, then there exists a formula $\phi'$ such that: (i) $\phi$ and $\phi'$ have the same free variables and are logically equivalent; and (ii) $\phi'$ is a Boolean combination of atomic \AFv{[k]}-formulas
	and $\FOts$-formulas having at most one free variable. 
	Since every sentence of $\AF$ is (by adding vacuous quantification if necessary) an 
	$\AFv{[0]}$-formula, and since any $\FOts$-sentence is logically equivalent to an $\FOt$-sentence, putting $k=0$ yields the desired result.
	Observe that, since all predicates have arity at most 2, atomic formulas can contain at most two variables; and since $\phi$ is in $\AF$, any two distinct free variables occurring in an atomic formula must have indices differing by exactly one. To avoid special cases in the following proof, we assume that 
	$x_0$ is a variable, though this variable will never actually appear in any formulas. 
	We prove the claim by structural induction on $\phi$. If $\phi$ is atomic, the claim is immediate. Likewise, if  $\phi$ is
	$\neg \psi$ or $\psi_1 \circ \psi_2$, where $\circ$ is a Boolean connective, it follows immediately by inductive hypothesis. 
	
	Suppose then $\phi$ has the form $\forall x_{k'} \psi$. By the formation rules for $\AF$, we have $\psi \in \AFv{[k']}$. Let $\psi'$ be the formula
	guaranteed by the inductive hypothesis applied to $\psi$. Using standard Boolean identities, we may re-write the Boolean combination 
	$\psi'$ in conjunctive normal form, say as $\chi:= \chi_1 \wedge \cdots \wedge \chi_\ell$, where each $\chi_i$ is a disjunction, say $\theta_{i,1} \vee \cdots \vee \theta_{i,m_i}$, and each $\theta_{i,j}$ is either a literal in $\AFv{[k']}$ or an 
	$\FOts$-formula having at most one free variable. Thus, $\phi$ is logically equivalent to the conjunction $\forall x_{k'}\, \chi_1 \wedge \cdots \wedge \forall x_{k'}\, \chi_\ell$. 
	Fixing one of these conjuncts, say $\forall x_{k'}\, \chi_i$, consider the various disjuncts of $\chi_{i}$. 
	Let $\gamma_i$ be the disjunction of those disjuncts $\theta_{i,j}$ in which $x_{k'}$ does not occur free, and let let $\delta_i$ be the disjunction of the rest. Thus, $\forall x_{k'} \chi_i$ is logically equivalent to $\gamma_i \vee \forall x_{k'} \delta_i$. 	
	Clearly, $\gamma_i$ is a Boolean combination of atomic \AFv{[k']}-formulas and $\FOts$-formulas with at most one free variable. 
	Moreover, the free variables of $\gamma_i$ are confined to $\bx_k$, since the free variables of $\phi$ are are confined to $\bx_k$ as well.
	On the other hand, if $\theta_{i,j}$ occurs in $\delta_i$, then either its free variables
	are included in $\set{x_{k'}}$, or it is a literal with free variables exactly $\set{x_{k'{-}1}, x_{k'}}$. Hence, $\forall x_{k'} \delta_i$
	is an $\FOts$-formula with free variables included in $\set{x_{k'{-}1}}$. Moreover, if $x_{k'{-}1}$ actually occurs
	free in $\theta_{i,j}$, we must have $k'{-}1 \leq k$, since the free variables of $\phi$ are confined to $\bx_k$. Thus, $\chi$ and hence $\phi$, is logically equivalent to a Boolean combination
	of formulas of the required forms. 
	
	Suppose finally that $\phi$ has the form $\exists x_{k'} \psi$. Then we proceed as in the previous case, but use disjunctive normal form rather than conjunctive normal form.
	
	This completes the induction, and hence the proof. 
\end{proof}

Now let us return to the question of extending the adjacent fragment.
The adjacent fragment is defined by restricting the permitted argument sequences appearing in atoms. More precisely, within the régime of index-normal formulas,
we insist that, in contexts where the variable $x_k$ is available for quantification, all atoms have the form $p(\bx^f_k)$, where $p$ is an
$m$-ary predicate, and $f\colon [1,m] \rightarrow [1,k]$ is an adjacent function. 
It is natural to ask whether the adjacency restriction might be relaxed without 
compromising the decidability of satisfiability. As we now show, the answer must be no. 

Let $f\colon [1, m] \to [1, k]$ be a non-adjacent function, with $m, k \geq 2$: that is,  
there is some $j$ ($1 \leq j < m$) for which $|f(j+1) - f(j)| \geq 2$.
Denote by $\AF^f$ the extension of~$\AF$ obtained by allowing atoms of the form $p(\bx_k^f)$ 
in clause (1) of the definition of $\AF^{[k]}$ on page~\pageref{page:Def}. We show that $\AF^f$
is expressive enough to state that a given binary relation is transitive.
\begin{lemma}
	Let $f \colon [1, m] \to [1, k]$ be a non-adjacent function, $T$ a binary predicate, and $Q$ an $m$-ary predicate. There exists a formula $\phi_T$ of $\AF^f$, such that: 
	\textup{(i)} if  $\fA \models \phi_T$ then $T^{\fA}$ is transitive, and \textup{(ii)} any structure~$\fA$ interpreting $T$ as a transitive relation, but
	not interpreting $Q$, can be expanded to a model of $\phi_T$.
\label{lma:transitivity} 
\end{lemma}
\begin{proof}
	Since $f$ is non-adjacent, fix an index $j \in [1, m{-}1]$ for which $|f(j{+}1) - f(j)| \geq 2$. 
	We may assume without loss of generality that $f(j{+}1) > f(j)$, as the proof for the other case is obtained by swapping all occurrences of $j$ and $j{+}1$.
	Notice that, in this case, we must have $f(j) <f(j+1){-}1$ and hence $f(j) < k{-}1$.
	Define~$\phi_T \coloneqq \forall{\bx_m}\, \phi_T^1 \land \forall{\bx_k}\, \varphi_T^2$,~where
	\begin{align*}
	& \phi_T^1 \coloneqq \left(Q(\bx_{m}) \rightarrow T(x_j,x_{j{+}1})\right)\\ 	
	\begin{split}
	& \phi_T^2 \coloneqq \Big( T(x_{f(j)}, x_{f(j){+}1}) \land T(x_{f(j){+}1}, x_{f(j){+}2}) \land\\
	&\hspace{7cm} \bigwedge_{i = f(j)+2}^{f(j{+}1){-}1} x_i = x_{i{+}1} \Big) \to Q(\bx_k^f).
	\end{split}
	\end{align*}
	To establish (i), we take any $\fA \models \varphi_T$ and any elements $a, b, c$~of~$\fA$ with $(a,b) \in T^{\fA}$ and $(b,c) \in T^{\fA}$.
	We must show that $(a,c) \in T^{\fA}$.
	Let $\bar{d}= d_1 \cdots d_k$ be a tuple with $d_{f(j)}= a$, $d_{f(j){+}1}= b$, and $d_i = c$ for all $i$ except $f(j)$ and $f(j){+}1$. Thus, 
    $\fA \models T[d_{f(j)},d_{f(j){+}1}]$,
	$\fA \models T[d_{f(j){+}1},d_{f(j){+}2}]$, and $d_{f(j){+}2} = \cdots = d_{f(j{+}1)}$. Since
	$\fA \models \phi^2_T[\bar{d}]$, we have $\fA \models Q[\bar{d}^f]$. But the $j$th position of $\bar{d}^f$ is occupied by $d_{f(j)} = a$, and
	the $(j{+}1)$th position is occupied by $d_{f(j+1)} = c$; and since $\fA \models \phi^1_T[\bar{d}^f]$, we have $(a,c) \in T^{\fA}$, as required.

	To establish (ii), suppose $\fA$ interprets $T$ as a transitive relation. We expand $\fA$ to $\fA^+$ by fixing the interpretation of $Q$
	to be the set of all $m$-tuples $\bar{a}$ such that $(a_j, a_{j{+}1}) \in T^{\fA}$. It is immediate that $\fA^+ \models \forall{\bx_m}\, \phi_T^1$.
	Now take any $k$-tuple~$\bar{c}$ satisfying the antecedent of~$\phi_T^2$ in $\fA^+$. Thus, $(c_{f(j)},c_{f(j){+}1})$ and $(c_{f(j){+}1},c_{f(j){+}2})$
	are both in the relation $T^{\fA}$, and, moreover, $c_{f(j){+}2} = \cdots = c_{f(j{+}1)}$. By transitivity, $(c_{f(j)},c_{f(j){+}2}) \in T^{\fA}$,
	whence $(c_{f(j)},c_{f(j{+}1)}) \in T^{\fA}$. But $c_{f(j)}$ and $c_{f(j{+}1)}$ are, respectively, the $j$th and $(j{+}1)$th element of $\bar{c}^f$, 
	and so, by construction, $\fA^+ \models Q[\bar{c}^f]$. Thus $\fA^+ \models \forall{\bx_k}\, \varphi_T^2$.
\end{proof}

\begin{theorem}
Let $f\colon [1, k] \to [1, n]$ be a non-adjacent function. Then the satisfiability and finite satisfiability problems for $\AF^f$ are undecidable.
\end{theorem}
\begin{proof}
We reduce from the (finite) satisfiability problem for $\FOt$ with two transitive relations, known to be undecidable~\cite[Thm.~3]{Kieronski05}.
Let an $\FOt$-sentence~$\phi$ be given (possibly featuring the two binary predicates $T$ and $T'$ (required to be interepreted as transitive relations), and let 
$\phi^*$ to be its logically equivalent $\AF$-formula, guaranteed by Theorem~\ref{theo:fo2-and-af-over-binary-sig-are-the-same}.
Define $\phi_T$ and $\phi_{T'}$ to be the $\AF^f$-formulas guaranteed by Lemma~\ref{lma:transitivity}, stating that $T$ and $T'$ are transitive.
(The respective predicates $Q$ appearing in these formulas are chosen afresh.) 
Applying Lemma~\ref{lma:transitivity} we see that $\phi$ has a (finite) model interpreting $T$ and $T'$ as transitive relations if and only if the $\AF^f$-formula $\phi^* \wedge \phi_{T} \wedge \phi_{T'}$ has
a (finite)~model.
\end{proof}

\bibliographystyle{asl}
\bibliography{references}

\end{document}